\documentclass[12pt]{article}
\usepackage{tikz}
\usetikzlibrary{spy}
\usepackage{amsmath,amssymb,amsthm}
\usepackage{bbm}  
\usepackage{fancyhdr}
\usepackage[utf8]{inputenc}    
\usepackage{fullpage}
\usepackage{float}
\usepackage{authblk}
\usepackage{color}
\usepackage{multicol,lipsum}
\usepackage{amsthm} 
\usepackage{cite} 
\newcommand{\be}{\begin{equation}}
\newcommand{\ee}{\end{equation}}
\newcommand{\bea}{\begin{eqnarray}}
\newcommand{\eea}{\end{eqnarray}}

\newcommand*{\equal}{=}

\let\OLDthebibliography\thebibliography
\renewcommand\thebibliography[1]{
  \OLDthebibliography{#1}
  \setlength{\parskip}{0pt}
  \setlength{\itemsep}{0pt plus 0.3ex}
}

\newcommand{\sign}{\,\text{sgn}\,}

\newcommand{\argth}{\,\text{argth}\,}

\newtheorem{theorem}{Theorem}

\newtheorem{prop}{Proposition}
\usepackage{tikz}

\usepackage{pgfplots}
\usetikzlibrary{spy}
\usepackage{ulem}
\usetikzlibrary{patterns}
\usepackage{slashed}
\usetikzlibrary{calc,arrows,positioning}
\usetikzlibrary{arrows,shadows}
\usetikzlibrary{decorations.pathmorphing}	
\usetikzlibrary{decorations.markings}
\author[1]{Etienne Granet\textsuperscript{$\star$,}}
\author[2,3,4]{ Jesper Lykke Jacobsen\textsuperscript{$\dag$,}}
\author[2,5]{Hubert Saleur\textsuperscript{$\ddag$,}}
\affil[1]{The Rudolf Peierls Centre for Theoretical Physics, Oxford University, Oxford OX1 3PU, UK  }
\affil[2]{Institut de Physique Th\'eorique, Paris Saclay, CEA, CNRS, 91191 Gif-sur-Yvette, France}
\affil[3]{Laboratoire de Physique de l'\'Ecole Normale Sup\'erieure, ENS, Universit\'e PSL, CNRS, Sorbonne Universit\'e, Universit\'e de Paris, F-75005 Paris, France}
\affil[4]{Sorbonne Universit\'e, \'Ecole Normale Sup\'erieure, CNRS, Laboratoire de Physique (LPENS), 75005 Paris, France}
\affil[5]{USC Physics Department, Los Angeles CA 90089, USA}

\title{\textsc{\large Analytic continuation of Bethe energies \\and application to the thermodynamic limit \\of the $SL(2,\mathbb{C})$ non-compact spin chains}}
\date{${}^\star$ {\small \sf etienne.granet@physics.ox.ac.uk} \quad
${}^\dag$ {\small \sf jesper.jacobsen@ens.fr} \quad
${}^\ddag$ {\small \sf hubert.saleur@cea.fr}}

\begin{document}

\maketitle
\begin{abstract}

We consider the problem of analytically continuing energies computed with the Bethe ansatz, as posed by the study of non-compact integrable spin chains. By introducing an imaginary extensive twist in the Bethe equations, we show that one can expand the analytic continuation of energies in the scaling limit around another 'pseudo-vacuum' sitting at a negative number of Bethe roots, in the same way as around the usual pseudo-vacuum. We show that this method can be used to compute the energy levels of some states of the $SL(2,\mathbb{C})$ integrable spin chain in the infinite-volume limit, and as a proof of principle recover the  ground-state value previously obtained  in \cite{derkachovkorchemsky2} (for the case of spins $s=0,
\bar{s}=-1$) by extrapolating  results in small sizes. These results  represent, as far as we know,  the first (partial) description of the spectrum of $SL(2,\mathbb{C})$ non-compact spin chains in the thermodynamic limit.
\end{abstract}
\tableofcontents
\section{Introduction}

The $SL(2,\mathbb{C})$ non-compact Heisenberg spin chains arose originally in high-energy physics as  model Hamiltonians for   interacting quantum particles in a two-dimensional plane \cite{lipatov} (QCD in the Regge limit). It was quickly realized \cite{lipatov,fadeevkorchemsky} that these spin chains are   integrable analogs of the well-known $su(2)$ and spin-$1/2$ Heisenberg spin chains, where an  \textit{infinite-dimensional} (or \textit{non-compact}) irreducible representation of $SL(2,\mathbb{C})$ sits at each site instead of a finite-dimensional one.

While the arsenal of the Quantum Inverse Scattering Method \cite{bethe,faddeev1,faddeev2,sklyanin1,baxter} is in principle applicable to study these chains, tremendous difficulties  are encountered in practice \cite{derkachovkorchemsky1,korchmeskybepom,derkachovkorchemsky2,derkachovkorchemsky3}. 
Explicit expressions for the eigenstates and energy levels exist only for two sites, and even the determination of the ground-state energies for larger but still small sizes is notably complicated \cite{fadeevkorchemsky,derkachovkorchemsky2}. The infinite dimension of the Hilbert space is clearly a serious obstacle, since one cannot diagonalize numerically the Hamiltonian in small sizes. This is in contrast with  finite-dimensional (compact)  integrable spin chains where the  Bethe ansatz not only allows one to efficiently track and compute energy levels from small to very large system sizes, but also provides a description of the energy levels in the infinite-size limit with Bethe-root densities and Thermodynamic Bethe Ansatz (TBA). Consequently, although the continuum limit of compact integrable $su(2)$ spin chains is well-understood in terms of Wess-Zumino-Witten (WZW) models \cite{affleckwzw}, close to nothing is known on the continuum limit of their non-compact $SL(2,\mathbb{C})$ cousins.

The objective of this paper is to provide a (partial) description of the energy levels of the non-compact $SL(2,\mathbb{C})$ spin chain in the thermodynamic limit. Our approach relies on the analytic continuation of energies computed with the Bethe ansatz in the thermodynamic limit, to a {\sl negative number of Bethe roots}, that is performed by the introduction of an imaginary extensive twist. In the same way that energies can be expanded in convergent series around the pseudo-vacuum defined by the absence of Bethe roots \cite{granetjacobsensaleurxxz}, the analytic continuation of energies is found to be expandable as well around another 'pseudo-vacuum' sitting at an extensive negative number of Bethe roots. We explain that it permits to obtain convergent series for the energy levels of a certain (but large) class of states in the $SL(2,\mathbb{C})$ spin chain. As a proof of principle, we recover the value of the ground state previously deduced from finite-size extrapolation \cite{derkachovkorchemsky2} in the case (see below) $s=0,\bar{s}=-1$, this state being identified here as being of minimal energy with respect to particle-hole excitations.

We note that non-compact spin chains, although relevant in and originating from the context of high-energy physics, Yang-Mills theories and AdS-CFT correspondence \cite{reviewadscft,kotikov,integrabilityqcd,alfimov,gromovfishnet}, also play a role in quantum and statistical physics. Some finite-dimensional statistical mechanics models---such as the alternating six-vertex model, the antiferromagnetic Potts model, or certain loop models---are described by non-compact field theories \cite{esslerfrahmsaleur,esslerfrahmsaleur2,ikhlefnoncompact,ikhlef2,alternating,vernier1,vernier2,vernier3,robertson}, while other models are genuine infinite-dimensional spin chains or lattice models, such as the quantum Toda chain \cite{toda,gaudin}, the Chalker-Coddington model \cite{chalker} or stochastic particle processes \cite{frassek}.
Some of the models in the latter class can in turn be investigated approximately through a series of finite-dimensional truncations \cite{ikhleffendleycardy,couvreurvernierjacobsensaleur}.

The paper is organized as follows. In Section \ref{S1}, we recall some properties of compact integrable spin chains with $su(2)$ symmetry, and present the $SL(2,\mathbb{R})$ and $SL(2,\mathbb{C})$ spin chains as well as a review of their known properties relevant to our discussion. In Section \ref{analy} we study the Bethe equations for the $s=-1$ Heisenberg spin chain when an imaginary extensive twist  $\varphi$ is included, and show that one can write a large class of energy levels as convergent series in $e^{-2\varphi}$ for $\varphi\to\infty$. In Section \ref{specialroot} we exhibit a special state in the spectrum whose energy (as well as its derivatives) can be exceptionally continued analytically. In Section \ref{explore} we explain that this special  state plays the role of another `pseudo-vacuum', i.e., that we can obtain from it series expansions for other eigenenergies in the spectrum.

\section{A reminder on spin chains with $su(2)$ symmetry \label{S1}}
\subsection{Compact $su(2)$ spin chains}
We start with some reminders on spin chains with $su(2)$ symmetry that are `compact', i.e., whose on-site Hilbert space is finite-dimensional.

We consider a Hamiltonian $\mathcal{H}_L$ for $L$ particles that acts on a tensor product $V^{\otimes L}$ of $L$ copies of a vector space $V$. We recall that $\mathcal{H}_L$ is integrable if it is built from an $R$-matrix that satisfies the Yang-Baxter equation \cite{faddeev1,faddeev2,sklyanin1,baxter}
\begin{equation}
\label{yb}
R_{12}(\lambda-\mu)R_{13}(\lambda)R_{23}(\mu)=R_{23}(\mu)R_{13}(\lambda)R_{12}(\lambda-\mu) \,,
\end{equation}
and we refer to \cite{slavnov} for the details of this construction. $su(2)$-invariant solutions to \eqref{yb} are known for the situation where each $V$ is an irreducible representation of $su(2)$ \cite{kulishyu,tarasovtakhtajanfaddeev}. These representations are necessarily of spin $s$, with $s$ being integer or half-integer, and thus  of finite dimension $2s+1$. The solutions read explicitly \cite{kulishyu,tarasovtakhtajanfaddeev}
\begin{equation}
\label{rmat}
R(\lambda)=\sum_{l=0}^{2s}\prod_{k=l+1}^{2s}\frac{\lambda-ik}{\lambda+ik}\prod_{j=0,\neq l}^{2s}\frac{\pmb{\sigma}-x_j}{x_l-x_j} \,,
\end{equation}
with $x_j=j(j+1)/2-s(s+1)$ and $\pmb{\sigma}=\sum_{\alpha=x,y,z}S^\alpha\otimes S^\alpha$, where $S^x,S^y,S^z$ act in the spin-$s$ irreducible representation (irrep). The Hamiltonian reads then
\begin{equation}
\mathcal{H}_L=\sum_{i=1}^L f(\pmb{\sigma}_{i,i+1}) \,,
\end{equation}
with $\pmb{\sigma}_{i,i+1}=\sum_{\alpha=x,y,z}S^\alpha_i\otimes S^\alpha_{i+1}$, $S^\alpha_i$ being a copy of $S^\alpha$ at site $i$ with periodic boundary conditions (we identify the sites $L+1\equiv 1$), and \footnote{We put a minus sign compared to \cite{tarasovtakhtajanfaddeev} in order for the ground state to be anti-ferromagnetic.}
\begin{equation}
f(x)=-2\sum_{l=0}^{2s}\sum_{k=l+1}^{2s}\frac{1}{k}\prod_{j=0,\neq l}^{2s}\frac{x-x_j}{x_l-x_j} \,.
\end{equation}
For example, the case $s=1/2$ of these formulae gives the well-known spin-$1/2$ Heisenberg XXX spin chain \cite{heisenberg} with Hamiltonian
\begin{equation}
\mathcal{H}_L=2\sum_{i=1}^L \left(S_i^xS_{i+1}^x+S_i^yS_{i+1}^y+S_i^zS_{i+1}^z-\tfrac{1}{4}\right) \,.
\end{equation}
These spin chains are all solvable by the algebraic Bethe ansatz (ABA) \cite{bethe,tarasovtakhtajanfaddeev}. Their energy levels read
\begin{equation}
\label{efin}
E=-\sum_{k=1}^N \frac{2s}{\lambda_k^2+s^2} \,,
\end{equation}
where $\lambda_1,\ldots,\lambda_N$ is an (admissible \cite{bethe,granetjacobsen}) solution to the Bethe equations
\begin{equation}
\label{befin}
\left(\frac{\lambda_k+is}{\lambda_k-is}\right)^L=\prod_{l=1,\neq k}^N\frac{\lambda_k-\lambda_l+i}{\lambda_k-\lambda_l-i} \,.
\end{equation}
Since the Hamiltonian commutes with the generators of $su(2)$, the eigenspaces can be decomposed into spin-$u$ irreps of $su(2)$, with $u$ a positive integer or half-integer. More precisely, $u$ is the value of $\sum_{i=1}^L S^z_i$ on the highest-weight state of this irrep and is related to $N$ through $u=sL-N$, where $N$ is the number of Bethe roots.

\subsection{Non-compact $SL(2,\mathbb{R})$ spin chains}
Since all the irreps of $su(2)$ are finite-dimensional, one has to consider representations of more general groups in order to obtain `non-compact' spin chains. The Lie group $SL(2,\mathbb{R})$ whose Lie algebra is $su(2)$ provides the simplest examples of infinite-dimensional irreps.

Among these are the continuous series representations, labelled by a real spin $s \in \mathbb{R}$, and the discrete series representations, labelled by a spin $s=0,-1/2,-1,-3/2,\ldots$ taking non-positive integer or half-integer values%
\footnote{The most common -- but not systematic -- convention in the literature has been to take the opposite sign for the spin $s=-s'$, with $s'$ the spin in their convention, yielding a Casimir $\propto s'(1-s')$, and as a result the left-hand side of the Bethe equations \eqref{befin} would be inverted. We prefer to stick to the convention where the Casimir is $\propto s(1+s)$ and \eqref{befin} unchanged.} (the case $s=0$ being obtained as a `limit') \cite{gelfand,bargmann,gelfandbook}. In both cases the generators  can be realized with differential operators
\begin{equation}
\label{gensl2r}
S^+\equiv S^x+iS^y=z^2\partial_z-2sz\,,\qquad S^-\equiv S^x-iS^y=-\partial_z\,,\qquad S^z=z\partial_z-s
\end{equation}
that verify the usual relations
\begin{equation}
[S^+,S^-]=2S^z\,,\qquad [S^z,S^\pm]=\pm S^\pm \,.
\end{equation}

The space on which these generators  act has been sometimes considered to be the space of polynomials, although it lacks a Hilbert space structure \cite{derkachovbaxter}. A proper choice of a Hilbert space is the set of analytic functions on the upper half-plane---or, up to a conformal transformation, on the unit disk \cite{riemann}---with a precise scalar product \cite{derkachovseparation,gelfandbook}.

\medskip

The construction of an $R$-matrix for the foregoing  values of $s$ requires the continuation of \eqref{rmat} to any real $s$. It can be rewritten as \cite{kulishyu,tarasovtakhtajanfaddeev}
\begin{equation}
R(\lambda)=\frac{\Gamma(i\lambda-2s)\Gamma(i\lambda+2s+1)}{\Gamma(i\lambda-\pmb{J})\Gamma(i\lambda+\pmb{J}+1)} \,,
\end{equation}
where $\pmb{J}$ satisfies
\begin{equation}
\pmb{J}(\pmb{J}+1)=2\pmb{\sigma}+2s(s+1) \,.
\end{equation}
As for the function $f(x)$, it can be rewritten as
\begin{equation}
f(\pmb{\sigma})=\psi(\pmb{J}+1)+\psi(-\pmb{J})-\psi(2s+1)-\psi(-2s)\,,
\end{equation}
with $\psi(x)=\Gamma'(x)/\Gamma(x)$.

The Hamiltonians thus defined
\begin{equation}
\mathcal{H}_L=\sum_{i=1}^L (\psi(\pmb{J}_{i,i+1}+1)+\psi(-\pmb{J}_{i,i+1})-\psi(2s+1)-\psi(-2s)) \,,
\end{equation}
where $\pmb{J}$ satisfies
\begin{equation}
\pmb{J}_{i,i+1}(\pmb{J}_{i,i+1}+1)=2\pmb{\sigma}_{i,i+1}+2s(s+1) \,,
\end{equation}
are called the non-compact $SL(2,\mathbb{R})$ Heisenberg spin chain of spin $s$ \footnote{There is minus sign difference with the definition of the Hamiltonian sometimes encountered in the literature \cite{fadeevkorchemsky,korchmeskybepom}, but in these papers the state with \textit{maximal} energy was studied.}. Since $s$ appears only as a parameter (albeit a crucial one) in this case---in contrast with the $su(2)$ case where the spin determines the dimension of the space---these chains are sometimes referred to generically as `the'  $SL(2,\mathbb{R})$ spin chain. The same remark applies to the $SL(2,\mathbb{C})$ spin chains below.
We note that this spin chain emerges in a QCD context in high-energy physics \cite{braunderkachovmanashov,braunderkachovmanashov2,gorsky,belitsky1,belitsky2,belitsky3}. 

\medskip

The function $\Omega(z_1,\ldots,z_L)=z_1^{2s} \cdots z_L^{2s}$ is a heighest-weight state, i.e., it satisfies $S^+_i \Omega=0$ and $\sum_{i=1}^L S^z_i \Omega=sL \Omega$, and is an eigenstate of the Hamiltonian\footnote{As it stands, it is actually not normalizable with respect to the scalar product considered. However, one can take a kind of Fourier transform of $\Omega(z_1-w,...,z_L-w)$, seen as a function of $w$, and obtain a normalizable vector with the same properties. We refer the reader to section 2.1 and appendix B of \cite{derkachovseparation} for the details of this construction.}. The ABA can then be applied to obtain eigenstates with $\Omega$ acting as pseudo-vacuum \cite{fadeevkorchemsky,derkachovbaxter,derkachovseparation}. The  expression of energy levels and the Bethe equations are exactly the same as in the finite-dimensional case, viz.\ \eqref{efin} and \eqref{befin} for $s\neq 0$, and with
\begin{equation}
 u=Ls-N \label{uLN}
\end{equation}
being the spin of the representation to which the state belongs, where $N$ denotes the number of Bethe roots. However, since $s$ is negative the structure of the Bethe roots changes dramatically \cite{korchmeskyqqcd,korchemskyduality,beisert,korchmeskylog,haokharzaeev}. Moreover, since $N$ has to be obviously a non-negative integer, the ABA  construction can only provide eigenstates for which $Ls-u$ is a non-negative integer, and continuous series representations for a real arbitrary $u$ cannot be obtained directly this way \cite{manashovkirch}.

\medskip

Let us now comment on the special case $s=0$. In this case the function $\Omega(z_1,\ldots,z_L)=z_1^{2s} \cdots z_L^{2s}$ is both heighest-weight and lowest-weight state and we cannot use it as a pseduo-vacuum. However, as shown in \cite{fadeevkorchemsky}, there is actually a one-to-one correspondence between the transfer matrices of the spin $s=0$ and $s=-1$ models. For each eigenstate $\hat{\varphi}(z_1,\ldots,z_L)$ of the $s=-1$ model the function $\varphi(z_1,\ldots,z_L)=(z_1-z_2)(z_2-z_3)\cdots(z_L-z_1)\hat{\varphi}(z_1,\ldots,z_L)$ is an eigenstate of the spin $s=0$ model. Consequently, the energies of the $s=0$ Hamiltonian read \cite{fadeevkorchemsky}
\begin{equation}
\label{efins0}
E=2L+\sum_{k=1}^N \frac{2}{\lambda_k^2+1} \,,
\end{equation}
where $\lambda_1,\ldots,\lambda_N$ satisfy the $s=-1$ Bethe equations
\begin{equation}
\label{befins0}
\left(\frac{\lambda_k-i}{\lambda_k+i}\right)^L=\prod_{l=1,\neq k}^N\frac{\lambda_k-\lambda_l+i}{\lambda_k-\lambda_l-i}\,,
\end{equation}
and $u$, the spin of the eigenstate, is related to $N$ through $u=-L-N$ \cite{fadeevkorchemsky}. We note that (\ref{befins0})  can be exactly interpreted as a set of  spin $s=0$ equations for which $L$ roots are imposed to be degenerate and equal to $0$.

\subsection{Non-compact $SL(2,\mathbb{C})$ spin chains}
Another Lie group whose Lie algebra is $su(2)$ and which has infinite-dimensional irreps is $SL(2,\mathbb{C})$, the universal cover of the Lorentz group. This is the case that we study in this article.

The only unitary irreps of $SL(2,\mathbb{C})$ are infinite-dimensional \cite{dirac} and are labelled by two complex numbers $s,\bar{s}$ that satisfy \cite{gelfandbook}
\begin{subequations}
\begin{eqnarray}
 s+\bar{s}^*+1 &=& 0 \,, \label{s-sbar} \\
 2(s-\bar{s}) &\in& \mathbb{Z} \,.
\end{eqnarray}
\end{subequations}
The six generators of $SL(2,\mathbb{C})$ can be represented by \eqref{gensl2r} and
\begin{equation}
\bar{S}^+\equiv \bar{S}^x+i\bar{S}^y=\bar{z}^2\partial_{\bar{z}}-2\bar{s}\bar{z}\,,\qquad \bar{S}^-\equiv \bar{S}^x-i\bar{S}^y=-\partial_{\bar{z}}\,,\qquad \bar{S}^z=\bar{z}\partial_{\bar{z}}-\bar{s} \,.
\end{equation}

The Hamiltonian of the non-compact $SL(2,\mathbb{C})$ spin chain is then given by two copies of that of the $SL(2,\mathbb{R})$ spin chain
\begin{equation}
\begin{aligned}
\mathcal{H}_L=&\sum_{i=1}^L (\psi(\pmb{J}_{i,i+1}+1)+\psi(-\pmb{J}_{i,i+1})-\psi(2s+1)-\psi(-2s))\\
&+\sum_{i=1}^L (\psi(\bar{\pmb{J}}_{i,i+1}+1)+\psi(-\bar{\pmb{J}}_{i,i+1})-\psi(2\bar{s}+1)-\psi(-2\bar{s}))\,,
\end{aligned}
\end{equation}
where $\bar{\pmb{J}}$ satisfies
\begin{equation}
\bar{\pmb{J}}_{i,i+1}(\bar{\pmb{J}}_{i,i+1}+1)=2\bar{\pmb{\sigma}}_{i,i+1}+2\bar{s}(\bar{s}+1) \,.
\end{equation}
The Hamiltonian is Hermitian \cite{derkachovkorchemsky1} and its two holomorphic and anti-holomorphic components (the two $SL(2,\mathbb{R})$ spin-chain copies) commute.
The case $(s,\bar{s})=(0,-1)$ has been particularly studied because of its relation with QCD, from which the model actually originates \cite{lipatov,fadeevkorchemsky} \footnote{We note that $\psi(-2s)+\psi(2\bar{s}+1)$ is not divergent when $(s,\bar{s})\to (0,-1)$, although each individual term is.}. This is the case that we will consider as well.

\medskip

Although the Hamiltonian of the non-compact $SL(2,\mathbb{C})$ spin chain is expressed as a sum of two commuting Hamiltonians to which one can apply the ABA separately to find eigenstates, this latter property does not hold for the total Hamiltonian. This can be understood as follows. Since the Hamiltonian is $SL(2,\mathbb{C})$-invariant and Hermitian, its eigenspaces can be decomposed into unitary irreps of $SL(2,\mathbb{C})$, and  labelled by two complex numbers $(u,\bar{u})$ satisfying\footnote{The notation always adopted in the references is to define $h=-u$, $\bar{h}=-\bar{u}$. We decided to change the notation in order to keep the same sign as in the compact case, and also because in the CFT context $h$ is used to denote the conformal weights of the operators, that we plan to study in another piece of work.}
\begin{subequations}
\begin{eqnarray}
 u+\bar{u}^*+1 &=& 0 \,, \label{u-ubar} \\
 2(u-\bar{u}) &\in& \mathbb{Z} \,.
\end{eqnarray}
\end{subequations}
Here, $u$ is the value of $\sum_{i=1}^L S^z_i$ on the highest-weight state of this representation, and $\bar{u}$ is the value of $\sum_{i=1}^L \bar{S}^z_i$. Since the holomorphic and anti-holomorphic generators of $SL(2,\mathbb{C})$ commute, this highest-weight state  also has to be an eigenstate of the separate $SL(2,\mathbb{R})$ Hamiltonians with spin $u$ and $\bar{u}$. Such an eigenstate can be constructed with the ABA only if $Ls-u$ and $L\bar{s}-\bar{u}$ are non-negative integers. Because of the relations \eqref{s-sbar} and \eqref{u-ubar}, these two constraints can never be satisfied simultaneously. Hence no eigenstate of the $SL(2,\mathbb{C})$ spin chain can be built with the ABA.

\medskip

The original attempts to work around this problem was based on the idea of rewriting the Bethe equations in such a way that $u$ can take any real value \cite{fadeevkorchemsky}. It is known \cite{baxter} that the Bethe equations \eqref{befin} can be recast into so-called $TQ$ relations
\begin{equation}
\label{tq}
T(\lambda)Q(\lambda)=(\lambda-is)^LQ(\lambda+i)+(\lambda+is)^LQ(\lambda-i) \,,
\end{equation}
where $T(\lambda)$ is a polynomial of degree $L$ and $Q(\lambda)$ a polynomial of degree $N$. In the case $L=2$ and $s=-1$, by inspecting the coefficients of $\lambda^{N+2}, \lambda^{N+1}, \lambda^N$ in \eqref{tq}, one has to have $T(\lambda)=2\lambda^2-(N+2)(N+1)=2\lambda^2-u(u+1)$, so that for an arbitrary $u$, \eqref{tq} with this value of $T(\lambda)$ can be seen as a functional equation on $Q(\lambda)$ (that needs not be a polynomial anymore). This equation---and thus, the problem---can then be solved in size $L=2$ \cite{fadeevkorchemsky}. It was   shown later  that the corresponding eigenstate can be obtained more directly \cite{derkachovkorchemsky1,derkachovkorchemsky2,derkachovkorchemsky3}. 

The case $L=2$ is however a bit special since in that case  the sole value of $u$ directly fixes the state and $T(\lambda)$, which can be seen from the fact that \eqref{befin} at $s=-1$ and $L=2$ has only one solution for each value of $N$, as follows from \eqref{range} hereafter. For $L\geq 3$ this is not true anymore, and additional  conserved charges (other than the spin  $u$) are needed to label the states. Considerable work has been  focused on obtaining the ground state of the model for higher values of $L$ \cite{braunderkachovmanashov,devegalipatov,janik2,janik}, up to $L=8$ \cite{derkachovkorchemsky2}. From these values it was conjectured that the ground-state energy goes to $0$ for $L\to\infty$ \cite{derkachovkorchemsky2}.

\medskip

We can now state the ideas of this paper. Although one cannot use the ABA to build the eigenstates, the fact that the Hamiltonian is a sum of two commuting $SL(2,\mathbb{R})$ Hamiltonians implies that an $SL(2,\mathbb{C})$ energy level at $(u,\bar{u})$ is necessarily a sum of two $SL(2,\mathbb{R})$ energies at $u$ and $1-u^*$, and obtaining both requires continuing the solutions of Bethe equations to a negative number of Bethe roots. Instead of analytically continuing the Bethe or $TQ$ relations in finite-size to reach arbitrary real values of $u$, we perform an analytic continuation of the Bethe equations {\sl directly in the thermodynamic limit}. This is done by introducing an imaginary extensive twist $\varphi$ in \eqref{befin}, which permits us to expand the energy levels in $e^{-2\varphi}$, yielding an expansion `dual' to that of \cite{granetjacobsensaleurxxz} (where the magnetization $m$ was used as an expansion parameter). We obtain that the energy levels can be expanded around another 'pseudo-vacuum' so as to reach other states in the spectrum in the thermodynamic limit. As a proof of principle, we recover in this paper the thermodynamic ground-state value previously obtained by extrapolating the ground state from small sizes \cite{derkachovkorchemsky2}. In our case, this ground state is identified by being of minimal energy with respect to a certain (but large) class of particle-hole excitations. Our approach provides, as far as we know,  the first description of the $SL(2,\mathbb{C})$ non-compact spin chain in the thermodynamic limit. Our study also reveals new insights on the analytic continuation of the energies in Bethe-ansatz solvable models.

\section{Bethe equations with an imaginary extensive twist \label{analy}}
\subsection{Generalities}
In a nutshell, our goal is to perform the analytic continuation of the energies of the spin $s=-1$ chain
\begin{equation}
\label{eqen}
e\equiv \frac{E}{L}=\frac{1}{L}\sum_{k=1}^N \frac{2}{\lambda_k^2+1} \,,
\end{equation}
where the $\lambda_k$ satisfy the Bethe equations
\begin{equation}
\label{befinm1}
\left(\frac{\lambda_k-i}{\lambda_k+i}\right)^L=\prod_{l=1,\neq k}^N\frac{\lambda_k-\lambda_l+i}{\lambda_k-\lambda_l-i} \,,
\end{equation}
to any real (including negative) values of
\begin{equation} \label{m-def}
m=\frac{N}{L}
\end{equation}
in the thermodynamic limit $L\to\infty$.  Once some energies (per site) of the $SL(2,\mathbb{R})$ spin chain at a given $m$ in the thermodynamic limit, denoted here $e_i(m)$, are identified, one obtains an energy level $\mathcal{E}_{i,j}(m)$ of the $SL(2,\mathbb{C})$ spin chain as
\begin{equation}
\label{mathcalf}
\mathcal{E}_{i,j}(m)=2+e_i(m)+e_j(-2-m) \,,
\end{equation}
with possibly some constraints on $i,j$. Indeed, $e_i(m)$ is the energy corresponding to the sub-$SL(2,\mathbb{R})$ Hamiltonian with $s=-1$ at magnetization $m$, whereas the other $\bar{s}=0$ sub-$SL(2,\mathbb{R})$ Hamiltonian is then at magnetization $m'=-2-m$ in the thermodynamic limit, since $\bar{u}=-1-u$ from \eqref{u-ubar} with $u=-L-mL$ and $\bar{u}=-L-m'L$ from \eqref{uLN}. The state of the latter has thus a intensive energy $2+e_j(-2-m)$ because of \eqref{efins0}.

The writing \eqref{mathcalf} emphasizes that there is not necessarily the same state for the two sub-$SL(2,\mathbb{R})$ spin chains.

Due to the fact that $\mathcal{E}_{i,i}(m)$ in \eqref{mathcalf} has an extremum at $m=-1$, we will look for the ground state at $m=-1$. This is in agreement with the fact that in finite size $L$ (for example $L=2$) the ground state is at $u=-1/2$ \cite{derkachovkorchemsky2}, meaning that $u=\bar{u}$ and hence $m=m'$ in that case as well.

\subsection{Structure of the solutions at zero twist}

One can rewrite the Bethe equations \eqref{befinm1} in the following form by taking their logarithm
\begin{equation}
\label{belog}
\frac{1}{\pi}\arctan \lambda_k=\frac{I_k}{L}-\frac{1}{\pi L}\sum_{l=1}^N \arctan(\lambda_k-\lambda_l) \,,
\end{equation}
where the $I_k$ (for $L$ even: integer if $N$ is odd, and half-integer if $N$ is even) are  called Bethe numbers. These (half)-integers emerge from $\log(zz')=\log z+\log z'+2i\pi n$ with $n=-1,0,1$, valid for $z,z'$ two non-zero complex numbers. These Bethe equations have been extensively studied previously: we give here only the properties that will be of importance to our discussion. We will consider $L$ even only.

\medskip

\begin{prop} \label{thm01}
The solutions to \eqref{belog} with $\lambda_k\neq \lambda_l$ if $k\neq l$, are all real and characterized by the choice of $N$ distinct (half-)integers $I_k$ satisfying 
\begin{equation}
\label{range}
-\frac{L+N-1}{2}< I_k < \frac{L+N-1}{2}
\end{equation}
\end{prop}

\begin{proof}

Let us show first that the equations \eqref{befinm1} only have real solutions. The proof is identical to that of the same property for the repulsive Lieb-Liniger model \cite{korepin,haokharzaeev}, and can be formulated as follows. Let us denote $\lambda_+$ the root with the largest imaginary part. The differences $\lambda_+-\lambda_l$ thus always have a positive or zero imaginary part. Since $|\frac{\lambda+i}{\lambda-i}|\geq 1$ if and only if $\Im \lambda\geq 0$, we deduce from \eqref{befinm1} for $\lambda_k=\lambda_+$ that $|\frac{\lambda_+-i}{\lambda_++i}|\geq 1$. From the same inequality one infers thus that $\Im \lambda_+\leq 0$, which means that the imaginary part of all the roots are negative or zero. Doing the same reasoning with $\lambda_-$ the root with the smallest imaginary part, one infers that the imaginary part of all the roots are positive or zero. Hence all the roots are real.

 Now, using $| \! \arctan x|<\pi/2$ in \eqref{belog}, one directly obtains \eqref{range}.

To show that  with this constraint \eqref{range} a solution to \eqref{belog} does exist and is unique, we follow again \cite{korepin,haokharzaeev} and introduce
\begin{equation}
M(\lambda_1,\ldots,\lambda_N)=\frac{1}{\pi}\sum_{k=1}^N A( \lambda_k)-\frac{1}{L}\sum_{k=1}^N \lambda_k I_k+\frac{1}{2\pi L}\sum_{k,l}A(\lambda_k-\lambda_l)\,,
\end{equation}
where $A(x)$ is the primitive of $\arctan(x)$ that vanishes at $0$. The Bethe equations \eqref{belog} are exactly the stationary conditions $\partial_{\lambda_k}M(\lambda_1,\ldots,\lambda_N)=0$ necessary for $M$ to be minimal at $\lambda_1,\ldots,\lambda_N$. To prove that this minimum exists and is unique, we show that $M$ is strictly convex. To that end, we consider $v_i$ a non-zero vector of size $N$ and compute
\begin{equation}
\sum_{i,j}v_iv_j\partial_{\lambda_i}\partial_{\lambda_j}M =\frac{1}{\pi}\sum_{i=1}^N \frac{v_i^2}{1+\lambda_i^2}+\frac{1}{2\pi L}\sum_{i,j}\frac{(v_i-v_j)^2}{1+(\lambda_i-\lambda_j)^2}>0\,.
\end{equation}
This shows that the matrix $\partial_{\lambda_i}\partial_{\lambda_j}M$ is definite positive and hence $M$ strictly convex.

Finally, to show that $\lambda_k\neq \lambda_l$ requires the Bethe numbers to be distinct, let us subtract \eqref{belog} for $k$ and $l$
\begin{equation}
\frac{1}{\pi}\arctan \lambda_k+\frac{1}{\pi L}\sum_{j=1}^N\arctan(\lambda_k-\lambda_j)-\left( \frac{1}{\pi}\arctan \lambda_l+\frac{1}{\pi L}\sum_{j=1}^N\arctan(\lambda_l-\lambda_j)\right)=\frac{I_k-I_l}{L}\,.
\end{equation}
Since the function $x\mapsto \frac{1}{\pi}\arctan x+\frac{1}{\pi L}\sum_{j=1}^N\arctan(x-\lambda_j)$ is strictly increasing for any $x$, we conclude that $\lambda_k>\lambda_l$ if and only if $I_k>I_l$. Hence all the roots are distinct if and only if all the Bethe numbers are distinct.
\end{proof}

We remark that from inequality \eqref{range} one sees another property of these equations: even in finite size $L$ they admit an infinite quantity of solutions, since $N$ can be taken as large as desired, which reflects the non-compactness of the spin chain.
\medskip

The scaling limit $L\to\infty$ of \eqref{belog} is then taken as follows. We recall from \eqref{m-def} that $m=N/L$. The {\sl filling function} $\chi_m(x)$ is defined such that $L\chi_m(x) \, {\rm d}x$ is the number of Bethe numbers $I$ with $x<\frac{I}{L}<x+{\rm d}x$ for large $L$. The inverse of the counting function $z(x)$ is defined as the value of the roots $\lambda_k$ such that their Bethe number verifies $I_k/L\to x$ for large $L$. Using \eqref{range} we can then rewrite the logarithmic Bethe equations \eqref{belog} as
\begin{equation}
\frac{1}{\pi}\arctan z(x)=x-\frac{1}{\pi}\int_{-\tfrac{1+m}{2}}^{\tfrac{1+m}{2}}\arctan(z(x)-z(y))\chi_m(y) \, {\rm d}y \,.
\end{equation}
The possible filling functions $\chi_m(x)$ are exactly the functions that satisfy
\begin{subequations}
\label{fill}
\begin{eqnarray}
& & \forall x\in [-\tfrac{1+m}{2},\tfrac{1+m}{2}]\,: \quad 0\leq \chi_m(x)\leq 1\,, \\
& & \int_{-\tfrac{1+m}{2}}^{\tfrac{1+m}{2}}\chi_m(x)\, {\rm d}x=m \,.
\end{eqnarray}
\end{subequations}

\subsection{An expansion in terms of the twist \label{exptwist}}
Our strategy is now to add an imaginary extensive twist $\varphi\geq 0$ in the Bethe equations \eqref{befinm1}, that become\footnote{We put a factor $2$ in the exponent so that $\varphi$ is the conjugate variable to $m$, see \cite{granetdubailjacobsen}, and a minus sign so that the expansion at $\varphi\to +\infty$ leads to roots converging to $+i \mathbb{N}^*$, see hereafter.}
\begin{equation}
\label{befinm1tw}
\left(\frac{\lambda_k-i}{\lambda_k+i}\right)^L=e^{-2\varphi L}\prod_{l=1,\neq k}^N\frac{\lambda_k-\lambda_l+i}{\lambda_k-\lambda_l-i} \,,
\end{equation}
and to study the energy as a function of $\varphi$ when expanded around $\varphi\to\infty$. The logarithmic form of the Bethe equations with this twist is
\begin{equation}
\label{belogtwist}
\frac{1}{\pi}\arctan \lambda_k=\frac{I_k}{L}+\frac{i\varphi}{\pi}-\frac{1}{\pi L}\sum_{l=1}^N \arctan(\lambda_k-\lambda_l) \,.
\end{equation}
This kind of imaginary twist has been studied in different contexts in the XXZ spin chain \cite{fukuikawakami,nohkim,granetdubailjacobsen,granetjacobsensaleurxxz}. Our point is to show that it is actually suited for the convergent extrapolation from $\varphi=+\infty$ down to $\varphi=0$. We start our reasoning with the following
\begin{prop} \label{thm0}
When $\varphi\to\infty$ at fixed $L$, the roots $\{\lambda_k\}$ of a solution to \eqref{belogtwist} satisfy $\{\lambda_k\}\subset i \mathbb{N}^*$. There is necessarily a root that converges to $i$, and if there exists a root converging to $ni$ for $n>1$ then there exists another root converging to $(n-1)i$. Moreover, all the roots converge to $i$ if and only if all the Bethe numbers satisfy
\begin{equation}
\label{eqbenum}
-\frac{L}{2}< I_k \leq\frac{L}{2}
\end{equation}

\end{prop}

\begin{proof}
First, let us show that no roots go to $\infty$ when $\varphi\to\infty$. Indeed, let us denote $K$ the (possible empty) set of roots such that $\lambda_k\to\infty$ when $\varphi\to\infty$. Taking the product of \eqref{befinm1tw} for these $\lambda_k$, we obtain
\begin{equation}
\prod_{\lambda_k\in K}\left(\frac{\lambda_k-i}{\lambda_k+i}\right)^L=e^{-2\varphi L |K|}(-1)^{|K|}\prod_{\lambda_k\in K, \lambda_l\notin K}^N\frac{\lambda_k-\lambda_l+i}{\lambda_k-\lambda_l-i} \,.
\end{equation}
The left-hand side goes to $1$ when $\varphi\to\infty$, so one needs $|K|=0$ in the right-hand side for it to not vanish when $\varphi\to\infty$. Hence all the roots stay finite.

Now, in \eqref{befinm1tw} if we consider $\lambda_k$ the root with the smallest imaginary part, $\lambda_k-\lambda_l-i$ cannot vanish, so that when $\varphi\to\infty$ we must have $\lambda_k\to i$ for the left-hand side to vanish, since all the roots stay finite when $\varphi\to\infty$. If we now consider an arbitrary $\lambda_k$, in the limit $\varphi\to\infty$ we must have either $\lambda_k \to i$ or there exists another $\lambda_l$ such that $\lambda_k-\lambda_l-i\to 0$. Hence by recurrence we must have $\lambda_k\to ni$ with $n>0$ an integer, and then $\lambda_l\to (n-1)i$.

Now, since $-\pi/2<\Re \arctan z\leq \pi/2$ for all complex $z$, by taking the real part of \eqref{belogtwist} we have
\begin{equation}
-\frac{L+N}{2}\leq I_k \leq\frac{L+N}{2}\,.
\end{equation}
Let us consider then a solution for which all $\lambda_k\to i$ when $\varphi\to\infty$. Then we have for all $l$, $\lambda_k-\lambda_l\to 0$ and from the real part of \eqref{belogtwist}, with again $-\pi/2<\Re \arctan z\leq \pi/2$ for all complex $z$, we obtain \eqref{eqbenum}.

We admit the other direction of the equivalence, i.e., that if \eqref{eqbenum} is verified, then all the roots converge to $i$, which is indeed observed numerically.
\end{proof}

We will call \textit{first-level filling function} a filling function $\chi_m(x)$ such that $\chi_m(x)=0$ for $\tfrac{1}{2}<|x|<\tfrac{1+m}{2}$, i.e., such that all the Bethe numbers satisfy \eqref{eqbenum} in the thermodynamic limit. Then, according to Proposition \ref{thm0}, when $\varphi\to\infty$ all the roots converge to $i$. Then we have the following 

\medskip

\begin{theorem} \label{thm1}

The energy $F_{\chi_m}(\varphi)$ as a function of $\varphi$, for a given first-level filling function $\chi_m(x)$ at a given value of $m>0$, can be expanded as
\begin{equation}
\label{repfree}
F_{\chi_m}(\varphi)=\sum_{b\geq -1}e^{-2b\varphi}f_b(\chi_m) \,,
\end{equation}
where the  $b$'s  are integers. The functions  $f_b({\chi_m})$ depend only on the moments $X_a(\chi_m)$ of $\chi_m$, defined for $a$ integer by
\begin{equation}
\label{momentsdef}
X_a(\chi_m)\equiv \int_{-\tfrac{1+m}{2}}^{\tfrac{1+m}{2}}e^{2i\pi a x}\chi_m(x) \, {\rm d}x=\underset{L\to\infty}{\lim}\, \frac{1}{L}\sum_{k=1}^Ne^{2i\pi a \tfrac{I_k}{L}} \,,
\end{equation}
and can be computed recursively in terms of a finite number of $X_a(\chi_m)$ with only algebraic manipulations.

\end{theorem}

\medskip

In order to obtain this result, we show that the following ansatz for each Bethe root $\lambda_k$
\begin{equation}
\label{lamdaexp}
\lambda_k=i+\sum_{a,b\geq 1}e^{-2b\varphi} e^{\tfrac{2i\pi I_k}{L}a}c_{ab} \,,
\end{equation}
with $c_{ab}$ coefficients that satisfy a yet-to-be-determined recurrence relation, solves the Bethe equations.
Note that the fact that $\lambda_k\to i$ when $\varphi\to\infty$ is consistent with the second part of Proposition~\ref{thm0}, because we have assumed the filling function to be first-level.
We will use the convenient notation
\begin{equation}
\label{ctilde}
c_{ab}^{[k]}=\sum_{\substack{a_1+\ldots+a_k=a\\b_1+\ldots+b_k=b}}c_{a_1b_1} c_{a_2 b_2} \cdots c_{a_kb_k} \,,
\end{equation}
with the convention $c_{00}^{[0]}=1$. In \eqref{lamdaexp} we can take $a,b\geq 0$, if we set $c_{ab}=0$ whenever $a=0$ or $b=0$. 

\begin{proof}

The ideas of the derivation are close to those used in \cite{granetjacobsensaleurxxz}. We first notice the identity
\begin{equation}
\label{arctan}
\arctan(i+x)=\frac{\log \tfrac{ix}{2}}{2i}+\frac{1}{2i}\sum_{n\geq 1}\frac{(-x)^n}{n(2i)^n} \,.
\end{equation}
Inserting the expansion \eqref{lamdaexp} into \eqref{arctan} (with $x = \lambda_k - i$), we have
\begin{equation}
\begin{aligned}
\frac{1}{\pi}\arctan(\lambda_i)=&-\frac{\varphi}{i\pi}+\frac{I_i}{L}+\frac{\log \tfrac{ic_{11}}{2}}{2i\pi}+\frac{1}{2i\pi}\log\left(1+\sum_{a,b\geq 0}e^{-2b\varphi} e^{\tfrac{2i\pi I_i}{L}a}\tilde{c}_{ab}\right)\\
&+\frac{1}{2i\pi}\sum_{n\geq 1}\frac{(-1)^n}{n(2i)^n}\sum_{a,b\geq 1}e^{-2b\varphi} e^{\tfrac{2i\pi I_i}{L}a} c^{[n]}_{ab}\\
=&\frac{i\varphi}{\pi}+\frac{I_i}{L}+\frac{\log \tfrac{ic_{11}}{2}}{2i\pi}-\frac{1}{2i\pi}\sum_{n\geq 1}\frac{(-1)^n}{n}\sum_{a,b\geq 0}e^{-2b\varphi} e^{\tfrac{2i\pi I_i}{L}a}\tilde{c}^{[n]}_{ab}\\
&+\frac{1}{2i\pi}\sum_{n\geq 1}\frac{(-1)^n}{n(2i)^n}\sum_{a,b\geq 1}e^{-2b\varphi} e^{\tfrac{2i\pi I_i}{L}a} c^{[n]}_{ab} \,,
\end{aligned}
\end{equation}
where we have set
\begin{equation}
\label{ctildedef}
\tilde{c}_{ab}=\begin{cases}
\frac{c_{a+1,b+1}}{c_{11}}\quad\text{if }(a,b)\neq(0,0)\\
0\quad\text{if }(a,b)=(0,0)
\end{cases}
\end{equation}
and with an identical definition for $\tilde{c}_{ab}^{[k]}$ as in \eqref{ctilde}:
\begin{equation}
\tilde{c}_{ab}^{[k]}=\sum_{\substack{a_1+\ldots+a_k=a\\b_1+\ldots+b_k=b}}\tilde{c}_{a_1b_1} \tilde{c}_{a_2 b_2} \cdots \tilde{c}_{a_kb_k} \,.
\end{equation}
We used that the Bethe numbers all satisfy $-L/2< I_k \leq L/2$ to write $\log e^{2i\pi I_k/L}=2i\pi I_k/L$. The right-hand side of \eqref{belog} can also be written in terms of the $c_{ab}$'s. We expand $\arctan x$ around $0$,
\begin{equation}
\frac{1}{\pi}\arctan(\lambda_i-\lambda_j)=\frac{1}{\pi}\sum_{n\geq 0}\frac{\arctan^{(n)}(0)}{n!} (\lambda_i - \lambda_j)^n \,,
\end{equation}
perform a binomial expansion
\begin{equation}
 (\lambda_i - \lambda_j)^n = \sum_{q=0}^n {n \choose q} (-1)^{n-q} (\lambda_i - i)^q (\lambda_j - i)^{n-q}
\end{equation}
and insert again \eqref{lamdaexp}, yielding
\begin{equation}
\begin{aligned}
\frac{1}{\pi}\arctan(\lambda_i-\lambda_j)=&\frac{1}{\pi}\sum_{n\geq 0}\frac{\arctan^{(n)}(0)}{n!}\sum_{q=0}^n {n\choose q}(-1)^{n-q}\\
&\times \sum_{a_1,b_1\geq 1}e^{-2b_1\varphi} e^{\tfrac{2i\pi I_i}{L}a_1}c^{[q]}_{a_1b_1}\sum_{a_2,b_2\geq 1}e^{-2b_2\varphi}e^{\tfrac{2i\pi I_j}{L}a_2}c^{[n-q]}_{a_2b_2} \,.
\end{aligned}
\end{equation}
In this form, the sum over the roots $\lambda_j$ can be expressed in the thermodynamic limit in terms of the moments $X_a(\chi_m)$, using \eqref{momentsdef}. It yields
\begin{equation}
\begin{aligned}
\frac{1}{L}\sum_{j}\frac{1}{\pi}\arctan(\lambda_i-\lambda_j)=&\frac{1}{\pi}\sum_{n\geq 0}\frac{\arctan^{(n)}(0)}{n!}\sum_{q=0}^n {n\choose q}(-1)^{n-q}\\
&\times \sum_{a_1,b_1\geq 1}e^{-2b_1\varphi} e^{\tfrac{2i\pi I_i}{L}a_1}c^{[q]}_{a_1b_1}\sum_{a_2,b_2\geq 1}e^{-2b_2\varphi}X_{a_2}(\chi_m)c^{[n-q]}_{a_2b_2}+O(L^{-1})\,.
\end{aligned}
\end{equation}
Plugging these expressions into the logarithmic form of the Bethe equations \eqref{befinm1tw}, we obtain
\begin{equation}
\label{eqeq}
\begin{aligned}
&\frac{\log \tfrac{ic_{11}}{2}}{2i\pi}-\frac{1}{2i\pi}\sum_{a,b\geq 0}e^{-2b\varphi} e^{\tfrac{2i\pi I_i}{L}a}\sum_{n\geq 1}\frac{(-1)^n}{n}\tilde{c}^{[n]}_{ab}\\
&=\sum_{a,b\geq 1}e^{-2b\varphi} e^{\tfrac{2i\pi I_i}{L}a}\left[-\sum_{n\geq 0}\sum_{q=0}^n \sum_{\substack{ a_2,b_1,b_2\geq 1\\b_1+b_2=b}}\frac{\arctan^{(n)}(0)}{n!\pi}{n\choose q}(-1)^{n-q}c_{ab_1}^{[q]}c_{a_2b_2}^{[n-q]}X_{a_2}(\chi_m) \right. \\
& \qquad \qquad \qquad \qquad \left. +\sum_{n\geq 1}\frac{(-1)^{n+1}}{n(2i)^{n+1}\pi}c_{ab}^{[n]}\right] \,.
\end{aligned}
\end{equation}
We see now that we can solve this equation if we impose the initial condition
\begin{equation}
c_{11}=-2i
\end{equation}
that cancels out the first term of \eqref{eqeq}, as well as requiring the recurrence relation
\begin{eqnarray}
\label{recur}
\frac{\tilde{c}_{ab}}{2i\pi} &=& -\sum_{n\geq 0}\sum_{q=0}^n \sum_{\substack{ a_2,b_1,b_2\geq 1\\b_1+b_2=b}}\frac{\arctan^{(n)}(0)}{n!\pi}{n\choose q}(-1)^{n-q}c_{ab_1}^{[q]}c_{a_2b_2}^{[n-q]}X_{a_2}(\chi_m) \nonumber \\
&+& \sum_{n\geq 1}\frac{(-1)^{n+1}}{n(2i)^{n+1}\pi}c_{ab}^{[n]}+\frac{1}{2i\pi}\sum_{n\geq 2}\frac{(-1)^n}{n}\tilde{c}_{ab}^{[n]} \,.
\end{eqnarray}

For this to make sense, we first have to make sure that the sums on the right-hand side are finite, namely that the sums over $n$ and $a_2$ truncate. To this end, let us prove by recurrence on $b$ that $c_{ab}=0$ for all $a>b$. For $b=1$ this follows from \eqref{recur} with $b=0$ (recall $\tilde{c}_{a0}=\tfrac{c_{a+1,1}}{c_{1,1}}$) by recurrence on $a$: it is true for $c_{1,0}=0$, and the right-hand side only involves $c_{a',0}$ for $1\leq a'\leq a$. We assume now it is true for all $b'$ until and including $b$, and consider \eqref{recur} for $a>b$. First, we have $c_{ab}^{[n]}=0$ for all $n\geq 1$, since in $c_{ab}^{[n]}$ there must be a term $c_{a'b'}$ with $a'>b'$ for the sum over $a'$ to be strictly larger than the sum over $b'$. Since $b_1\leq b$ in \eqref{recur}, we also conclude that $c_{ab_1}^{[q]}=0$. We also have $\tilde{c}_{ab}^{[n]}=0$ for $n\geq 2$, since it involves only $\tilde{c}_{a'b'}$ for $b'<b$, and at least one $a'$ has to be larger than $b'$ for their sum to be strictly larger than $b$ in $\tilde{c}_{ab}^{[n]}$. Hence $\tilde{c}_{ab}=0$, which concludes our recurrence. From this it follows that the sums over $n,a_2$ are always finite sums, since $c^{[n]}_{ab}$ is zero for $n$ or $a$ large enough, and $\tilde{c}^{[n]}_{ab}$ is zero for $n$ large enough. 

Now, let us check that \eqref{recur} is indeed a recurrence relation for $c_{ab}$. The right-hand side of \eqref{recur} depends on $c_{a'b'}$ for $a'\leq a+1$ and $b'\leq b+1$, with at least $b'<b+1$ or $a'<a+1$. Indeed $c_{a'b'}^{[k]}$ depends only on $c_{a''b''}$ with $a''\leq a'-(k-1)$ and $b''\leq b'-(k-1)$, because  $c_{a''b''}=0$ if $a''=0$ or $b''=0$; and $\tilde{c}_{a'b'}^{[k]}$ depends only on $\tilde{c}_{a''b''}$ with $a''\leq a'$ and $b''\leq b'$, with at least $a''<a'$ or $b''<b'$, because $\tilde{c}_{00}=0$. Hence \eqref{recur} is indeed a recurrence relation for $c_{ab}$.

\medskip

We can now express the energy \eqref{eqen} in terms of these $c_{ab}$. Indeed, differentiating \eqref{arctan} that we evaluate at $\lambda_k-i$ with the representation \eqref{lamdaexp} for $\lambda_k$, we have
\begin{equation}
\begin{aligned}
\frac{2}{\lambda_k^2+1}&=\frac{e^{2\varphi}e^{-\tfrac{2i\pi I_k}{L}}}{ic_{11}}\frac{1}{1+\sum_{a,b\geq 0}e^{-2b\varphi} e^{\tfrac{2i\pi I_k}{L}a}\tilde{c}_{ab}}+\frac{1}{2}\sum_{a,b\geq 1}e^{-2b\varphi} e^{\tfrac{2i\pi I_k}{L}a}\sum_{n\geq 0}\frac{c_{ab}^{[n]}}{(-2i)^n}\\
&=\frac{1}{ic_{11}}\sum_{a,b\geq -1}e^{-2b\varphi}e^{\tfrac{2i\pi I_k}{L}a}\sum_{n\geq 0}(-1)^n\tilde{c}^{[n]}_{a+1,b+1}+\frac{1}{2}\sum_{a,b\geq 1}e^{-2b\varphi} e^{\tfrac{2i\pi I_k}{L}a}\sum_{n\geq 0}\frac{c_{ab}^{[n]}}{(-2i)^n} \,.
\end{aligned}
\end{equation}
After summing over $\lambda_k$, we obtain the representation \eqref{repfree} for the  energy $F_{\chi_m}(\varphi)$ with
\begin{equation}
\label{fbchi}
f_b(\chi_m)=\sum_{n\geq 0}\sum_{a\geq -1}X_a(\chi_m)\left(\frac{(-1)^n}{ic_{11}}\tilde{c}^{[n]}_{a+1,b+1}+\frac{c^{[n]}_{ab}}{2(-2i)^n}\right)\,.
\end{equation}
Because $c_{ab}=0$ for $a>b$ as proven before, the sum over $a$ in \eqref{fbchi} is truncated after $b+1$, and the sum over $n$ is finite as well. Hence \eqref{fbchi} is indeed a finite expression, and this concludes the proof of our claim.
\end{proof}

For example, we have the first terms (where we recall that the $X_a$ are the moments defined in \eqref{momentsdef})
\begin{equation}
\label{expflargephi}
\begin{aligned}
F_{\chi_m}(\varphi)&=e^{2\varphi}\frac{X_{-1}}{2}\\
&+X_0+2X_0^2-2X_1X_{-1}\\
&+e^{-2\varphi}\left[\left(\frac{1}{2}-2X_0-4X_0^2\right)X_1-4X_{-1}X_1^2+2(1+4X_0)X_{-1}X_2\right]\\
&+O(e^{-4\varphi}) \,.
\end{aligned}
\end{equation}
We refer the reader to appendix \ref{mathematica1} for a numerical code that computes the values of the expansion coefficients \eqref{fbchi}.

\subsection{Examples of root configurations and numerical checks \label{examples}}
Let us give some examples of root configurations. The simplest choice of a filling function satisfying \eqref{fill} is
\begin{equation}
\label{root1eq}
\chi_m^{(1)}(x)=\begin{cases}
1 & \text{if }-m/2<x<m/2 \,, \\
0 & \text{otherwise} \,.
\end{cases}
\end{equation}
This  corresponds to the `standard' root configuration where all the Bethe roots are symmetric and closely packed around the origin, and appears to be relatively often the ground-state configuration for various spin chains \cite{yangyang66a}. For this reason we will sometimes denote by 'free energy' the energy of this state as a function of the magnetization $m$. With the  expression \eqref{eqen} for the energy, however,  it is natural to expect (because of the sign) that it will, in the case of interest here, rather maximise the energy at $m>0$ fixed. The corresponding  moments are
\begin{equation}
\label{moment1}
X_a(\chi_m^{(1)})=\begin{cases}
m & \text{if }a=0 \,, \\
\frac{\sin (\pi a m)}{\pi a} & \text{otherwise} \,.
\end{cases}
\end{equation}
In Figure \ref{freeen1} we show a sketch of this root configuration. At the top, we indicated in red where the roots $\lambda_k$ lie on the black line $[-\tfrac{1+m}{2},\tfrac{1+m}{2}]$. At the bottom, we indicated in red where the quantities $e^{2i\pi\lambda_k}$ lie on the unit circle. In the right panel, we compare the numerical solutions of the Bethe equations to the series in $e^{-2\varphi}$ within their radius of convergence as a function of $\varphi$.

\begin{figure}[H]
\begin{tikzpicture}[scale=1]
\draw[black] (-3,2) -- (3,2);
\draw[red, thick, line width=2pt] (-1,2.1) -- (1,2.1);
\draw [black,thick,domain=0:360] plot ({cos(\x)}, {sin(\x)});
\draw [red,thick, line width=2pt,domain=-45:45] plot ({1.1*cos(\x)}, {1.1*sin(\x)});
\filldraw[white] ({cos(-45)},{sin(-45)-2}) circle (2pt);
\end{tikzpicture}
\begin{tikzpicture}[scale=1]
\begin{axis}[
    enlargelimits=false,
    ymax =6,
    xmin=0,
    xmax=0.25,
    xlabel = $e^{-2\varphi}$,
    ylabel = $F_{\chi_m^{(1)}}(\varphi)$,
     y tick label style={
        /pgf/number format/.cd,
            fixed,
            fixed zerofill,
            precision=0,
        /tikz/.cd
    }
]
\addplot [
    domain=0:0.105, 
    samples=100, 
    color=blue,
    ]
    {1.77368 + 0.11254/x - 1.0637* x - 6.61139 *x^2 - 32.0443* x^3 - 
 95.1535* x^4 + 450.139 *x^5 + 11148.5* x^6 + 115374. *x^7 + 
 740609. *x^8 + 1.11878*10^6 *x^9 - 4.75151*10^7 *x^10 - 
 8.05694*10^8 *x^11 - 7.68942*10^9* x^12 - 3.935*10^10 *x^13 + 
 1.60805*10^11* x^14 + 6.58422*10^12 *x^15 + 8.76788*10^13 *x^16 + 
 6.98516*10^14* x^17 + 1.70802*10^15 *x^18 - 4.80401*10^16 *x^19};
\addplot [
    domain=0:0.25, 
    samples=100, 
    color=blue,
    ]
    {0.11254/x+0.273679+0.0414095*x+0.153148*x^2+(-0.0544894)*x^3+(-0.945448)*x^4+3.33247*x^5+(-0.0541935)*x^6+(-41.2164)*x^7+139.177*x^8+122.571*x^9+(-2583.74)*x^10+7281.69*x^11+16958.1*x^12+(-189236.)*x^13+403966*x^14+1.94049*10^6*x^15+(-1.49515*10^7)*x^16+(2.04716*10^7)*x^17+2.1042*10^8*x^18+(-1.22246*10^9)*x^19
    };
        \addplot+[
    only marks,
    mark=+,
    mark size=2.9pt,
    color=red]
table{freeenergysl2c_1_L200.dat};

        \addplot+[
    only marks,
    mark=+,
    mark size=2.9pt,
    color=red]
table{freeenergysl2c_2_L100.dat};

\end{axis}
\end{tikzpicture}
\caption{Left: sketch of the root configuration (red) and the vacancies (black), on the real axis (top) and on the unit circle in the form $e^{2i\pi\lambda_k}$ (bottom). Right: $F_{\chi_m^{(1)}}(\varphi)$ as a function of $e^{-2\varphi}$, with $m=0.25$ (bottom) and $m=0.75$ (top), using twenty terms of \eqref{repfree} within its radius of convergence (blue) and solving numerically the Bethe equations in size $L=200$ and $L=100$ respectively (red). The radius of convergence is estimated numerically from the fact that the partial sums are stable within it.}
\label{freeen1}
\end{figure}
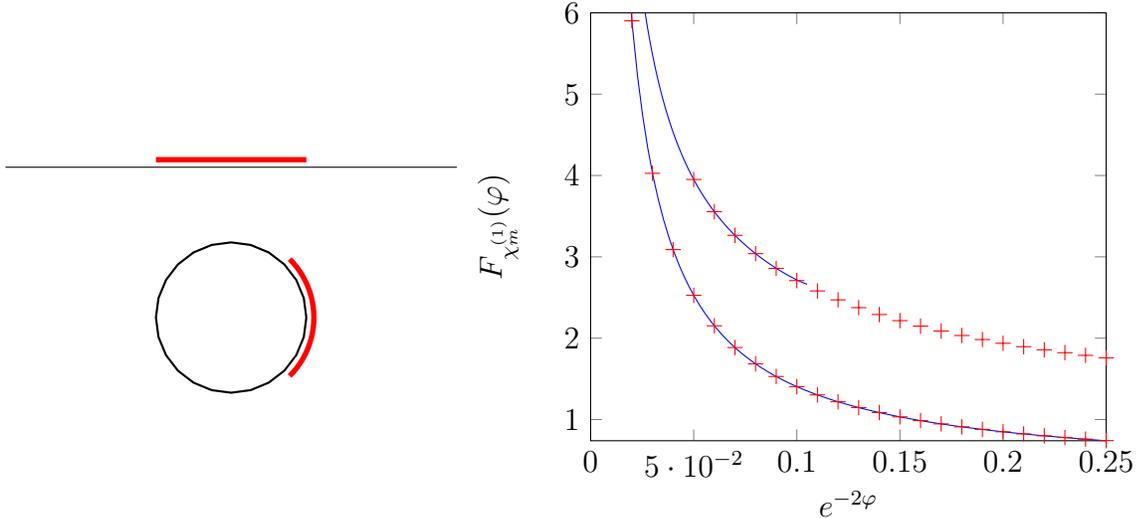

Another example of a root configuration is described by the filling function
\begin{equation} \label{chi-m-2}
\chi_m^{(2)}(x)=\begin{cases}
1 & \text{if }1/2-m/2<x<1/2\,, \\
1 & \text{if }-1/2<x<-1/2+m/2 \,, \\
0 & \text{otherwise} \,.
\end{cases}
\end{equation}
The corresponding  moments read
\begin{equation}
\label{moment2}
X_a(\chi_m^{(2)})=\begin{cases}
m & \text{if }a=0 \,, \\
(-1)^a\frac{\sin (\pi a m)}{\pi a} & \text{otherwise} \,.
\end{cases}
\end{equation}
In Figure \ref{freeen2} we show a sketch of this root configuration with the same conventions as before.

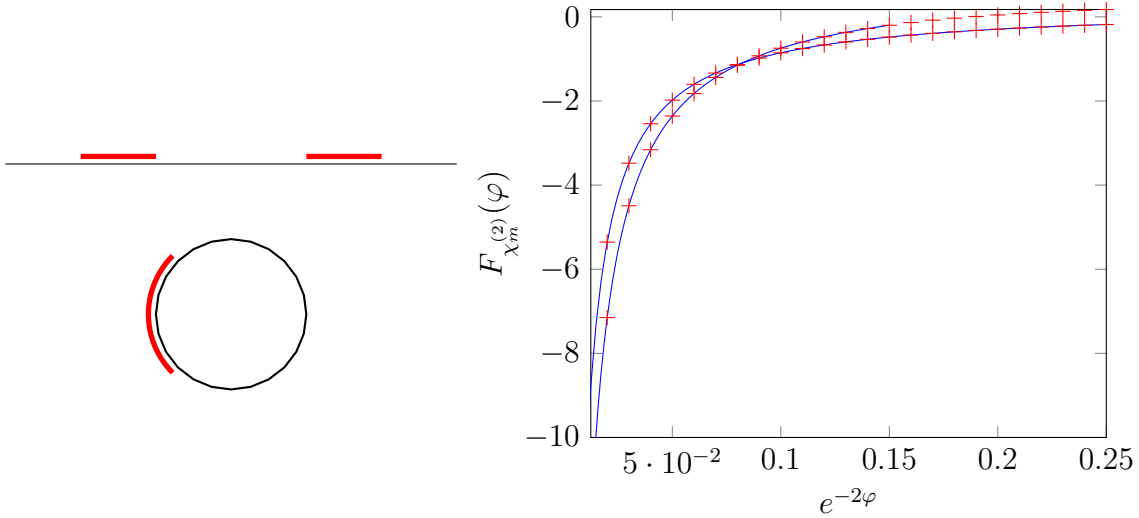
\begin{figure}[H]
\begin{tikzpicture}[scale=1]
\draw[black] (-3,2) -- (3,2);
\draw[red, thick, line width=2pt] (-2,2.1) -- (-1,2.1);
\draw[red, thick, line width=2pt] (1,2.1) -- (2,2.1);
\draw [black,thick,domain=0:360] plot ({cos(\x)}, {sin(\x)});
\draw [red,thick, line width=2pt,domain=135:225] plot ({1.1*cos(\x)}, {1.1*sin(\x)});
\filldraw[white] ({cos(-45)},{sin(-45)-2}) circle (2pt);

\end{tikzpicture}
\begin{tikzpicture}[scale=1]
\begin{axis}[
    enlargelimits=false,
    ymin=-10,
    xlabel = $e^{-2\varphi}$,
    ylabel = $F_{\chi_m^{(2)}}(\varphi)$,
     y tick label style={
        /pgf/number format/.cd,
            fixed,
            fixed zerofill,
            precision=0,
        /tikz/.cd
    }
]
\addplot [
    domain=0.01:0.25, 
    samples=100, 
    color=blue,
    ]
    {0.273679 - 0.11254/x - 0.0414095 *x + 0.153148 *x^2 + 0.0544894* x^3 - 
 0.945448 *x^4 - 3.33247 *x^5 - 0.0541935 *x^6 + 41.2164 *x^7 + 
 139.177 *x^8 - 122.571* x^9 - 2583.74* x^10 - 7281.69* x^11 + 
 16958.1 *x^12 + 189236. *x^13 + 403966. *x^14 - 1.94049*10^6 *x^15 - 
 1.49515*10^7 *x^16 - 2.04716*10^7 *x^17 + 2.1042*10^8* x^18};
 \addplot [
    domain=0.01:0.15, 
    samples=100, 
    color=blue,
    ]
    {0.797358 - 0.159155/x + 0.606471 *x - 0.669324 *x^2 - 7.66761 *x^3 + 
 27.9624 *x^4 + 137.679 *x^5 - 1237.37* x^6 - 1731.32 *x^7 + 
 54740.1 *x^8 - 82946.5 *x^9 - 2.28175*10^6* x^10 + 1.09586*10^7 *x^11 + 
 8.12174*10^7 *x^12 - 8.2764*10^8 *x^13 - 1.7244*10^9 *x^14 + 
 5.11536*10^10 *x^15 - 6.96627*10^10 *x^16 - 2.70593*10^12 *x^17 + 
 1.30358*10^13 *x^18};
        \addplot+[
    only marks,
    mark=+,
    mark size=2.9pt,
    color=red]
table{freeenergysl2c_2nd_m0v25_L240.dat};
        \addplot+[
    only marks,
    mark=+,
    mark size=2.9pt,
    color=red]
table{freeenergysl2c_2nd_m0v5_L144.dat};
\end{axis}
\end{tikzpicture}
\caption{Left: sketch of the root configuration (red) and the vacancies (black), on the real axis (top) and on the unit circle in the form $e^{2i\pi\lambda_k}$ (bottom). Right: $F_{\chi_m^{(2)}}(\varphi)$ as a function of $e^{-2\varphi}$, with $m=0.25$ (lower curve at the top right corner) and $m=0.5$ (upper curve at the top right corner), using twenty terms of \eqref{repfree} within its radius of convergence (blue) and solving numerically the Bethe equations in size $L=240$ and $L=144$ (red).}
\label{freeen2}
\end{figure}

Yet another example of a root configuration is defined by the filling function
\begin{equation}\label{eqgs}
\chi_m^{(3)}(x)=\begin{cases}
1 & \text{if }-m/4<x<m/4 \,, \\
1 & \text{if }1/2-m/2<x<1/2-m/4 \,, \\
1 & \text{if }-1/2+m/4<x<-1/2+m/2 \,, \\
0 & \text{otherwise} \,.
\end{cases}
\end{equation}
The moments read
\begin{equation}
\label{moment3}
X_a(\chi_m^{(3)})=\begin{cases}
m & \text{if }a=0 \,, \\
\frac{\sin ( \pi a m/2)}{\pi a}+\frac{(-1)^a}{\pi a}\big( \sin(\pi a m) -\sin(\pi a m/2)\big) & \text{otherwise} \,.
\end{cases}
\end{equation}
In Figure \ref{freeen3} we show a sketch of this root configuration with the same conventions as before.

\begin{figure}[H]
\begin{tikzpicture}[scale=1]
\draw[black] (-3,2) -- (3,2);
\draw[red, thick, line width=2pt] (-1.9,2.1) -- (-1.4,2.1);
\draw[red, thick, line width=2pt] (1.4,2.1) -- (1.9,2.1);
\draw[red, thick, line width=2pt] (-0.5,2.1) -- (0.5,2.1);
\draw [black,thick,domain=0:360] plot ({cos(\x)}, {sin(\x)});
\draw [red,thick, line width=2pt,domain=-22:22] plot ({1.1*cos(\x)}, {1.1*sin(\x)});
\draw [red,thick, line width=2pt,domain=135:157] plot ({1.1*cos(\x)}, {1.1*sin(\x)});
\draw [red,thick, line width=2pt,domain=203:225] plot ({1.1*cos(\x)}, {1.1*sin(\x)});
\filldraw[white] ({cos(-45)},{sin(-45)-2}) circle (2pt);

\end{tikzpicture}
\begin{tikzpicture}[scale=1]
\begin{axis}[
    enlargelimits=false,
    ymax=3,
    ymin =0,
    xmin=0,
    xlabel = $e^{-2\varphi}$,
    ylabel = $F_{\chi_m^{(2)}}(\varphi)$,
     y tick label style={
        /pgf/number format/.cd,
            fixed,
            fixed zerofill,
            precision=0,
        /tikz/.cd
    },
    xtick={0.03,0.06,0.09,0.12}
]
\addplot [
    domain=0.001:0.13, 
    samples=100, 
    color=blue,
    ]
    {0.374312 + 0.00927238/x + 0.00714426 *x + 0.6046 *x^2 - 1.79528* x^3 - 
 1.90975 *x^4 - 54.3636* x^5 - 151.756 *x^6 - 598.88 *x^7 - 1310.62* x^8 + 
 9719.72* x^9 + 32185.7 *x^10 + 445545. *x^11 + 742714.* x^12 + 
 5.35294*10^6 *x^13 + 4.97069*10^6 *x^14 + 1.61204*10^7 *x^15 + 
 5.87663*10^8 *x^16 + 1.51878*10^9 *x^17 + 2.74689*10^10* x^18};
 \addplot [
    domain=0.01:0.13, 
    samples=100, 
    color=blue,
    ]
    {0.965232 + 0.0659241/x - 0.206941 *x - 1.218 *x^2 - 19.584* x^3 - 
 25.3671* x^4 + 396.317* x^5 + 1040.29 *x^6 + 5115.68 *x^7 + 
 54915.9 *x^8 + 137017. *x^9 - 1.33218*10^6* x^10 - 3.78088*10^7* x^11 - 
 2.33221*10^8 *x^12 + 3.72594*10^8* x^13 + 5.85528*10^9 *x^14 + 
 4.12415*10^10 *x^15 + 5.86851*10^11* x^16 + 4.13647*10^12 *x^17 + 
 6.56714*10^12 *x^18};
        \addplot+[
    only marks,
    mark=+,
    mark size=2.9pt,
    color=red]
table{freeenergysl2c_3rd_m0v25_L240.dat};
        \addplot+[
    only marks,
    mark=+,
    mark size=2.9pt,
    color=red]
table{freeenergysl2c_3rd_m0v5_L144.dat};
\end{axis}
\end{tikzpicture}
\caption{Left: sketch of the root configuration (red) and the vacancies (black), on the real axis (top) and on the unit circle in the form $e^{2i\pi\lambda_k}$ (bottom). Right: $F_{\chi_m^{(3)}}(\varphi)$ as a function of $e^{-2\varphi}$, with $m=0.25$ (bottom) and $m=0.5$ (top), using twenty terms of \eqref{repfree} (blue) and solving numerically the Bethe equations in size $L=240$ and $L=144$ (red). }
\label{freeen3}
\end{figure}
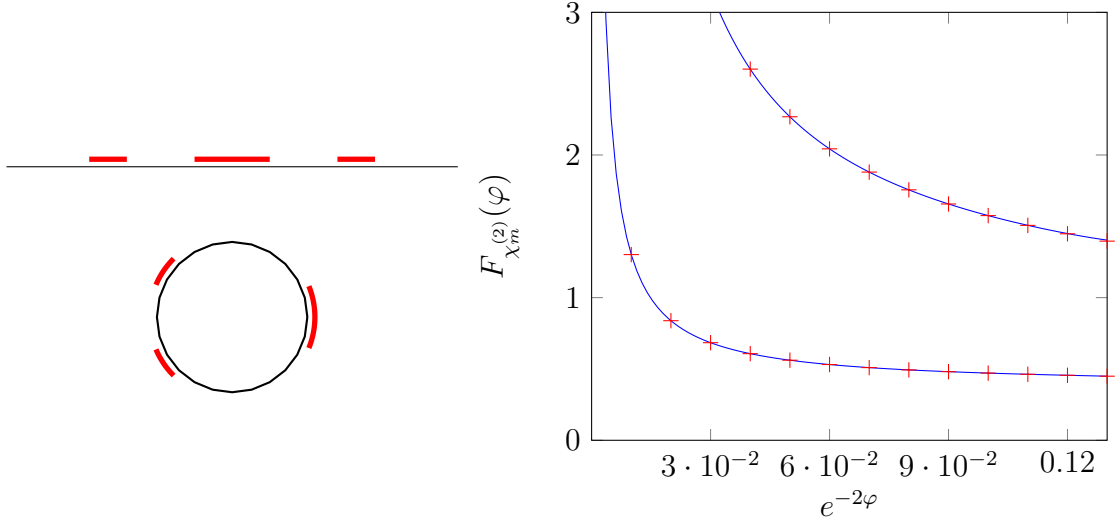

Evaluating the moments at $m=-1$, we obtain in Figure \ref{freeen78} the continuation of the energies of these states as a series in $e^{-2\varphi}$ for $\varphi\to\infty$. However, these series are not convergent at $\varphi=0$. We recall that one can solve the Bethe equations numerically only at $m>0$, whence the absence of numerical red points in Figure \ref{freeen78}. 

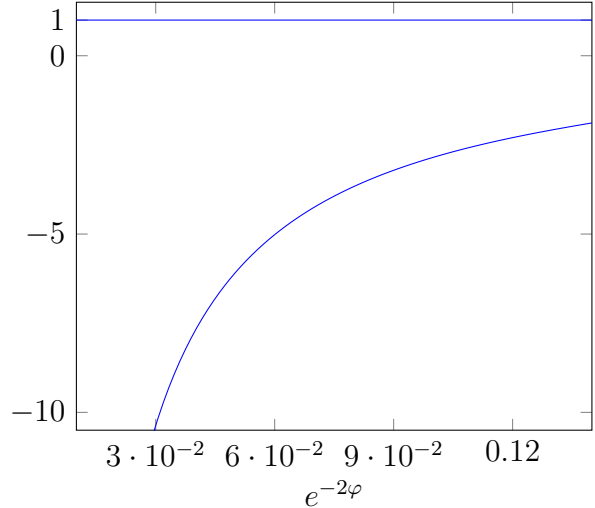
\begin{figure}[H]
\begin{center}
\begin{tikzpicture}[scale=1]
\begin{axis}[
ymax=1.5,
ymin=-10.5,
    enlargelimits=false,
    xlabel = $e^{-2\varphi}$,
     y tick label style={
        /pgf/number format/.cd,
            fixed,
            fixed zerofill,
            precision=0,
        /tikz/.cd
    },
    xtick={0.03,0.06,0.09,0.12},
    ytick={-10,-5,0,1}
]
\addplot [
    domain=0.001:0.14, 
    samples=100, 
    color=blue,
    ]
    {-0.318309886/x+0.18943053086129+1.986978760413*x-5.305387429*x^2-11.8502888270*x^3+159.6796741766*x^4-345.25260761*x^5-4137.8582552*x^6+34971.03935025*x^7+7972.8317983*x^8-1.778539900*10^6*x^9+9.4040852846*10^6*x^10+4.5539762230*10^7*x^11};
 \addplot [
    domain=0.01:0.14, 
    samples=100, 
    color=blue,
    ]
    {1+0*x};
     \end{axis}
\end{tikzpicture}
\caption{The series for $F_{\chi_m^{(i)}}(\varphi)$ as a function of $e^{-2\varphi}$, evaluated at $m=-1$, for $i=1,2$ (top, the two curves are superimposed) and $i=3$ (bottom).}
\label{freeen78}
\end{center}
\end{figure}

A few remarks are in order to summarize these three test cases.
First, we notice that the agreement between the series expansion (within its radius of convergence) and the numerical resolution of the Bethe equations (for sizes $L \gg 1$ close to the thermodynamic limit) is excellent, with the deviation between the two methods being invisible on the scale of the figures over the whole range of (convergent) $e^{-2\varphi}$ values. Second, we observe that $F_{\chi_m^{(1)}}(\varphi) > F_{\chi_m^{(3)}}(\varphi) > F_{\chi_m^{(2)}}(\varphi)$ for all $\varphi$; however, this ordering is not verified anymore at $m=-1$, which indicates that one cannot infer the root configuration of the ground state at $m=-1$ from the ordering of the states at $m>0$. Third, the series considered in the previous examples are in general not convergent down to $\varphi=0$, so that in this form they are not well suited for determining analytic continuation of energies at $\varphi=0$. And lastly, the root configurations \eqref{root1eq} and \eqref{chi-m-2} are evidently special since their energy seems to be independent of $\varphi$---we will come back to this fact in Section \ref{specialroot}.

\subsection{Comparison with a dual series expansion}
We can also give the following additional check for the energy given by the filling function $\chi_m^{(1)}$ \eqref{root1eq}. In \cite{granetjacobsensaleurxxz} we gave a way to compute recursively the coefficients $g_b(\varphi)$ of the energy $F_{\chi_m^{(1)}}(\varphi)$,
\begin{equation}
F_{\chi_m^{(1)}}(\varphi)=\sum_{b\geq 0}g_b(\varphi)m^b \,,
\end{equation}
with $g_b(\varphi)$ having an explicit dependence on $\varphi$. This kind of expansion is in a sense dual to the one in  \eqref{repfree}: it is an expansion in $m$ around $0$ with `resummed' $\varphi$-dependent coefficients, whereas \eqref{repfree} is an expansion in $e^{-2\varphi}$ around $\varphi\to\infty$ with `resummed' $m$-dependent coefficients. In the case of the filling function $\chi_m^{(1)}$ of \eqref{root1eq} we have from \eqref{expflargephi} and \eqref{moment1} on the one hand, and from \cite{granetjacobsensaleurxxz} on the other hand, the respective expansions
\begin{subequations}
\begin{eqnarray}
F_{\chi_m^{(1)}}(\varphi)&=& e^{2\varphi}\frac{\sin \pi m}{2\pi} \nonumber \\
& & \quad+m+2m^2-\frac{2\sin^2 \! \pi m}{\pi^2} \nonumber \\
& & \quad+e^{-2\varphi}\frac{\sin \pi m}{2\pi^3}(-4+\pi^2-4m\pi^2-8m^2\pi^2+4\cos(2\pi m)+2(1+4m)\sin(2\pi m)) \nonumber \\
& & \quad+O(e^{-4\varphi})\\
&=& m\times 2\cosh^2 \! \varphi \nonumber \\
& & \quad-m^3\frac{1}{6}\cosh 2\varphi \nonumber \\
& & \quad+m^4 \frac{\pi^2}{3}(1+\tanh^2 \! \varphi) \nonumber \\
& & \quad+O(m^5) \,,
\end{eqnarray}
\end{subequations}
and we can check that expanding $f_b(m)$ around $m=0$ and $g_b(\varphi)$ around $\varphi\to\infty$, we obtain two double series in $m,e^{-2\varphi}$ whose coefficients exactly match. We checked this correspondance until order $8$ in $m,e^{-2\varphi}$. This obviously provides a stringent check of both \cite{granetjacobsensaleurxxz} and Theorem~\ref{thm1} of the present paper.

\section{A special root configuration \label{specialroot}}
The coefficients of the series \eqref{repfree} can all be recursively computed, and in practice the first $\approx 20$ terms are relatively fast to calculate. One obtains  the energy levels of a state at large $\varphi$, within the ($m$ and $\chi(x)$-dependent) radius of convergence of the series \eqref{repfree}. The magnetization $m$ that is necessarily positive when solving numerically the Bethe equations in finite size, enters these series as a mere parameter that can be set to $m<0$. This permits to analytically continue the energy level of a state down to $m=-1$, at least for large $\varphi$. In this logic, a state at $m<0$ is still characterized by its moments, but those \textit{do not} derive anymore from a filling function with the constraints \eqref{fill}, but are obtained as analytic continuations of moments at $m>0$. 

In the rest of the paper, we will graphically depict a state at $m=-1$ (i.e., its moments) with the following conventions. Although its moments do not directly derive from a filling function, it may be that at $m=-1$ they can be written as $\int_{-1/2}^{1/2}f(x)e^{2i\pi a x} \, {\rm d}x$, with $f(x)$ a function. If it takes the values $-n$ with $n$ a positive or zero integer at $x$, then we depict it with a black circle with $n$ red layers at $e^{2i\pi x}$ where $f(x)=-n$. For example, Figure \ref{special} depicts the states \eqref{root1eq} and  \eqref{eqgs} at $m=-1$. We indeed have for \eqref{root1eq} at $m=-1$
\begin{equation}
\label{secondpseudo}
\begin{aligned}
X_a&=-\delta_{a,0}\\
&=-\int_{-1/2}^{1/2}e^{2i\pi a x} \, {\rm d}x\,,
\end{aligned}
\end{equation}
and for \eqref{eqgs} at $m=-1$
\begin{equation}
\label{groundstate}
\begin{aligned}
X_a&=\begin{cases}-1\quad \text{if }a=0\\
-\frac{2(-1)^b}{\pi (2b+1)}\quad \text{if }a=2b+1\, \text{odd}\\
0\quad \text{if }a=2b\neq 0\, \text{even}\\
\end{cases}\\
&=-2\int_{-1/4}^{1/4}e^{2i\pi a x} \, {\rm d}x\,.
\end{aligned}
\end{equation}
Of course, not all states can be written with a function $f(x)$ taking only integer values, but those that are relevant to us in this paper can.
\begin{figure}[H]
\begin{center}
\begin{tikzpicture}[scale=1]
\draw [black,thick,domain=0:360] plot ({cos(\x)}, {sin(\x)});
\draw [red,thick, line width=2pt,domain=0:360] plot ({1.1*cos(\x)}, {1.1*sin(\x)});
\end{tikzpicture}
\hspace{2cm}
\begin{tikzpicture}[scale=1]
\draw [black,thick,domain=0:360] plot ({cos(\x)}, {sin(\x)});
\draw [red,thick, line width=2pt,domain=-90:90] plot ({1.1*cos(\x)}, {1.1*sin(\x)});
\draw [red,thick, line width=2pt,domain=-90:90] plot ({1.2*cos(\x)}, {1.2*sin(\x)});
\end{tikzpicture}
\end{center}
\caption{Sketch of root configurations \eqref{secondpseudo} and \eqref{groundstate}, corresponding to the continuation of \eqref{root1eq} and  \eqref{eqgs} at $m=-1$.}
\label{special}
\end{figure}
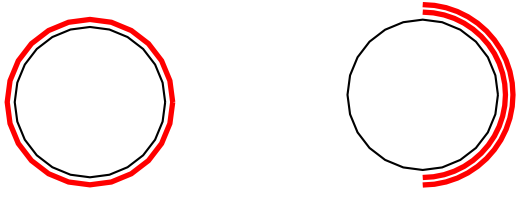

One now faces the following difficulty. Although one can obtain the energies at $m=-1$ as series in $e^{-2\varphi}$, the only value of $\varphi$ relevant to us is $\varphi=0$ (or its vicinity to obtain  derivatives), and the series \eqref{repfree} are observed to be not convergent down to $\varphi=0$, see Section \ref{examples}. Thus one would have to resum the series \eqref{repfree} in order to be able to set $\varphi=0$, which requires finding the generic explicit expression for the terms in the series. 

\subsection{Solution for a special root configuration \label{spesp}}

For an arbitrary root configuration given by an arbitrary filling function, it is evidently difficult to find a generic explicit expression for all the terms of the series \eqref{repfree}. However, in the case of the root structure \eqref{root1eq}, one can exceptionally find such a generic expression at $m=-1$. Indeed, all the moments $X_a(\chi_m)$ vanish at $m=-1$ but one, that is, $X_{0}(\chi_{-1})=-1$. As we will see, this allows us to compute all the terms in \eqref{repfree} as well as their $m$-derivatives, evaluated at $m=-1$, and this will be crucial in order to be able to continue the energy levels down to $\varphi=0$.

\subsubsection{The series \eqref{repfree} at $m=-1$}
For the root configuration given in \eqref{secondpseudo} the recurrence relation \eqref{recur} becomes
\begin{equation}
\begin{aligned}
\frac{\tilde{c}_{ab}}{2i\pi}&=\sum_{n\geq 0}\sum_{q=0}^n \sum_{\substack{b_1,b_2\geq 1\\b_1+b_2=b}}\frac{\arctan^{(n)}(0)}{n!\pi}{n\choose q}(-1)^{n-q}c_{ab_1}^{[q]}c_{0b_2}^{[n-q]}\\
&\qquad\qquad\qquad\qquad+\sum_{n\geq 1}\frac{(-1)^{n+1}}{n(2i)^{n+1}\pi}c_{ab}^{[n]}+\frac{1}{2i\pi}\sum_{n\geq 2}\frac{(-1)^n}{n}\tilde{c}_{ab}^{[n]}\\
&=\sum_{n\geq 0} \frac{\arctan^{(n)}(0)}{n!\pi}c_{ab}^{[n]}+\sum_{n\geq 1}\frac{(-1)^{n+1}}{n(2i)^{n+1}\pi}c_{ab}^{[n]}+\frac{1}{2i\pi}\sum_{n\geq 2}\frac{(-1)^n}{n}\tilde{c}_{ab}^{[n]} \,,
\end{aligned}
\end{equation}
where we used that $c_{0b_2}^{[n-q]}=0$ unless $b_2=0,n-q=0$, in which case $c_{00}^{[0]}=1$, because $c_{0b}=0$ for $b\geq 0$. Substituting back the expansion of $\arctan$ \eqref{arctan}, this equation is exactly
\begin{equation}
\arctan \left(i+\sum_{a,b\geq 1}e^{-2b\varphi} e^{\tfrac{2i\pi I_k}{L}a}c_{ab}\right)=\pi\frac{I_k}{L}+i\varphi+\arctan \left(\sum_{a,b\geq 1}e^{-2b\varphi} e^{\tfrac{2i\pi I_k}{L}a}c_{ab}\right) \,.
\end{equation}
Introducing the generating function
\begin{equation}
\gamma_0(t,x)=\sum_{a,b\geq 1}t^a x^b c_{ab}|_{m=-1} \,,
\end{equation}
where the $c_{ab}$'s are evaluated at $m=-1$, and in which one can interpret $t=e^{2i\pi I/L}$ and $x=e^{-2\varphi}$, this equation reads
\begin{equation}
\label{eqgamma0}
\arctan(i+\gamma_0(t,x))=\frac{1}{2i}\log(tx)+\arctan\gamma_0(t,x) \,,
\end{equation}
which can be solved by
\begin{equation}
\label{gamma}
\gamma_0(t,x)=-\frac{i}{2}+\frac{i}{2}\sqrt{1-\frac{8tx}{1-tx}}\equiv \Delta(tx) \,.
\end{equation}
As for equation \eqref{fbchi} for the values of the coefficients in the series \eqref{repfree}, it becomes
\begin{equation}
f_b(\chi_{-1})=-\sum_{n\geq 0}\left(\frac{(-1)^n}{ic_{11}}\tilde{c}^{[n]}_{1,b+1}+\frac{c^{[n]}_{0b}}{2(-2i)^n}\right) \,.
\end{equation}
From the solution \eqref{gamma} it follows that $c_{ab},\tilde{c}_{ab}=0$ whenever $a\neq b$. Hence the only non-vanishing term is $f_0(\chi)$ given by, with $n=1$ for the first term and $n=0$ for the second one,
\begin{equation}
\begin{aligned}
f_0(\chi_{-1})&=\frac{\tilde{c}_{11}}{ic_{11}}-\frac{1}{2}\\
&=\frac{\Delta''(0)}{2i\Delta'(0)^2}-\frac{1}{2}\\
&=1 \,,
\end{aligned}
\end{equation}
where we recall the definition \eqref{ctildedef} for $\tilde{c}_{ab}$. Hence we obtain that for this root configuration
\begin{equation}
\label{f1}
F_{\chi_{-1}}(\varphi)=1+O(e^{-2n\varphi})
\end{equation}
for any $n>0$, in the limit $\varphi\to\infty$.
Since all the expansions in $e^{-2\varphi}$ are observed to be convergent series, this equation is expected to hold at least within a finite radius of convergence near $\varphi\to\infty$. As explained hereafter in section \ref{criterion}, one needs to know the behaviour of the $m$-derivatives of $F$ at $m=-1$ in order to know the range of validity of this expression.

\medskip

Let us make the following side comment. If we keep working backwards from \eqref{eqgamma0}, we obtain that the Bethe equations for $\lambda_k=i+\gamma_0(e^{2i\pi I_k/L},e^{-2\varphi})$ are
\begin{equation}
\label{befinm1twnew}
\left(\frac{\lambda_k-i}{\lambda_k+i}\frac{\lambda_k}{\lambda_k-2i}\right)^L=e^{-\varphi L}
\end{equation}
but this identification works \textit{only} in the thermodynamic limit $L\to\infty$.

\subsubsection{The $m$-derivatives of the series \eqref{repfree} at $m=-1$ \label{derivsec}}
For the particular root structure \eqref{root1eq}, one can also compute all the $m$-derivatives of the coefficients $f_{b}(\chi_m)$ evaluated at $m=-1$. In order to show this, let us introduce the following generating functions
\begin{equation}
\label{gamap}
\gamma_p(t,x)=\sum_{a,b\geq 1}t^a x^b \frac{1}{p!} \left. \left(\frac{d}{dm}\right)^p c_{ab} \right|_{m=-1} \,.
\end{equation}
Let us consider a function $F(t)$ with a Laurent series
\begin{equation}
F(t)=\sum_{a\geq -n}F_a t^a
\end{equation}
with a certain $n$. Then in the limit $L\to\infty$, by definition \eqref{momentsdef} of the moments $X_a(\chi_m)$,
\begin{equation}
\frac{1}{L}\sum_{k}F(e^{2i\pi I_k/L})=\sum_{a\geq -n} X_a(\chi_m)F_a+O(L^{-1}) \,.
\end{equation}
For the particular root structure under consideration, one has the moments given in \eqref{moment1}. We denote $\mu=m+1$ and expand these moments around $\mu=0$. This yields for $a\neq 0$
\begin{equation}
X_a(\chi_m)= (-1)^a\sum_{p\geq 0}\frac{(\pi a)^{2p}}{(2p+1)!}(-1)^p \mu^{2p+1} \,.
\end{equation}
Hence
\begin{equation}
\label{sumfff}
\frac{1}{L}\sum_{k}F(e^{2i\pi I_k/L})=-F_0+\sum_{p\geq 0}\frac{\pi^{2p}\mu^{2p+1}}{(2p+1)!}(-1)^p \slashed{\partial}^{2p}F(-1)+O(L^{-1}) \,,
\end{equation}
where we have defined  $\slashed{\partial}_tF(t)=t\partial_t F(t)=\sum_{a\geq -n}at^a F_a$. The Bethe equations \eqref{befinm1tw} yield in the thermodynamic limit $L\to\infty$
\begin{equation}
\label{genderiv}
\begin{aligned}
&\arctan\left(i+\sum_{p\geq 0}\gamma_p(t,x)\mu^p\right) = \\
& \qquad \frac{\log(xt)}{2i}+\arctan\left(\sum_{p\geq 0}\gamma_p(t,x)\mu^p\right)\\
& \qquad -\sum_{p\geq 0}\frac{\pi^{2p}\mu^{2p+1}}{(2p+1)!}(-1)^p \slashed{\partial}_u^{2p}\arctan\left.\left(\sum_{p\geq 0}\gamma_p(t,x)\mu^p-\sum_{p\geq 0}\gamma_p(u,x)\mu^p\right)\right|_{u=-1} \,,
\end{aligned}
\end{equation}
which is the generalisation of \eqref{eqgamma0} to $\mu=m+1\neq 0$. This equation allows us to solve for $\gamma_p(t,x)$ recursively in $p$, by expressing them in terms of $\gamma_0(t,x)=\Delta(tx)$. Let us take the example of $\gamma_1(t,x)$. At order $\mu$, \eqref{genderiv} is
\begin{equation}
\begin{aligned}
&\arctan(i+\Delta(tx)+\mu \gamma_1(t,x))-\arctan(\Delta(tx)+\mu\gamma_1(t,x))\\
&=\frac{\log(xt)}{2i}-\mu \arctan(\Delta(tx)-\Delta(-x))+O(\mu^2) \,.
\end{aligned}
\end{equation}
Expanding at order $\mu$, the $\mu^0$ term vanishes due to \eqref{eqgamma0}, while the $\mu$ term gives
\begin{equation}
\label{gamma1}
\gamma_1(t,x)=\frac{\arctan(\Delta(tx)-\Delta(-x))}{\arctan'(\Delta(tx))-\arctan'(i+\Delta(tx))} \,.
\end{equation}
And in this way, one can determine all the $\gamma_p(t,x)$ recursively in terms of $\Delta$.

\medskip

As for the  energy, it reads
\begin{equation}
\label{ffderiv}
F_{\chi_{m}}(\varphi)=-e_0(\mu)+\sum_{p\geq 0}\frac{\pi^{2p}\mu^{2p+1}}{(2p+1)!}(-1)^p \slashed{\partial}_t^{2p}\left.\left[\frac{2}{1+\left(i+\sum_{p\geq 0}\mu^p\gamma_p(t,x)\right)^2}\right]\right|_{t=-1} \,,
\end{equation}
with $x=e^{-2\varphi}$, and where $e_0(\mu)$ is the term in $t^0$ in the Laurent series of $\frac{2}{1+\left(i+\sum_{p\geq 0}\mu^p\gamma_p(t,x)\right)^2}$. Hence all the $m$-derivatives of $F_{\chi_{m}}(\varphi)$ can be expressed in terms of the $\gamma_p(t,x)$ and computed explicitly.

Let us for instance compute the first derivative. We have 
\begin{equation}
\frac{2}{1+\left(i+\sum_{p\geq 0}\mu^p\gamma_p(t,x)\right)^2}=\frac{2}{\Delta(xt)(2i+\Delta(xt))}-\frac{4(i+\Delta(xt))\gamma_1(t,x)}{\Delta(xt)^2(2i+\Delta(xt))^2}\mu+O(\mu^2) \,.
\end{equation}
Using the expression \eqref{gamma1} one finds the $t^0$ term
\begin{equation}
e_0(\mu)=-1+\frac{2\mu}{1+\Delta(-x)^2}+O(\mu^2) \,.
\end{equation}
This gives 
\begin{equation}
F_{\chi_{m}}(\varphi)=-e_0(\mu)+\frac{2\mu}{\Delta(-x)(2i+\Delta(-x))}+O(\mu^2) \,.
\end{equation}
One deduces, with $x=e^{-2\varphi}$,
\begin{equation}
\label{valdev1}
\partial_m F_{\chi_m}(\varphi)|_{m=-1}=-2\cosh^2\varphi \sqrt{5-4\tanh\varphi} \,.
\end{equation}

The next terms can be computed efficiently by noting that only the knowledge of the expansion of $\gamma_p(t,x)$ for $t$ close to $0$ and $t$ close to $-1$ are actually needed to compute the successive terms. A recurrence relation is given in Appendix \ref{seriesmath}.

The important aspect of this calculation is that the computation of $\gamma_p(t,x)$ only involves $t$-derivatives of $\gamma_0(t,x)$ evaluated at $t=-1$, i.e., derivatives of $\Delta$ evaluated at $-x=-e^{-2\varphi}$, a negative real. The function $\Delta$ has no singularity for negative real (it only has a pole at $1$ and a branch point at $\tfrac{1}{9}$), and the only division is by $\arctan'(\Delta(tx))-\arctan'(i+\Delta(tx))$ which has no zeros for $t=-1$ and $0\leq x\leq 1$, so that no singularity can arise. Hence, all the $m$-derivatives of $F_{\chi_m}(\varphi)$ evaluated at $m=-1$ are regular for $0\leq \varphi <\infty$. As explained in section \ref{criterion} below, this ensures that the range of validity of the analytic continuations \eqref{f1} and \eqref{valdev1} are at least $0\leq \varphi <\infty$, which includes $\varphi=0$. Hence we obtain the analytic continuations
\begin{equation}
\label{resfirst}
F_{\chi_{m}}(\varphi=0)|_{m=-1}=1\,,\qquad \partial_m F_{\chi_{m}}(\varphi=0)|_{m=-1}=-2\sqrt{5}\,,
\end{equation}
and all the other derivatives can be analytically computed. We were able to evaluate more than $20$ terms.

\subsection{Conditions for analytically continuing series at $\varphi=0$ \label{criterion}}

\subsubsection{A counter-example and a criterion}
In section \ref{spesp} we saw that for the second pseudo-vacuum root configuration, one can compute all the terms in the series \eqref{repfree}, which yields $F_{\chi_{m}}(\varphi)=1$ at $m=-1$ within a certain radius of convergence $\varphi_c\leq \varphi<\infty$. Although the obtained (trivial) function of $\varphi$ can be obviously analytically continued to all real $\varphi$, this does not guarantee that the function $F_{\chi_{m}}(\varphi)$ will actually take these values, because analytic continuation of $F_{\chi_{m}}(\varphi)$ should be considered with respect to both variables $m$ and $\varphi$.

We can first exhibit a counterexample. The function
\begin{equation}
\label{counterx}
f(m,\varphi)=\frac{2}{\pi}\arctan \frac{\varphi-\varphi_c}{(m-m_c)^2}
\end{equation}
is analytic everywhere except at $(m,\varphi)=(m_c,\varphi_c)$. At $m=m_c$ one has
\begin{equation}
f(m_c,\varphi)=\sign(\varphi-\varphi_c) \,,
\end{equation}
whose expansion in $1/\varphi$ around $\varphi\to\infty$ is
\begin{equation}
f(m_c,\varphi)=1+O(\varphi^{-n})\qquad \forall n\geq 0 \,,
\end{equation}
that can be trivially resummed into the function $1$, which can itself be analytically continued for all $\varphi$. However, it does not correspond to the actual value of $f(m_c,\varphi)$ for $\varphi<\varphi_c$, which is $-1$. Using this function as a building block, one can obtain functions $f(m,\varphi)$ whose expansions at $m_c$ around $\varphi\to \infty$ will be perfectly regular and that can be analytically continued to all $\varphi$ without anything special happening at $\varphi=\varphi_c$, but that will actually \textit{not} be the true value of $f(m,\varphi)$ after $\varphi<\varphi_c$, which can take essentially any value.

After this sobering example we see that to have more information on the validity of $f(m_c,\varphi)=1$, one needs to know the behaviour of the same series for $f(m,\varphi)$ for $m$ close to $m_c$. In the case of \eqref{counterx}, the radius of convergence of $f(m,\varphi)$ as a series in $1/\varphi$ for $\varphi\to\infty$ is larger than $\frac{1}{\varphi_c+(m-m_c)^2}$ for $m\neq m_c$. For $m\to m_c$ we only know it is larger than $\frac{1}{\varphi_c}$, and we have indeed $f(m_c,\varphi)=1$ for $\frac{1}{\varphi}<\frac{1}{\varphi_c}$. In any case, the resummed value of $f(m_c,\varphi)$ has to be correct within the radius $\underset{m\to m_c}{\lim}\rho(m)$ where $\rho(m)$ is the radius of convergence as a function of $m$. But the radius of convergence for $m$ close to $m_c$ gives  too strong a constraint for the validity of the analytic continuation is general. For example, the function
\begin{equation}
\tilde{f}(m,\varphi)=\frac{1}{1+a\varphi}+\frac{2}{\pi}\arctan \frac{\varphi-\varphi_c}{(m-m_c)^2}
\end{equation}
for $a>0$ can have a radius of convergence when $m\to m_c$ for the series in $1/\varphi$ around $\varphi\to\infty$ arbitrarily small provided $a$ is sufficiently large, whereas the analytic continuation of the series will work down to $\varphi>\varphi_c$ for any $a>0$. 

To find a sensible constraint on the range of validity of the analytic continuation with respect to $\varphi$, one can make the following reasoning. To analytically continue a function $f(m_c,\varphi)$ on $\varphi\in]\varphi_c,+\infty[$, one needs that $f(m,\varphi)$,  considered as a function of two variables $(m,\varphi)$, is analytic in a domain of $(m,\varphi)\in \mathbb{C}\times \mathbb{C}$ strictly containing $\{m_c\}\times ]\varphi_c,\infty[$. This implies in particular that none of the derivatives with respect to $m$ at $m_c$ is singular for any $\varphi\in]\varphi_c,+\infty[$, but also that the radius of convergence of the series in $m$ is non-zero for all $\varphi_c<\varphi<\infty$. In the case of the example \eqref{counterx}, we have
\begin{equation}
\partial^2_m f(m_c,\varphi)=-\frac{4}{\pi(\varphi-\varphi_c)} \,,
\end{equation}
which is regular for $\varphi_c<\varphi<+\infty$ but singular at $\varphi_c$, and indeed its analytic continuation $f(m_c,\varphi)=1$ is valid only for $\varphi_c<\varphi<+\infty$.

A counterexample where all the $m$-derivatives are regular for $\varphi_c<\varphi<\infty$, but whose series in $m$ has a zero radius of convergence beyond some value of $\varphi$ is
\begin{equation}
f(m,\varphi)=\sum_{p\geq 1}\left(\frac{2\varphi_c}{\varphi}\right)^{p^2}(m-m_c)^p\,.
\end{equation}
Indeed, for any $m\neq m_c$ the series in $\varphi$ cannot be analytically continued for $\varphi<2\varphi_c$, since it is well known that $\sum_{n\geq 1}x^{n^2}$ has a natural boundary on the unit circle.\\

Hence to be able to perform these analytic continuations we need the following
\begin{prop}\label{convcriterion}
Let $f(x,y)$ be a function of two real variables defined and analytic in a neighbourhood of $(0,0)$, hence with an expansion
\begin{equation}
f(x,y)=\sum_{n\geq 0}x^n f_n(y) \,,
\end{equation}
with $f_n(y)$ analytic functions of $y$ in a neighbourhood of $0$.
\begin{enumerate}
\item If for all $p\geq 0$, there exists a function $g_p(x)$, analytic on $[0,1]$, whose expansion around $0$ is $\sum_{n\geq 0}x^n f_n^{(p)}(0)$,
\item and if the series $\sum_{p\geq 0}\frac{y^p}{p!}g_p(x)$ has a non-zero radius of convergence for all $0\leq x\leq 1$,
\end{enumerate}
then $f(x,y)$ can be analytically continued to a function that takes the values $g_0(x)$ on $[0,1]\times \{0\}$ (and whose $p$-th $y$-derivative takes the values $g_p(x)$ on $[0,1]\times \{0\}$).
\end{prop}
The proof is elementary---the point of this proposition is to avoid reaching naive conclusions, as illustrated by the examples shown above.
\begin{proof}
Defining $r(x)>0$, the radius of convergence of $\sum_{p\geq 0} \tfrac{y^p}{p!}g_p(x)$, we have $r_0\equiv \underset{0\leq x \leq 1}{\min}r(x)>0$. Then the function $\tilde{f}(x,y)=\sum_{p\geq 0} \tfrac{y^p}{p!}g_p(x)$ defined on $[0,1]\times ]-r_0,r_0[$ is analytic and coincides with $f(x,y)$ on an open non-empty set, so that it is the analytic continuation of $f(x,y)$ on $[0,1]\times ]-r_0,r_0[$, and we have $\partial_y^p \tilde{f}(x,0)=g_p(x)$.
\end{proof}
In the following, the first hypothesis of this proposition will be verified analytically. The second hypothesis will however be verified only numerically (leaving in many cases almost no doubt about its validity, for example when we have $\approx 15$ terms in the series).

\subsubsection{Direct numerical check at $m=1$ \label{secm1}}
Let us give a numerical check of this criterion in a situation very close to the one in section \ref{analy}. Considering the same root configuration as in section \ref{spesp} but for $m=1$, we have the moments
\begin{equation}
X_a(\chi_{1})=\begin{cases}
1 & \text{if }a=0 \,, \\
0 & \text{otherwise} \,,
\end{cases}
\end{equation}
which also simplifies greatly the recurrence relations as in the case $m=-1$. Similarly, one can show that the generating function $\gamma_0(t,x)$ satisfies then
\begin{equation}
\arctan(i+\gamma_0(t,x))=\frac{1}{2i}\log(tx)-\arctan\gamma_0(t,x) \,,
\end{equation}
which can be solved by
\begin{equation}
\label{gamma2}
\gamma_0(t,x)=\frac{i}{2}\frac{1+3tx}{1-tx}\left[1-\sqrt{1+\frac{8tx(1-tx)}{(1+3tx)^2}} \right]\equiv \Delta(t x) \,.
\end{equation}
Then in the series \eqref{repfree} one has
\begin{equation}
\label{3}
F_{\chi_1}(\varphi)=3+O(e^{-2n\varphi})
\end{equation}
for any $n>0$. As in the case $m=-1$, one can compute the $m$-derivatives of $F_{\chi_m}(\varphi)$ at $m=1$. For example
\begin{equation}
\label{eqm1dm}
\partial_m F_{\chi_m}(\varphi)|_{m=1}=2\cosh^2 \! \varphi \: \frac{\cosh2\varphi-7-2\sqrt{2}\sinh\varphi\sqrt{\cosh2\varphi-7}}{5\cosh2\varphi-11-4\sqrt{2}\sinh\varphi\sqrt{\cosh2\varphi-7}} \,.
\end{equation}
Generically, it will involve $\gamma_0(-1,e^{-2\varphi})$ as in the case $m=-1$. But in this case, $\gamma_0(-1,e^{-2\varphi})$ has a singularity at
\begin{equation}
\varphi_c=\frac{1}{2}\log(7+4\sqrt{3})>0 \,.
\end{equation}
Hence \eqref{3} and all the resummed values for the $m$-derivatives of $F_{\chi_m}(\varphi)$ at $m=1$ \eqref{eqm1dm} will work only for $\varphi_c<\varphi<+\infty$, in particular not at $\varphi=0$. The advantage of $m=1$ is that one can solve the Bethe equations in finite size and directly check this affirmation numerically. One indeed obtains Figure \ref{m1}, in agreement with Proposition \ref{convcriterion}.

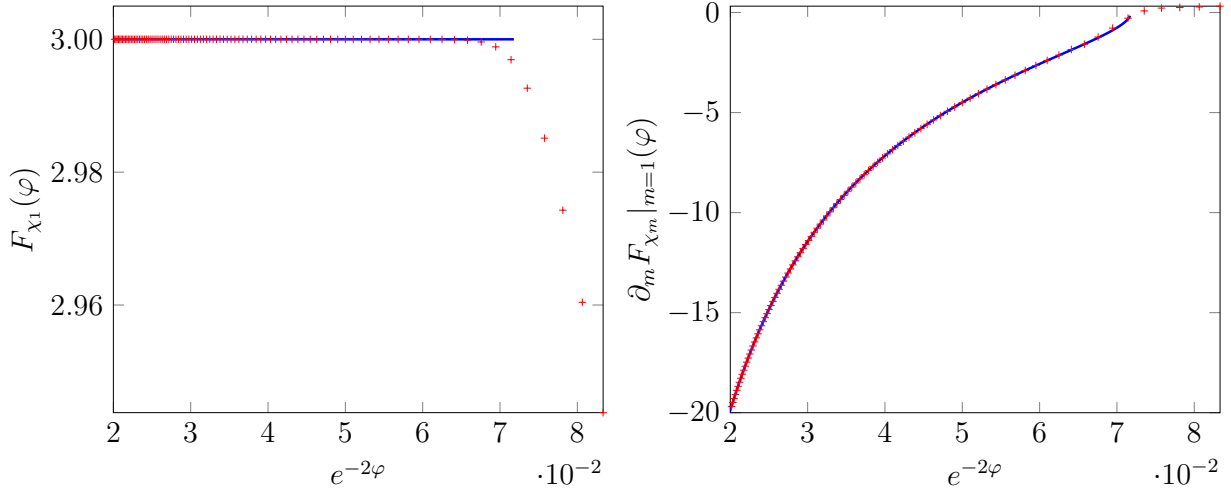
\begin{figure}[H]
\begin{tikzpicture}[scale=0.95]
\begin{axis}[
    enlargelimits=false,
    xlabel = $e^{-2\varphi}$,
    ylabel = $F_{\chi_1}(\varphi)$,
    ymax=3.005,
    xmin=0.02,
     y tick label style={
        /pgf/number format/.cd,
            fixed,
            fixed zerofill,
            precision=2,
        /tikz/.cd
    }
]
\addplot [
    domain=0:0.071796, 
    samples=100, 
        line width=1pt,
    color=blue,
    ]
    {3};
        \addplot+[
    only marks,
    mark=+,
    mark size=1.5pt,
    color=red]
table{sl2m1.dat};

\end{axis}
\end{tikzpicture}
\begin{tikzpicture}[scale=0.95]
\begin{axis}[
    enlargelimits=false,
    xlabel = $e^{-2\varphi}$,
    ylabel = $\partial_mF_{\chi_m}|_{m=1}(\varphi)$,
    xmin=0.02,
    ymin=-20,
     y tick label style={
        /pgf/number format/.cd,
            fixed,
            fixed zerofill,
            precision=0,
        /tikz/.cd
    }
]
        \addplot+[
        no marks,
        line width=1pt,
    color=blue]
table{sl2m1dmth.dat};
        \addplot+[
    only marks,
    mark=+,
    mark size=1.5pt,
    color=red]
table{sl2m1dm.dat};

\end{axis}
\end{tikzpicture}

\caption{$F_{\chi_1}(\varphi)$ (left panel) and $\partial_m F_{\chi_m}(\varphi)|_{m=1}$ (right panel) as functions of $e^{-2\varphi}$ for $\varphi_c\leq \varphi<\infty$. We show numerical values obtained from solving the Bethe equations (red crosses) and analytic values \eqref{3}--\eqref{eqm1dm} (blue curves).}
\label{m1}
\end{figure}

\subsubsection{Numerical check at $m=-1$: analytic continuation of the energy at $\varphi=0$}
As explained in \cite{granetjacobsensaleurxxz}, around $m=0$ one can efficiently expand $F_{\chi_m^{(1)}}(\varphi)$ in powers of $m$:
\begin{equation}
\label{exparoundm0}
\begin{aligned}
F_{\chi_m^{(1)}}(\varphi=0)&=2m-\frac{\pi^2}{6}m^3+\frac{\pi^2}{3}m^4+\frac{-60\pi^2+\pi^4}{120}m^5+\left(\frac{2\pi^2}{3}-\frac{11\pi^4}{180}\right)m^6\\
&+\left(-\frac{5\pi^2}{6}+\frac{2\pi^4}{9}-\frac{\pi^6}{5040}\right)m^7+\left(\pi^2-\frac{7\pi^4}{12}+\frac{31\pi^6}{2520}\right)m^8\\
&+O(m^9) \,,
\end{aligned}
\end{equation}
as well as with the twist $\varphi$:
\begin{equation}
\label{aroundm0}
\begin{aligned}
F_{\chi_m^{(1)}}(\varphi)&=2m  \cosh^2\varphi-\frac{\pi^2}{6}\cosh(2\varphi)m^3+\frac{\pi^2}{3}(1+\tanh^2\varphi)m^4+O(m^5) \,.
\end{aligned}
\end{equation}
This matches the numerical solution of the $s=-1$ Bethe equations \eqref{befin} in large size $L$, obviously for a positive number of roots $N$, hence $m\geq 0$. In the derivation of these series in \cite{granetjacobsensaleurxxz}, the $m$-dependence comes from sums over the Bethe numbers of these root configurations
\begin{equation}
\frac{1}{L}\sum_{i=1}^N \left(\frac{I_i}{L}\right)^a=\begin{cases}
\frac{m^{a+1}}{a+1} & \text{if }a\text{ is even}\\
0 & \text{if }a\text{ is odd}\\
\end{cases}+O(L^{-1}) \,.
\end{equation}
Changing $m$ into $-m$ corresponds to placing a minus sign in front of every sum over Bethe numbers, hence to inverting the right-hand side of the Bethe equations \eqref{befin}. This is exactly equivalent to changing $s$ into $-s$. Since the expansions \eqref{aroundm0} hold only in the thermodynamic limit, this correspondence also holds only in the thermodynamic limit. Hence the free energy \eqref{aroundm0} for $m<0$ corresponds to the free energy of the $s=1$ Bethe equations for $|m|>0$ in the thermodynamic limit, with the same root configuration. See Figure \ref{checkexpansionneg} for the numerical verification of this fact. However, such root configuration for $s=1$ is valid only for $0\leq |m| \leq 1/2$, hence one cannot reach $m=-1$ with this technique. Moreover, at $\varphi=0$ the expansion is observed to have a radius of convergence $\approx 0.3$, which is not even enough to reach the limit point $m=-1/2$.

\begin{figure}[H]
\begin{center}
\begin{tikzpicture}[scale=1]
\begin{axis}[
    enlargelimits=false,
    ymin=-2.2,
    ymax=0,
    xlabel = $m$,
    ylabel = $F_{\chi_m^{(1)}}(\varphi)$,
     y tick label style={
        /pgf/number format/.cd,
            fixed,
            fixed zerofill,
            precision=1,
        /tikz/.cd
    }
]
\addplot+[
    only marks,
    mark=+,
    mark size=2.9pt,
    color=red]
table{freeenergysl2c_phi0.dat};
        \addplot+[
    only marks,
    mark=+,
    mark size=2.9pt,
    color=red]
table{freeenergysl2c_phi0v5.dat};
        \addplot+[
    only marks,
    mark=+,
    mark size=2.9pt,
    color=red]
table{freeenergysl2c_phi0v75.dat};
        \addplot+[
    only marks,
    mark=+,
    mark size=2.9pt,
    color=red]
table{freeenergysl2c_phi1.dat};
        \addplot+[
    only marks,
    mark=+,
    mark size=2.9pt,
    color=red]
table{freeenergysl2c_phi1v25.dat};
\addplot [
    domain=-0.3:-0, 
    samples=100, 
    line width=1pt,
    color=blue
    ]
    {2.* x - 1.64493 *x^3 + 3.28987 *x^4 - 4.12306 *x^5 + 0.626958 *x^6 + 
 13.231* x^7 - 35.1258* x^8 + 25.0563* x^9 + 110.675 *x^10 - 
 381.759* x^11 + 244.465* x^12 + 1734.06* x^13 - 5617.02 *x^14 + 
 2031.11 *x^15 + 32573.9 *x^16 - 90575.* x^17 - 14995.7* x^18};
\addplot [
    domain=-0.4:-0, 
    samples=100, 
    line width=1pt,
    color=blue
    ]
    {2.543080635 *x - 2.53826590 *x^3 + 3.99242693* x^4 - 3.4571686* x^5 - 
 0.618722 *x^6 + 10.059257 *x^7 - 19.89481 *x^8 + 8.6068 *x^9 + 
 55.267* x^10 - 154.347 *x^11 + 105.66* x^12 + 445.6 *x^13 - 1467. *x^14 + 
 1151. *x^15 + 4.59*10^3 *x^16 - 1.61*10^4* x^17 + 1.3*10^4 *x^18};
\addplot [
    domain=-0.45:-0, 
    samples=100, 
    line width=1pt,
    color=blue
    ]
    {3.352409615 *x - 3.8695587 *x^3 + 4.6170476* x^4 - 2.222147* x^5 - 
 2.16528 *x^6 + 8.39427 *x^7 - 12.0010* x^8 + 2.634* x^9 + 27.18 *x^10 - 
 60.20* x^11 + 32.5* x^12 + 137.* x^13 - 3.9*10^2 *x^14 + 3.*10^2 *x^15 + 
 8.*10^2 *x^16 - 3.*10^3 *x^17 + 0.*10^3 *x^18};
\addplot [
    domain=-0.5:-0, 
    samples=100, 
    line width=1pt,
    color=blue
    ]
    {4.762195691 *x - 6.1885639 *x^3 + 5.1980761* x^4 - 0.220654 *x^5 - 
 3.89044 *x^6 + 7.0360* x^7 - 6.4549* x^8 - 1.391 *x^9 + 15.31 *x^10 - 
 24.0 *x^11 + 8.2* x^12 + 42. *x^13 - 9.*10^1* x^14};
\addplot [
    domain=-0.5:-0, 
    samples=100, 
    line width=1pt,
    color=blue
    ]
    {7.132289480 *x - 10.08721187* x^3 + 5.657208330 *x^4 + 2.598291560 *x^5 - 
 5.398772790 *x^6 + 5.298818010 *x^7 - 2.098971390* x^8 - 
 4.078504650 *x^9 + 9.958734840* x^10 - 9.444429250* x^11 - 
 1.664262020 *x^12 + 19.38846740* x^13 - 28.23076612 *x^14 + 
 9.226442070* x^15 + 39.11473999 *x^16 - 80.05985420 *x^17 + 
 48.07928020 *x^18};
\end{axis}
\end{tikzpicture}
\end{center}
\caption{$F_{\chi_m^{(1)}}(\varphi)$ as a function of $m$, obtained by solving numerically the Bethe equations \eqref{befin} at $s=1$ with $N=L |m|$ roots for $m>-1/2$ with $N=80$ roots (red crosses), and by expanding around $m=0$ within the radius of convergence with equation \eqref{aroundm0} (blue curves), for different values of $\varphi \equal 0,\, 0.5,\, 0.75,\, 1,\, 1.25$ from top to bottom inside the panel.}
\label{checkexpansionneg}
\end{figure}
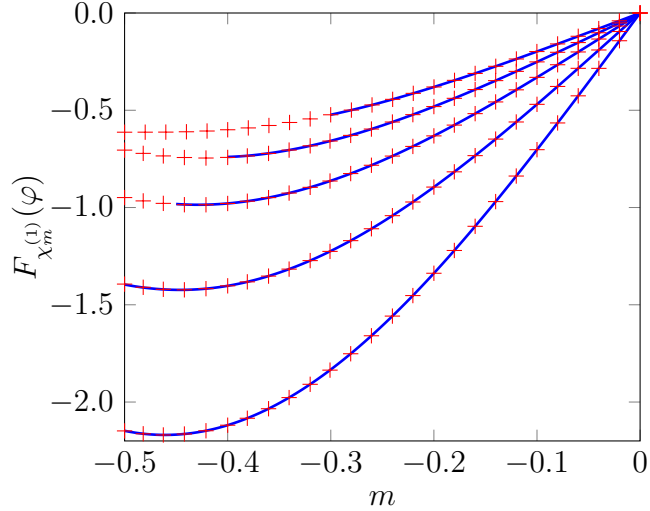

The results of section \ref{spesp} show that, remarkably, one can expand $F_{\chi_m^{(1)}}(\varphi=0)$ around $m=-1$ and compute recursively all the coefficients of the expansion, as around $m=0$ in \eqref{exparoundm0}. The coefficients of the expansion read
\begin{equation}
\label{aroundm-1}
\begin{aligned}
F_{\chi_m^{(1)}}(\varphi=0)&=1-2\sqrt{5}(m+1)+\frac{23\pi^2}{30\sqrt{5}}(m+1)^3+\frac{23\pi^2}{75}(m+1)^4\\
&+\left(\frac{23\pi^2}{50\sqrt{5}}-\frac{109\pi^4}{3000\sqrt{5}}\right)(m+1)^5+\left(\frac{46\pi^2}{375}-\frac{59\pi^4}{2500}\right)(m+1)^6\\
&+\left(\frac{23\pi^2}{150\sqrt{5}}-\frac{533\pi^4}{11250\sqrt{5}}+\frac{359\pi^6}{393750\sqrt{5}}\right)(m+1)^7\\
&+\left(\frac{23\pi^2}{625}-\frac{189\pi^4}{12500}+\frac{427\pi^6}{562500}\right)(m+1)^8+O((m+1)^9) \,.
\end{aligned}
\end{equation}
This result can be compared to the numerics by checking that this expansion around $m=-1$ matches the values for $m>-1/2$ in \eqref{exparoundm0}, as shown in Figure \ref{checkexpansion}. But one can go further and obtain the $\varphi$-dependence of $F_{\chi_m}(\varphi)$ when expanded around $m=-1$. The first terms read
\begin{equation}
\label{aroundm-1twist}
\begin{aligned}
F_{\chi_m^{(1)}}(\varphi)&=1-2\cosh^2\varphi \sqrt{5+4\tanh\varphi}(m+1)+O((m+1)^3) \,.
\end{aligned}
\end{equation}
This can again be compared with the numerics at $m=-1/2$, where one can simply solve numerically \eqref{befin} for $s=1$ and $N=L|m|=L/2$ for large $L$, with a twist $\varphi$ as in equation \eqref{befinm1tw}, see Figure \ref{checkexpansion}. 

All these expansions and their agreement with the numerics give a very strong check of our method.

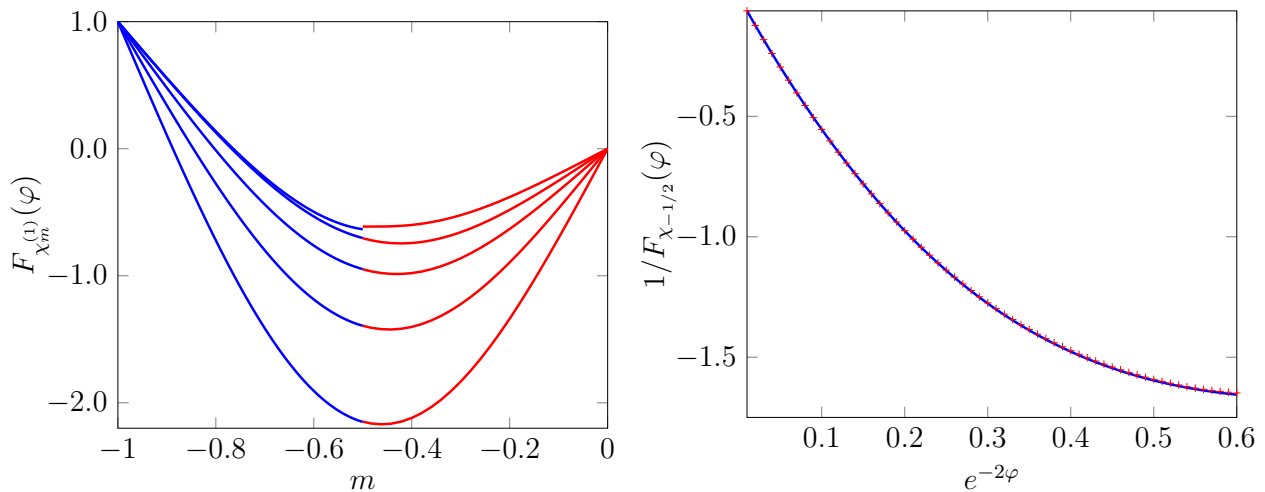
\begin{figure}[H]
\begin{tikzpicture}[scale=0.95]
\begin{axis}[
    enlargelimits=false,
    ymin=-2.2,
    xlabel = $m$,
    ylabel = $F_{\chi_m^{(1)}}(\varphi)$,
     y tick label style={
        /pgf/number format/.cd,
            fixed,
            fixed zerofill,
            precision=1,
        /tikz/.cd
    }
]
        \addplot+[
    no marks,
    line width=1pt,
    color=red]
table{freeenergysl2c_phi0.dat};
        \addplot+[
    no marks,
    line width=1pt,
    color=red]
table{freeenergysl2c_phi0v5.dat};
        \addplot+[
    no marks,
    line width=1pt,
    color=red]
table{freeenergysl2c_phi0v75.dat};
        \addplot+[
    no marks,
    line width=1pt,
    color=red]
table{freeenergysl2c_phi1.dat};
        \addplot+[
    no marks,
    line width=1pt,
    color=red]
table{freeenergysl2c_phi1v25.dat};
\addplot [
    domain=-1:-0.5, 
    samples=100, 
    line width=1pt,
    color=blue
    ]
    {1-4.472135954999*(x+1)+3.383929640619*(x+1)^3+3.0266786830007*(x+1)^4+0.447581*(x+1)^5-1.08986*(x+1)^6-0.995115*(x+1)^7-0.379823*(x+1)^8+0.141345*(x+1)^9+0.488351*(x+1)^10+0.748754*(x+1)^11+1.069798*(x+1)^12+1.55923 *(x+1)^13+2.32722 *(x+1)^14+3.53458*(x+1)^15+5.40816*(x+1)^16+8.30085 *(x+1)^17+12.7474*(x+1)^18+19.5391 *(x+1)^19+29.8324*(x+1)^20+45.2627*(x+1)^21+68.0469 *(x+1)^22};
\addplot [
    domain=-1:-0.5, 
    samples=100, 
    line width=1pt,
    color=blue,
    ]
    {1 - 4.51461737 *(1 + x) +  3.1872324 *(1 + x)^3 + 
 2.8239225 *(1 + x)^4 + 0.358950 *(1 + x)^5 - 1.213036 *(1 + x)^6 - 
 1.25476 *(1 + x)^7 - 0.77450* (1 + x)^8};
\addplot [
    domain=-1:-0.5, 
    samples=100, 
    line width=1pt,
    color=blue,
    ]
    {1 - 6.65621987* (1 + x) + 6.3513037* (1 + x)^3 + 
 3.8167631 *(1 + x)^4 - 1.486199 *(1 + x)^5 - 2.460975 *(1 + x)^6 - 
 1.4275 *(1 + x)^7 - 0.5829 *(1 + x)^8};
\addplot [
    domain=-1:-0.5, 
    samples=100, 
    line width=1pt,
    color=blue,
    ]
    {1 - 5.25741224* (1 + x) + 4.2271472 *(1 + x)^3 + 
 3.2161429* (1 + x)^4 - 0.261304 *(1 + x)^5 - 1.704217 *(1 + x)^6 - 
 1.35686 *(1 + x)^7 - 0.77509* (1 + x)^8};
\addplot [
    domain=-1:-0.5, 
    samples=100, 
    line width=1pt,
    color=blue,
    ]
    {1 - 9.04104679 *(1 + x) + 10.1323066 *(1 + x)^3 + 
 4.4828024 *(1 + x)^4 - 3.609027 *(1 + x)^5 - 3.443460 *(1 + x)^6 - 
 1.2598* (1 + x)^7 - 0.0406* (1 + x)^8};

\end{axis}
\end{tikzpicture}
\begin{tikzpicture}[scale=0.95]
\begin{axis}[
    enlargelimits=false,
    ymin=-1.75,
    xmax=0.6,
    xlabel = $e^{-2\varphi}$,
    ylabel = $1/F_{\chi_{-1/2}}(\varphi)$,
     y tick label style={
        /pgf/number format/.cd,
            fixed,
            fixed zerofill,
            precision=1,
        /tikz/.cd
    }
]
        \addplot+[
  only marks,
    mark=+,
    mark size=1.5 pt,
    color=red]
    table[
           x expr=\thisrowno{0}, 
           y expr=1/\thisrowno{1} 
         ]
{freeenergysl2c_phi_mm0v5_L200.dat};
        \addplot+[
  no marks,
  line width=1 pt,
    color=blue]
    table[
           x expr=\thisrowno{0}, 
           y expr=1/\thisrowno{1} 
         ]
{sl2rvaluestwistth.dat};
%
\end{axis}
\end{tikzpicture}
\caption{Left: $F_{\chi_m^{(1)}}(\varphi)$ as a function of $m$, obtained by solving numerically the Bethe equations \eqref{befin} at $s=1$ with $N=L |m|$ roots for $m>-1/2$ with $N=80$ roots (red), and by expanding around $m=-1$ with equation \eqref{aroundm-1} (blue), for different values of $\varphi \equal 0,\, 0.5,\, 0.75,\, 1,\, 1.25$ from top to bottom inside the panel. Right: $1/F_{\chi_{-1/2}}(\varphi)$ as a function of $e^{-2\varphi}$, obtained by solving numerically the Bethe equations \eqref{befin} at $s=1$ with the twist with $N=L |m|$ roots for $m=-1/2$ (red), and by evaluating the expansion \eqref{aroundm-1twist} at $m=-1/2$ (blue).}
\label{checkexpansion}
\end{figure}

\section{Exploring the spectrum \label{explore}}
The values \eqref{resfirst} and \eqref{aroundm-1} are non-trivial results since they constitute the analytic continuation of a function of $m$ around $m=-1$, whereas its definition is for $m\geq 0$, and its natural expansion is around $m=0$. Their calculation relied on the fact that for a very particular root configuration $\chi_m$ all the moments $X_a(\chi_m)$ except one vanish at $m=-1$, which allows one to compute all the coefficients of the series involved in the Bethe root $\lambda_k$, or the generating function $\gamma_p(t,x)$. However, because of that reason, the state considered is very particular and in the limit $\varphi\to \infty$ it is \textit{not} the ground state, and nothing guarantees that the root structure of the ground state allows for the same mechanism. 

To explore the rest of the spectrum, i.e. to compute the analytic continuation of energies at $m=-1$ whose moments are \textit{not} given by \eqref{secondpseudo}, we proceed as follows. We consider a {\sl trajectory} $\xi\to \{X_a^\xi\}_a$ in the space of moments at $m=-1$, such that at $\xi=0$, $\{X_a^{\xi=0}\}_a$ are the moments of the special root configuration \eqref{secondpseudo}. By this, we mean the analytic continuation at $m=-1$ of the moments of a family of filling functions $\xi\to \chi^\xi_m(x)$ for $m>0$. The idea is that, in the same way that all the $m$-derivatives of $F_{\chi_m}(\varphi)$ can be evaluated at $m=-1$ for this special root configuration, all the $\xi$-derivatives of the energy of state with moments $\{X_a^\xi\}_a$ can be evaluated at $\xi=0$, whenever $\{X_a^{\xi=0}\}_a$ are the moments of the special root configuration \eqref{secondpseudo}. 

\subsection{Expanding along a trajectory \label{expandtraj}}
To that end, we expand the moments in terms of $\xi$ along the trajectory
\begin{equation}
X_a^\xi=\sum_{p\geq 0}\xi^p X_{a,p} \,.
\end{equation}
By construction, we have $X_{0,0}=-1$ and $X_{a,0}=0$ for $a\neq 0$. Following section \ref{derivsec} for the $m$-derivatives, we have for a function $F(t)$ with a Laurent series at $t=0$
\begin{equation}
\frac{1}{L}\sum_{k}F(e^{2i\pi I_k/L})=\sum_{p\geq 0}\xi^p \Xi^p_t[F(t)]+O(L^{-1}) \,,
\end{equation}
where we introduced the operator $\Xi^p_t[F(t)]$ that takes a function of $t$ and  returns the following complex number
\begin{equation}
\label{xip}
\Xi^p_t[F(t)]=\sum_{a\geq -n} X_{a,p}F_a \,.
\end{equation}
The index $t$ merely indicates the dummy variable on which $\Xi^p_t$ acts. By construction, $\Xi^0_t[F(t)]=-F_0$. We introduce the generating functions
\begin{equation}
\gamma_p(t,x)=\sum_{a,b\geq 1}t^a x^b \frac{1}{p!}\left.\left(\frac{d}{d\xi}\right)^p c_{ab}\right|_{\xi=0} \,.
\end{equation}
The coefficients $c_{ab}$ indeed now depend on $\xi$ on the trajectory.  Again, by construction $\gamma_0(t,x)$ is given by \eqref{gamma}. The other $\gamma_p(t,x)$ satisfy an equation analogous to \eqref{genderiv}
\begin{equation}
\label{rectraj}
\begin{aligned}
\arctan\left(i+\sum_{p\geq 0}\gamma_p(t,x)\xi^p\right)&=\frac{\log(xt)}{2i}+\arctan \left(\sum_{p\geq 0}\gamma_p(t,x)\xi^p\right)\\
&-\sum_{p\geq 1}\xi^p \Xi^p_u\left[\arctan \left(\sum_{p\geq 0}(\gamma_p(t,x)-\gamma_p(u,x))\xi^p\right) \right] \,.
\end{aligned}
\end{equation}
This equation again allows us to solve for all the $\gamma_p(t,x)$ recursively. For example
\begin{equation}
\gamma_1(t,x)=\frac{\Xi^1_u[\arctan(\Delta(tx)-\Delta(ux))]}{\arctan'(\Delta(tx))-\arctan'(i+\Delta(tx))} \,.
\end{equation}
As for the  energy $F_\xi(\varphi)$, it reads
\begin{equation}
\label{fetraj}
F_\xi(\varphi)=\sum_{p\geq 0}\xi^p \Xi_t^p\left[\frac{2}{1+\left(i+\sum_{q\geq 0}\gamma_q(t,x)\xi^q \right)^2} \right] \,,
\end{equation}
with $x=e^{-2\varphi}$. Hence all the $\xi$-derivatives of $F_\xi(\varphi)$ can be expressed in terms of the $\gamma_p(t,x)$ and computed explicitly.

\medskip

We can now justify the use of the term `pseudo-vacuum' for the state \eqref{secondpseudo} at $m=-1$. It indeed shares remarkable properties with the usual pseudo-vacuum defined by having no Bethe roots (hence that is at $m=0$). First, the energies of these two states are both independent of the twist $\varphi$, which is never true for a generic root configuration. Second, and most importantly, one can compute the energy of any state whose root configuration is close to them: indeed, at $m$ close to zero the Bethe equations decouple and one can always solve for the Bethe roots, while for the second pseudo-vacuum we saw that one can calculate the perturbation of its energy along a trajectory. This means that the energy levels of the spin chains can be explored from the usual pseudo-vacuum as well as from this other pseudo-vacuum.

We note that the crucial ingredient for this other pseudo-vacuum to exist (i.e., for trajectories to be expandable around it) is the absence of singularities of the `kernel' $\gamma_0(-1,x)=\Delta(-x)$ in \eqref{gamma} for $0\leq x=e^{-2\varphi} \leq 1$, which allows us to analytically continue up to $\varphi=0$. At $m=1$, for example, we saw in Section \ref{secm1} that there is also a special root configuration for which the energy can be computed, but then \eqref{gamma2} \textit{has} a singularity for $0\leq x \leq 1$, so that these energies \textit{cannot} be continued to $\varphi=0$. Hence this state at $m=1$ cannot be considered as another pseudo-vacuum.

One should also note that this construction of another pseudo-vacuum is not an exceptional feature of the $s=-1$ chain. The same reasoning can indeed be performed for the usual $s=1/2$ Heisenberg chain, whose usual pseudo-vacuum is one of the two ferromagnetic ground states where all the spins are either up or down. Conventionally the $m=0$ state $| \! \Uparrow \rangle$ is taken as the pseudo-vacuum, so the other $m=1$ state $| \! \Downarrow \rangle$ is what we would call the second pseudo-vacuum. In this case, one finds at $m=1$ such a special root configuration with a kernel that has no singularity for $0\leq x \leq 1$. 
Expanding for example the free energy with all the roots symmetrically packed around the origin (which is the root configuration of the ground state in the antiferromagnetic regime), around $m=1$, one finds exactly the same coefficients as around $m=0$, up to a minus sign for odd coefficients. This implies that the function is symmetric around $m=1/2$ where half of the spins are down and half up, which implies, non surprisingly, that the energies are unchanged if all the spins are flipped. This means that in the case of the $s=1/2$ spin chain, this new pseudo-vacuum is \textit{exactly} the second ferromagnetic ground state, around which one could have performed the ABA. This gives another justification for the use of the term `pseudo-vacuum' for these special states.

\subsection{A trajectory to the state \eqref{groundstate} \label{trajgr}}
Let us now choose a trajectory that goes to the state \eqref{groundstate} at $m=-1$. We can take for example
\begin{equation}
\label{eqtraj}
X^\xi_a=\int_{-1}^{1}g(x)e^{2i\pi a x}dx\quad \text{with }g(x)=\begin{cases}
-1 & \text{for } -1/2+\xi/2<x<1/2-\xi/2 \,, \\
-1 & \text{for }3/4<x<3/4+\xi/2 \,, \\
-1 & \text{for }-3/4-\xi/2<x<-3/4 \,, \\
\phantom{-}0 & \text{otherwise} \,.
\end{cases}
\end{equation}
This trajectory at $m=-1$ is depicted in Figure \ref{trajecto} with th conventions given at the beginning of Section \ref{specialroot}. It has the property that at $\xi=0$ it is the second pseudo-vacuum, at $\xi=1/2$ it is the  state \eqref{groundstate}, and at $\xi=-1/2$ it is the second pseudo-vacuum again. This last property gives a strong check of the expansion: its evaluation at $\xi=-1/2$ should give back $1$, the second pseudo-vacuum energy.

\begin{figure}[H]
\begin{center}
\begin{tikzpicture}[scale=1]
\draw [black,thick,domain=0:360] plot ({cos(\x)}, {sin(\x)});
\draw [red,thick, line width=2pt,domain=-90:90] plot ({1.1*cos(\x)}, {1.1*sin(\x)});
\draw [red,thick, line width=2pt,domain=90:270] plot ({1.2*cos(\x)}, {1.2*sin(\x)});
\draw (0,0) node {\tiny $\xi=-1/2$};
\end{tikzpicture}
\hspace{1cm}
\begin{tikzpicture}[scale=1]
\draw [black,thick,domain=0:360] plot ({cos(\x)}, {sin(\x)});
\draw [red,thick, line width=2pt,domain=-90:90] plot ({1.1*cos(\x)}, {1.1*sin(\x)});
\draw [red,thick, line width=2pt,domain=150:210] plot ({1.2*cos(\x)}, {1.2*sin(\x)});
\draw [red,thick, line width=2pt,domain=120:240] plot ({1.1*cos(\x)}, {1.1*sin(\x)});
\draw (0,0) node {\tiny $-1/2<\xi<0$};
\end{tikzpicture}
\hspace{1cm}
\begin{tikzpicture}[scale=1]
\draw [black,thick,domain=0:360] plot ({cos(\x)}, {sin(\x)});
\draw [red,thick, line width=2pt,domain=0:360] plot ({1.1*cos(\x)}, {1.1*sin(\x)});
\draw (0,0) node {\tiny $\xi=0$};
\end{tikzpicture}
\hspace{1cm}
\begin{tikzpicture}[scale=1]
\draw [black,thick,domain=0:360] plot ({cos(\x)}, {sin(\x)});
\draw [red,thick, line width=2pt,domain=-150:150] plot ({1.1*cos(\x)}, {1.1*sin(\x)});
\draw [red,thick, line width=2pt,domain=60:90] plot ({1.2*cos(\x)}, {1.2*sin(\x)});
\draw [red,thick, line width=2pt,domain=-90:-60] plot ({1.2*cos(\x)}, {1.2*sin(\x)});
\draw (0,0) node {\tiny $0<\xi<1/2$};
\end{tikzpicture}
\hspace{1cm}
\begin{tikzpicture}[scale=1]
\draw [black,thick,domain=0:360] plot ({cos(\x)}, {sin(\x)});
\draw [red,thick, line width=2pt,domain=-90:90] plot ({1.1*cos(\x)}, {1.1*sin(\x)});
\draw [red,thick, line width=2pt,domain=-90:90] plot ({1.2*cos(\x)}, {1.2*sin(\x)});
\draw (0,0) node {\tiny $\xi=1/2$};
\end{tikzpicture}
\end{center}
\caption{Sketch of the root structure of the trajectory for different values of $\xi$.}
\label{trajecto}
\end{figure}
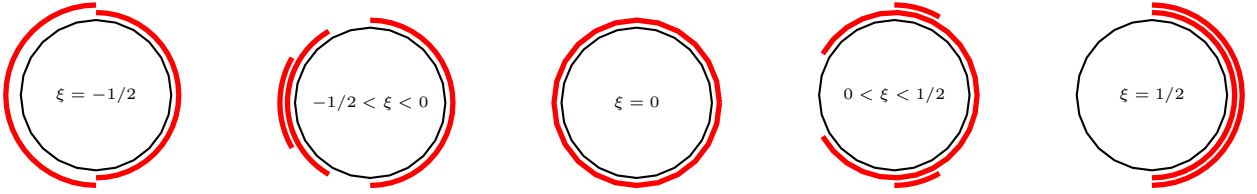

Calculating the moments $X_a^\xi$, we obtain the values for $\Xi_t^p[F(t)]$ for a function $F(t)$
\begin{equation}
\label{xitptraj}
\Xi_t^p[F(t)]=\frac{(i\pi)^{p-1}}{2p!}\left[ (1-(-1)^p)\slashed{\partial}^{p-1}F(-1)+(-1)^p\slashed{\partial}^{p-1}F(i)-\slashed{\partial}^{p-1}F(-i)\right] \,.
\end{equation}
Applying then the recurrence \eqref{rectraj} and formula \eqref{fetraj}, one obtains analytic expressions for all the coefficients in $\xi$ of the energy along the trajectory. For example the first two terms read
\begin{equation}
\begin{aligned}
F(\xi)&=1+\frac{\xi}{2}(-4\sqrt{5}+2\Re \sqrt{5-4i})\\
&-\frac{\xi^2}{2\sqrt{205}}\Big(\pi \Re (5+2i)\sqrt{25-20i}-4\Re ((2+5i)\sqrt{25+20i}+2\sqrt{41})\argth \tfrac{\sqrt{5}-\sqrt{5-4i}}{2}\\
&+2\sqrt{5}\Im (2+5i)\sqrt{5+4i}\arctan \Im \sqrt{5-4i}\Big)\\
&+O(\xi^3) \,.
\end{aligned}
\end{equation}
We computed the coefficients up to $\xi^{14}$ using the recurrence relations written in Appendix \ref{recurtraj}. The energy of the trajectory is reported in Figure \ref{resultgroundstate}. 

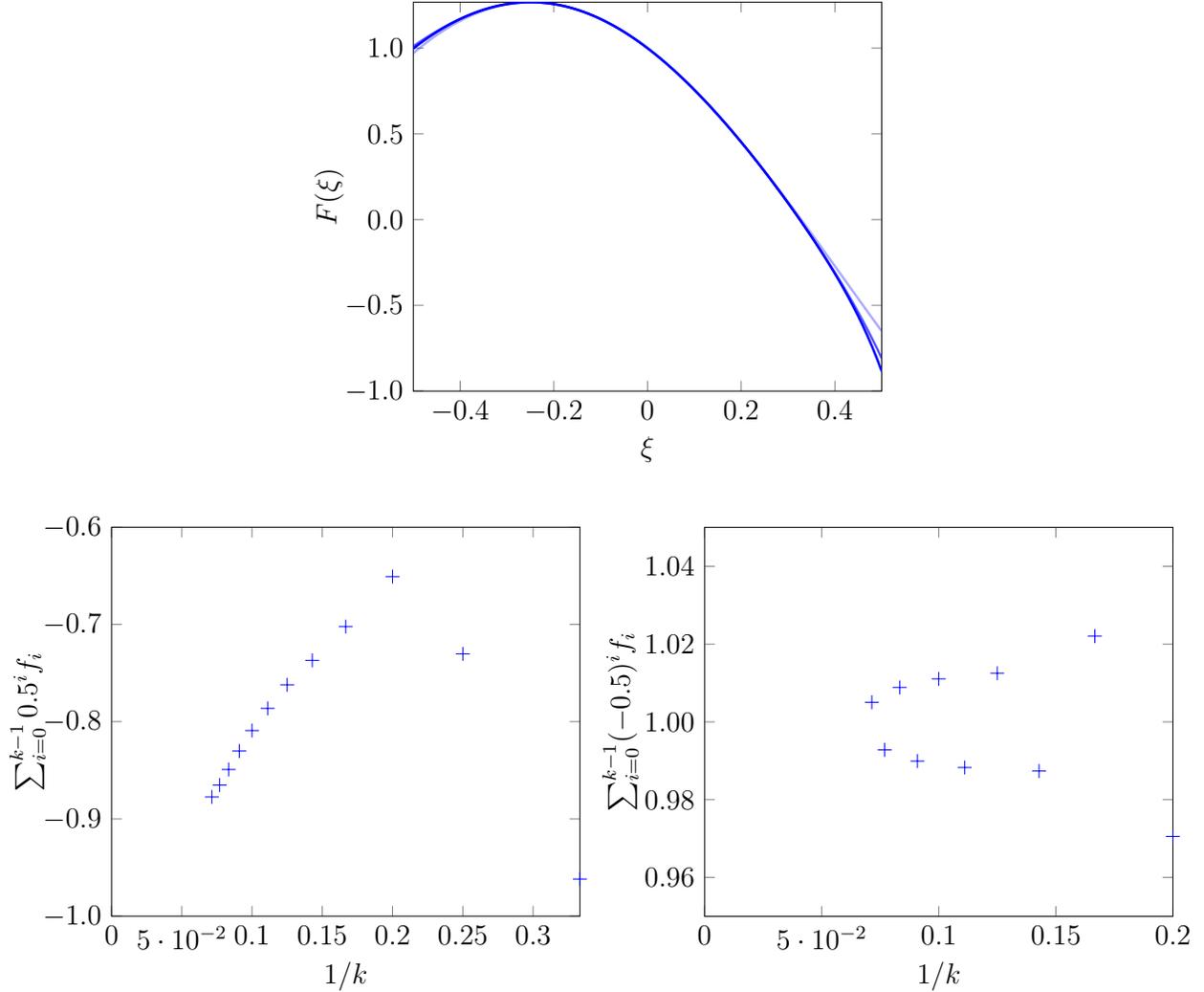
\begin{figure}[H]
\begin{center}
\begin{tikzpicture}[scale=0.95]
\begin{axis}[
    enlargelimits=false,
    ymin=-1,
    xlabel = $\xi$,
    ylabel = $F(\xi)$,
     y tick label style={
        /pgf/number format/.cd,
            fixed,
            fixed zerofill,
            precision=1,
        /tikz/.cd
    }
]
\addplot [
    domain=-0.5:0.5, 
    samples=100, 
    line width=1pt,
    color=blue,
    opacity=0.33,
    ]
{1.-2.0843415503833818*x-3.6786511453555093*x^2+1.8521456475091602*x^3+1.2729662567394795*x^4};
\addplot [
    domain=-0.5:0.5, 
    samples=100, 
    line width=1pt,
    color=blue,
    opacity=0.66,
    ]
{1.-2.0843415503833818*x-3.6786511453555093*x^2+1.8521456475091602*x^3+1.2729662567394795*x^4-1.647808543764578*x^5-2.219744951184276*x^6-3.2194722411616423*x^7-6.207614745862771*x^8-11.666870692572967*x^9};
\addplot [
    domain=-0.5:0.5, 
    samples=100, 
    line width=1pt,
    color=blue,
    opacity=1,
    ]
{1.-2.0843415503833818*x-3.6786511453555093*x^2+1.8521456475091602*x^3+1.2729662567394795*x^4-1.647808543764578*x^5-2.219744951184276*x^6-3.2194722411616423*x^7-6.207614745862771*x^8-11.666870692572967*x^9-21.633346365220138*x^10-38.76719813230275*x^11-65.69322063039509*x^12-100.21286387073893*x^13-122.14702607037522*x^14};

\end{axis}
\end{tikzpicture}
\end{center}

\begin{tikzpicture}[scale=0.95]
\begin{axis}[
    enlargelimits=false,
    ymin=-1,
    xmin=0,
    ymax=-0.6,
    xlabel = $1/k$,
    ylabel = $\sum_{i=0}^{k-1} 0.5^i f_i$,
     y tick label style={
        /pgf/number format/.cd,
            fixed,
            fixed zerofill,
            precision=1,
        /tikz/.cd
    }
]
\addplot[
    only marks,
    mark=+,
    mark size=2.9pt,
    color=blue]
table[
           x expr=1/\thisrowno{0}, 
           y expr=\thisrowno{1} 
         ]{SL2C_res_groundstate_traj1_tab.dat};

\end{axis}
\end{tikzpicture}
\begin{tikzpicture}[scale=0.95]
\begin{axis}[
    enlargelimits=false,
    xmin=0,
    ymax=1.05,
    ymin=0.95,
    xlabel = $1/k$,
    ylabel = $\sum_{i=0}^{k-1} (-0.5)^i f_i$,
     y tick label style={
        /pgf/number format/.cd,
            fixed,
            fixed zerofill,
            precision=2,
        /tikz/.cd
    }
]
\addplot[
    only marks,
    mark=+,
    mark size=2.9pt,
    color=blue]
table[
           x expr=1/\thisrowno{0}, 
           y expr=\thisrowno{1} 
         ]{SL2C_res_groundstate_traj1_check_tab.dat};
\addplot [
    domain=-0.5:0.5, 
    samples=100, 
    line width=1pt,
    color=black,
    opacity=1
    ]
{0+0*x};

\end{axis}
\end{tikzpicture}

\caption{Top: energy of the trajectory \eqref{eqtraj} as a function of $\xi$, up to $\xi^4$, $\xi^9$, $\xi^{14}$ (from light to dark blue). Bottom: energy at $\xi=1/2$ (left) and $\xi=-1/2$ (right), taking into account the first $k$ terms in the expansion in $\xi$, as a function of $1/k$.}
\label{resultgroundstate}
\end{figure}
Because of the small oscillations observed around a seemingly straight line, a sensible extrapolation to $k=\infty$ requires to take several points to average them out. Performing a simple linear fit $a+\tfrac{b}{k}$ on the almost aligned points for $k\geq 5$ we obtain $F(\xi=1/2)\approx -0.992$, and for $k\geq 6$ we obtain $F(\xi=1/2)\approx -1.002$. Hence this strongly suggests
\begin{equation} \label{gs-candidate}
F(\xi=1/2)=-1\,,
\end{equation}
and hence the energy level of the $SL(2,\mathbb{C})$ spin chain
\begin{equation}
\mathcal{E}(m=-1)=0 \,.
\end{equation}
This value corresponds to the value of the ground state obtained in \cite{derkachovkorchemsky2} by calculating numerically the ground state energy for small sizes up to $L=8$ and extrapolating to the thermodynamic limit\footnote{In \cite{derkachovkorchemsky2} is also obtained that the $1/L$ corrections to this result vanish.}. To ensure that we are dealing with the same state indeed, we need to check that the energy of this state is minimal with respect to all excitations.

\subsection{First-level particle and hole excitations above the state \eqref{groundstate}}
In this section, we verify that the state \eqref{groundstate} at $m=-1$ is minimal with respect to  microscopic excitations described within Theorem \ref{thm1}.

\medskip

At $m>0$, since all the roots are real, there are only particle excitations (i.e., adding a Bethe root with a Bethe number that is not already taken by another root, and thus increasing the value of $m$ by $1/L$) or hole excitations (i.e., removing one of the Bethe roots, and thus decreasing the value of $m$ by $1/L$). Because of the structure of \eqref{eqgs}, the only possible values of $z=\tfrac{I_k}{L}$ for the Bethe number involved are such that $\chi_m(z)=0$ for particle excitations and $\chi_m(z)=1$ for hole excitations. We will call \textit{first-level} particle or hole excitations, those such that $-1/2<z<1/2$, for which Theorem \ref{thm1} applies. As $m$ varies, the authorized values of $z$ for first-level particle or hole excitations vary correspondingly. At $m=-1$, they become $-1/4<z<1/4$ for hole excitations, and $1/4<z<1/2$ or $-1/2<z<-1/4$ for particle excitations.

\medskip

If we consider a macroscopic but tiny number $\eta L$ of such excitations around $z$, then denoting by $\chi^z_m$ the resulting filling function of the new Bethe root distribution, its moments $X_a(\chi_m^z)$ are, at first order in $\eta$
\begin{equation}
\label{momexc}
X_a(\chi_m^z)=X_a(\chi_m)+\eta e^{2i\pi a z}+O(\eta^2)\,.
\end{equation}
This writing encompasses the two types of excitations according to the sign of $\eta$: for particle excitations we have $\eta>0$, and for hole excitations we have $\eta<0$. From the moments, one can deduce at first order in $\eta$ the change in the energy at large $\varphi$, with the expansion presented in Theorem \ref{thm1}. One has the first terms
\begin{equation}
\partial_\eta F_{\chi^z_m}(\varphi)|_{\eta=0}=\frac{e^{-2i\pi z}}{2}e^{2\varphi}+1+4m-2e^{2i\pi z}X_{-1}-2e^{-2i\pi z}X_1+O(e^{-2\varphi})\,.
\end{equation}
We should now recall that this is the change to the energy of an eigenstate of only one of the two copies of the $SL(2,\mathbb{R})$ spin chain composing the whole $SL(2,\mathbb{C})$ spin chain. Since the eigenstate of the other $SL(2,\mathbb{R})$ spin chain copy must have a magnetisation $m'=-2-m$, we conclude that it must undergo an excitation of the opposite type, i.e., with $\eta$ changed into $-\eta$. Moreover, contrarily to the $SL(2,\mathbb{R})$ spin chain, the $SL(2,\mathbb{C})$ spin chain is Hermitian, hence with a real spectrum. The value of $z$ for the particle excitation (denoted $z_p$) and the value of $z$ for the hole excitation (denoted $z_h$) composing an elementary excitation of the whole $SL(2,\mathbb{C})$ spin chain are thus constrained to be such that the total excitation energy is real. Hence, the change of energy of the state \eqref{groundstate} of the $SL(2,\mathbb{C})$ spin chain after a particle-hole excitation $(z_p,z_h)$ is
\begin{equation}
\delta_{z_p,z_h}=\left. \partial_\eta F_{\chi^{z_p}_m}(\varphi) \right|_{\eta=0}- \left. \partial_\eta F_{\chi^{z_h}_m}(\varphi) \right|_{\eta=0}\,,
\end{equation}
where the analytic continuation is taken to $m=-1$, and with the constraints on $(z_p,z_h)$
\begin{equation}
\label{excconstraints}
\begin{aligned}
&-1/4<z_p<1/4\\
&1/4<|z_h|<1/2\\
&\delta_{z_p,z_h}\, \text{real}\,.
\end{aligned}
\end{equation}
In practice, the constraint $\Im \delta_{z_p,z_h}=0$ leaves only one of the two parameters $z_p$ or $z_h$, with the other becoming a (possibly multi-valued) function of the first. For example, in the limit $\varphi\to\infty$ we have still for the state \eqref{groundstate}
\begin{equation}
\left. \partial_\eta F_{\chi^z_m}(\varphi) \right|_{\eta=0}=\frac{e^{-2i\pi z}}{2}e^{2\varphi}+O(1)\,,
\end{equation}
from which one deduces that the couples $(z_p,z_h)$ satisfying \eqref{excconstraints} are
\begin{equation}
(z_p,(\tfrac{1}{2}-|z_p|)\sign(z_p))\,,\qquad -1/4<z_p<1/4\,,
\end{equation}
which gives
\begin{equation}
\delta_{z_p,z_h(z_p)}=\cos(2\pi z_p)e^{2\varphi}+O(1)\,,
\end{equation}
which is indeed always positive. So the state considered \eqref{groundstate} is indeed minimal with respect to first-level particle-hole excitations in the limit $\varphi\to\infty$.

To investigate the case $\varphi<\infty$, we start by plotting in Figure \ref{resultexcspv} the complex values of $\partial_\eta F_{\chi^{z}_m}(\varphi)|_{\eta=0}$ and the real values of $\delta_{z_p,z_h(z_p)}$ calculated with the expansion in $e^{-2\varphi}$, evaluated at $\varphi=1.5$ which is within its radius of convergence. We see that we have indeed $\delta_{z_p,z_h(z_p)}\geq 0$ for all $-1/4\leq z_p\leq 1/4$, which shows that the state \eqref{groundstate} at $\varphi=1.5$ is still a local minimum with respect to particle-hole excitations. Moreover, we see that the excitations are even gapped (with a gap extensive in $L$) at $\varphi=1.5$.
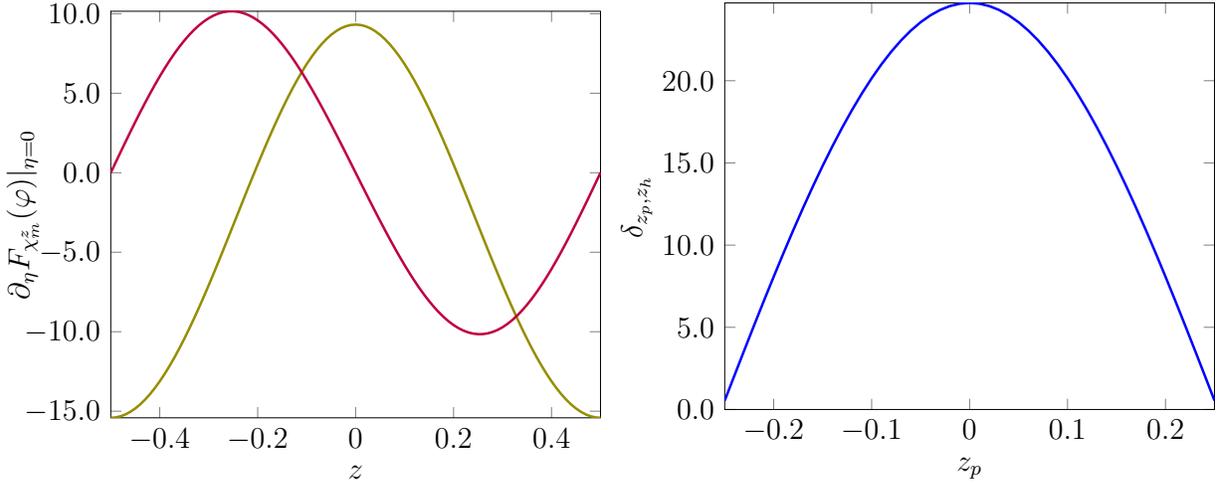
\begin{figure}[H]
\begin{center}
\begin{tikzpicture}[scale=0.95]
\begin{axis}[
    enlargelimits=false,
    ylabel = $\partial_\eta F_{\chi^{z}_m}(\varphi)|_{\eta=0}$,
    xlabel = $z$,
     y tick label style={
        /pgf/number format/.cd,
            fixed,
            fixed zerofill,
            precision=1,
        /tikz/.cd
    }
]
\addplot[
    color=olive,
    line width=1pt]
table[
           x expr=\thisrowno{0}, 
           y expr=\thisrowno{1} 
         ]{sl2c_exc_largephi.dat};
\addplot[
    color=purple,
    line width=1pt]
table[
           x expr=\thisrowno{0}, 
           y expr=\thisrowno{2} 
         ]{sl2c_exc_largephi.dat};
\end{axis}
\end{tikzpicture}
\begin{tikzpicture}[scale=0.95]
\begin{axis}[
    enlargelimits=false,
    ylabel = $\delta_{z_p,z_h}$,
    ymin=0,
    xlabel = $z_p$,
     y tick label style={
        /pgf/number format/.cd,
            fixed,
            fixed zerofill,
            precision=1,
        /tikz/.cd
    }
]
\addplot[
    color=blue,
    line width=1pt]
table[
           x expr=\thisrowno{0}, 
           y expr=\thisrowno{1} 
         ]{sl2c_exc_largephi_ph.dat};
\end{axis}
\end{tikzpicture}
\end{center}
\caption{Left: $ \partial_\eta F_{\chi^{z}_m}(\varphi)|_{\eta=0}$ for \eqref{groundstate} as a function of $z$, real part (green) and imaginary part (purple), at $\varphi=1.5$ with $13$ terms in the expansion in $e^{-2\varphi}$. Right: the corresponding $\delta_{z_p,z_h}$ for admissible values of $(z_p,z_h)$, as a function of $z_p$.}
\label{resultexcspv}
\end{figure}

Once again, the series in $e^{-2\varphi}$ are not convergent at $\varphi=0$. In order to investigate the values of $\delta_{z_p,z_h}$ at $\varphi=0$, we apply the reasoning presented in sections \ref{expandtraj} and \ref{trajgr}, with now the values \eqref{momexc} for the moments at order $1$ in $\eta$.

The values of the functional $\Xi_t^p[F(t)]$ for $p>0$ are not modified and are given by \eqref{xitptraj}, whereas for $p=0$ we have
\begin{equation}
\begin{aligned}
\label{xitptrajexc}
\Xi_t^0[F(t)]&=-F_0+\eta F(e^{2i\pi z})\,.
\end{aligned}
\end{equation}
Thus we obtain
\begin{equation}
\label{rectrajexc}
\begin{aligned}
\arctan\left(i+\sum_{p\geq 0}\gamma_p(t,x)\xi^p\right)&=\frac{\log(xt)}{2i}+\arctan \left(\sum_{p\geq 0}\gamma_p(t,x)\xi^p\right)\\
&-\eta\arctan\left(\sum_{p\geq 0}\gamma_p(t,x)\xi^p-\sum_{p\geq 0}\gamma_p(e^{2i\pi z},x)\xi^p\right)\\
&-\sum_{p\geq 1}\xi^p \Xi^p_u\left[\arctan \left(\sum_{p\geq 0}(\gamma_p(t,x)-\gamma_p(u,x))\xi^p\right) \right] \,.
\end{aligned}
\end{equation}
For example, the order $\xi^0$ gives the energy of the particle or hole excitations above the other pseudo-vacuum \eqref{secondpseudo} at $\varphi=0$
\begin{equation}
\label{excsecondpseudo}
\partial_\eta F(\chi_{-1}^{z;\xi=0})|_{\eta=0}=\frac{e^{-2i\pi z}}{2}(-1+e^{2i\pi z})^2\sqrt{\frac{-1+9e^{2i\pi z}}{-1+e^{2i\pi z}}}\,.
\end{equation}
In Figure \ref{resultexcspv2} is plotted the result of this expansion in $\xi$. We observe first that the results at $\varphi=0$ are qualitatively different from those at large $\varphi$ shown previously; with this expansion in $\xi$ (that can be performed at any value of $\varphi$) we observe indeed a change of regime as $\varphi$ decreases to $0$. Besides, we see that we indeed have $\delta_{z_p,z_h}\geq 0$ for all $(z_p,z_h)$ satisfying the constraints, which means that the state \eqref{groundstate} is indeed of minimal energy with respect to first-level particle-hole excitations. Moreover, we observe as in the case $\varphi=1.5$ that these excitations are gapped excitations.

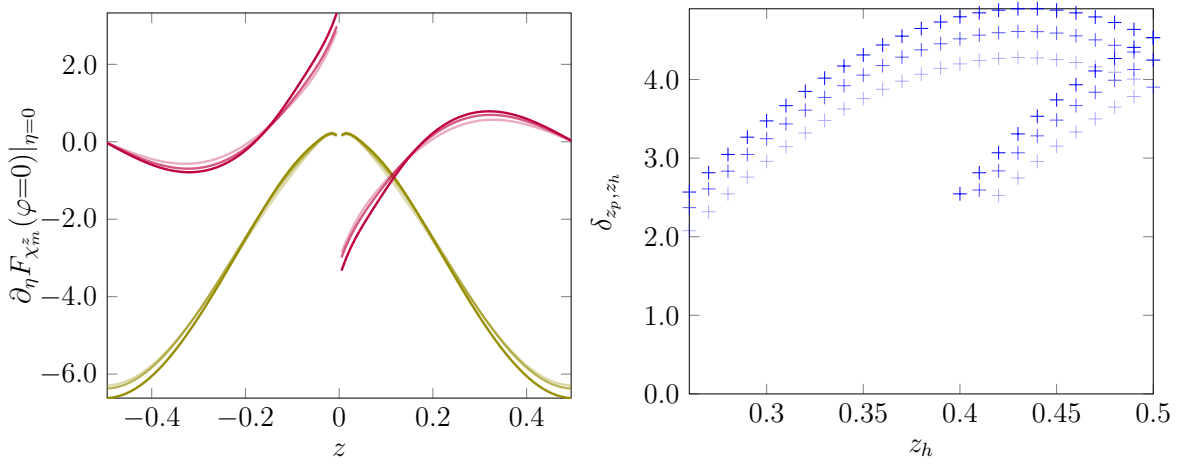
\begin{figure}[H]
\begin{center}
\begin{tikzpicture}[scale=0.9]
\begin{axis}[
    enlargelimits=false,
    ylabel = $\partial_\eta F_{\chi^{z}_m}(\varphi{=}0)|_{\eta=0}$,
    xlabel = $z$,
     y tick label style={
        /pgf/number format/.cd,
            fixed,
            fixed zerofill,
            precision=1,
        /tikz/.cd
    }
]
\addplot[
opacity=0.33,
    color=olive,
    line width=1pt]
table[
           x expr=\thisrowno{0}, 
           y expr=\thisrowno{1} 
         ]{excspectrum_3terms_x1.dat};
\addplot[
opacity=0.33,
    color=olive,
    line width=1pt]
table[
           x expr=\thisrowno{0}, 
           y expr=\thisrowno{1} 
         ]{excspectrum_3terms_x1_pos.dat};
\addplot[
opacity=0.66,
    color=olive,
    line width=1pt]
table[
           x expr=\thisrowno{0}, 
           y expr=\thisrowno{1} 
         ]{excspectrum_4terms_x1.dat};
\addplot[
opacity=0.66,
    color=olive,
    line width=1pt]
table[
           x expr=\thisrowno{0}, 
           y expr=\thisrowno{1} 
         ]{excspectrum_4terms_x1_pos.dat};
\addplot[
opacity=1,
    color=olive,
    line width=1pt]
table[
           x expr=\thisrowno{0}, 
           y expr=\thisrowno{1} 
         ]{excspectrum_5terms_x1.dat};
\addplot[
opacity=1,
    color=olive,
    line width=1pt]
table[
           x expr=\thisrowno{0}, 
           y expr=\thisrowno{1} 
         ]{excspectrum_5terms_x1_pos.dat};
         \addplot[opacity=0.33,
    color=purple,
    line width=1pt]
table[
           x expr=\thisrowno{0}, 
           y expr=\thisrowno{2} 
         ]{excspectrum_3terms_x1.dat};
         \addplot[opacity=0.33,
    color=purple,
    line width=1pt]
table[
           x expr=\thisrowno{0}, 
           y expr=\thisrowno{2} 
         ]{excspectrum_3terms_x1_pos.dat};
\addplot[opacity=0.66,
    color=purple,
    line width=1pt]
table[
           x expr=\thisrowno{0}, 
           y expr=\thisrowno{2} 
         ]{excspectrum_4terms_x1.dat};
\addplot[opacity=0.66,
    color=purple,
    line width=1pt]
table[
           x expr=\thisrowno{0}, 
           y expr=\thisrowno{2} 
         ]{excspectrum_4terms_x1_pos.dat};
\addplot[opacity=1,
    color=purple,
    line width=1pt]
table[
           x expr=\thisrowno{0}, 
           y expr=\thisrowno{2} 
         ]{excspectrum_5terms_x1.dat};
\addplot[opacity=1,
    color=purple,
    line width=1pt]
table[
           x expr=\thisrowno{0}, 
           y expr=\thisrowno{2} 
         ]{excspectrum_5terms_x1_pos.dat};
\end{axis}
\end{tikzpicture}
\begin{tikzpicture}[scale=0.9]
\begin{axis}[
    enlargelimits=false,
    ylabel = $\delta_{z_p,z_h}$,
    xlabel = $z_h$,
    ymin=0,
     y tick label style={
        /pgf/number format/.cd,
            fixed,
            fixed zerofill,
            precision=1,
        /tikz/.cd
    }
]
\addplot[
    color=blue,
    opacity=0.33,
    only marks,
    mark=+,
    mark size=2.9pt]
table[
           x expr=\thisrowno{0}, 
           y expr=\thisrowno{1} 
         ]{exc_x1_3terms_phi0.dat};
\addplot[
    color=blue,
    opacity=0.66,
    only marks,
    mark=+,
    mark size=2.9pt]
table[
           x expr=\thisrowno{0}, 
           y expr=\thisrowno{1} 
         ]{exc_x1_4terms_phi0.dat};
\addplot[
    color=blue,
    opacity=1,
    only marks,
    mark=+,
    mark size=2.9pt]
table[
           x expr=\thisrowno{0}, 
           y expr=\thisrowno{1} 
         ]{exc_x1_5terms_phi0.dat};
\end{axis}
\end{tikzpicture}
\end{center}
\caption{Left: $ \partial_\eta F_{\chi^{z}_m}(\varphi)|_{\eta=0}$ as a function of $z$, real part (green) and imaginary part (red), at $\varphi=0$ with $3,4,5$ terms in the expansion in $\xi$ (from light to dark colors). Right: the corresponding $\delta_{z_p,z_h}$ for admissible values of $(z_p,z_h)$, as a function of $z_h$, with $3,4,5$ terms in the expansion in $\xi$ (from light to dark blue).}
\label{resultexcspv2}
\end{figure}
To conclude this section, we presented evidence for the minimality of the state \eqref{groundstate} at $m=-1$ with respect to first-level particle-hole excitations, i.e. particle-hole excitations with Bethe numbers $-\tfrac{L}{2}<I_k<\tfrac{L}{2}$, which constitute all the possible excitations to which Theorem \ref{thm1} applies. Together with the fact that its energy in the continuum limit is the same as the one found in \cite{derkachovkorchemsky2}, this is strong evidence that it is the ground state indeed.

Our analysis also shows that these first-level excitations are even gapped excitations. However, there are also other possible excitations with Bethe numbers $|I_k|>L/2$, and also the possibility of giving to $m$ a small imaginary part, due to the fact that the spins $u,\bar{u}$ of the $SL(2,\mathbb{C})$ spin chain representations can be complex. This will be studied in further work.

\subsection{Another trajectory to the state \eqref{groundstate}}
The previous trajectory that goes to the ground state at $m=-1$ is clearly not unique. Another example of such a trajectory is
\begin{equation}
\label{eqtraj2}
X^\xi_a=\int_{-1}^{1}g(x)e^{2i\pi a x}dx\quad \text{with }g(x)=\begin{cases}
-1 & \text{for } -1/2+\xi/2<x<1/2-\xi/2 \,, \\
-1 & \text{for }7/8-\xi/4<x<7/8+\xi/4 \,, \\
-1 & \text{for }-7/8-\xi/4<x<-7/8+\xi/4 \,, \\
0 & \text{otherwise} \,.
\end{cases}
\end{equation}
This trajectory at $m=-1$ is depicted in Figure \ref{trajecto2}.

\begin{figure}[H]
\begin{center}
\begin{tikzpicture}[scale=1]
\draw [black,thick,domain=0:360] plot ({cos(\x)}, {sin(\x)});
\draw [red,thick, line width=2pt,domain=90:270] plot ({1.1*cos(\x)}, {1.1*sin(\x)});
\draw [red,thick, line width=2pt,domain=90:270] plot ({1.2*cos(\x)}, {1.2*sin(\x)});
\draw (0,0) node {\tiny $\xi=-1/2$};
\end{tikzpicture}
\hspace{1cm}
\begin{tikzpicture}[scale=1]
\draw [black,thick,domain=0:360] plot ({cos(\x)}, {sin(\x)});
\draw [red,thick, line width=2pt,domain=-30:30] plot ({1.1*cos(\x)}, {1.1*sin(\x)});
\draw [red,thick, line width=2pt,domain=60:300] plot ({1.1*cos(\x)}, {1.1*sin(\x)});
\draw [red,thick, line width=2pt,domain=150:210] plot ({1.2*cos(\x)}, {1.2*sin(\x)});
\draw [red,thick, line width=2pt,domain=120:240] plot ({1.1*cos(\x)}, {1.1*sin(\x)});
\draw (0,0) node {\tiny $-1/2<\xi<0$};
\end{tikzpicture}
\hspace{1cm}
\begin{tikzpicture}[scale=1]
\draw [black,thick,domain=0:360] plot ({cos(\x)}, {sin(\x)});
\draw [red,thick, line width=2pt,domain=0:360] plot ({1.1*cos(\x)}, {1.1*sin(\x)});
\draw (0,0) node {\tiny $\xi=0$};
\end{tikzpicture}
\hspace{1cm}
\begin{tikzpicture}[scale=1]
\draw [black,thick,domain=0:360] plot ({cos(\x)}, {sin(\x)});
\draw [red,thick, line width=2pt,domain=-150:150] plot ({1.1*cos(\x)}, {1.1*sin(\x)});
\draw [red,thick, line width=2pt,domain=30:60] plot ({1.2*cos(\x)}, {1.2*sin(\x)});
\draw [red,thick, line width=2pt,domain=-60:-30] plot ({1.2*cos(\x)}, {1.2*sin(\x)});
\draw (0,0) node {\tiny $0<\xi<1/2$};
\end{tikzpicture}
\hspace{1cm}
\begin{tikzpicture}[scale=1]
\draw [black,thick,domain=0:360] plot ({cos(\x)}, {sin(\x)});
\draw [red,thick, line width=2pt,domain=-90:90] plot ({1.1*cos(\x)}, {1.1*sin(\x)});
\draw [red,thick, line width=2pt,domain=-90:90] plot ({1.2*cos(\x)}, {1.2*sin(\x)});
\draw (0,0) node {\tiny $\xi=1/2$};
\end{tikzpicture}
\end{center}
\caption{Sketch of the root structure of the trajectory for different values of $\xi$.}
\label{trajecto2}
\end{figure}
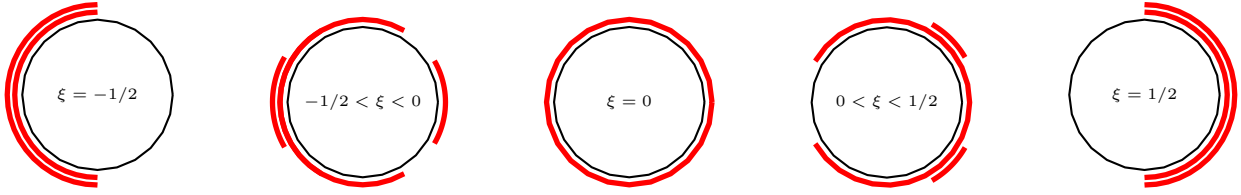

Calculating the moments $X_a(\chi^\xi_{-1})$, we obtain the values for $\Xi_t^p[F(t)]$ for a function $F(t)$
\begin{equation}
\label{xitptrajnew}
\Xi_t^p[F(t)]=(1-(-1)^p)\frac{(i\pi)^{p-1}}{2p!}\left[\slashed{\partial}^{p-1}F(-1)-\frac{\slashed{\partial}^{p-1}F(e^{i\pi/4})+\slashed{\partial}^{p-1}F(e^{-i\pi/4})}{2^p}\right] \,.
\end{equation}
We report in Figure \ref{resultgroundstate2} the result for the energy of this trajectory, by plotting the partial series taking into account $k$ terms, as a function of $1/k$. We see that the result is compatible with the value obtained with the other trajectory, with a curve moving towards around $-1$ in the limit $k\to\infty$. Having several different trajectories going to the same state offers the possibility of more consistency checks when studying its properties.

\begin{figure}[H]
\begin{tikzpicture}[scale=0.95]
\begin{axis}[
    enlargelimits=false,
    ymin=-1,
    xlabel = $\xi$,
    ylabel = $F(\xi)$,
     y tick label style={
        /pgf/number format/.cd,
            fixed,
            fixed zerofill,
            precision=1,
        /tikz/.cd
    }
]
\addplot [
    domain=-0.5:0.5, 
    samples=100, 
    line width=1pt,
    color=blue,
    opacity=0.33,
    ]
{1. - 3.64696 *x - 1.23503* x^2 + 2.61719* x^3 };
\addplot [
    domain=-0.5:0.5, 
    samples=100, 
    line width=1pt,
    color=blue,
    opacity=0.66,
    ]
{1. - 3.64696 *x - 1.23503* x^2 + 2.61719* x^3 + 0.486268* x^4 - 
 2.02445* x^5};
\addplot [
    domain=-0.5:0.5, 
    samples=100, 
    line width=1pt,
    color=blue,
    opacity=1,
    ]
{1. - 3.64696 *x - 1.23503* x^2 + 2.61719* x^3 + 0.486268* x^4 - 
 2.02445* x^5 - 2.67343* x^6 - 2.55268* x^7};
\end{axis}
\end{tikzpicture}
\begin{tikzpicture}[scale=0.95]
\begin{axis}[
    enlargelimits=false,
    ymin=-1,
    xmin=0,
    ymax=-0.7,
    xlabel = $1/k$,
    ylabel = $\sum_{i=0}^{k-1} 0.5^i f_i$,
     y tick label style={
        /pgf/number format/.cd,
            fixed,
            fixed zerofill,
            precision=1,
        /tikz/.cd
    }
]
\addplot[
    only marks,
    mark=+,
    mark size=2.9pt,
    color=blue]
table[
           x expr=1/\thisrowno{0}, 
           y expr=\thisrowno{1} 
         ]{groundstate2.dat};
\end{axis}
\end{tikzpicture}
\caption{Left: energy of the trajectory \eqref{eqtraj2} as a function of $\xi$, up to $\xi^3$, $\xi^5$, $\xi^7$ (from light to dark blue). Right: energy at $\xi=1/2$, taking into account the first $k$ terms in the expansion in $\xi$, as a function of $1/k$.}
\label{resultgroundstate2}
\end{figure}
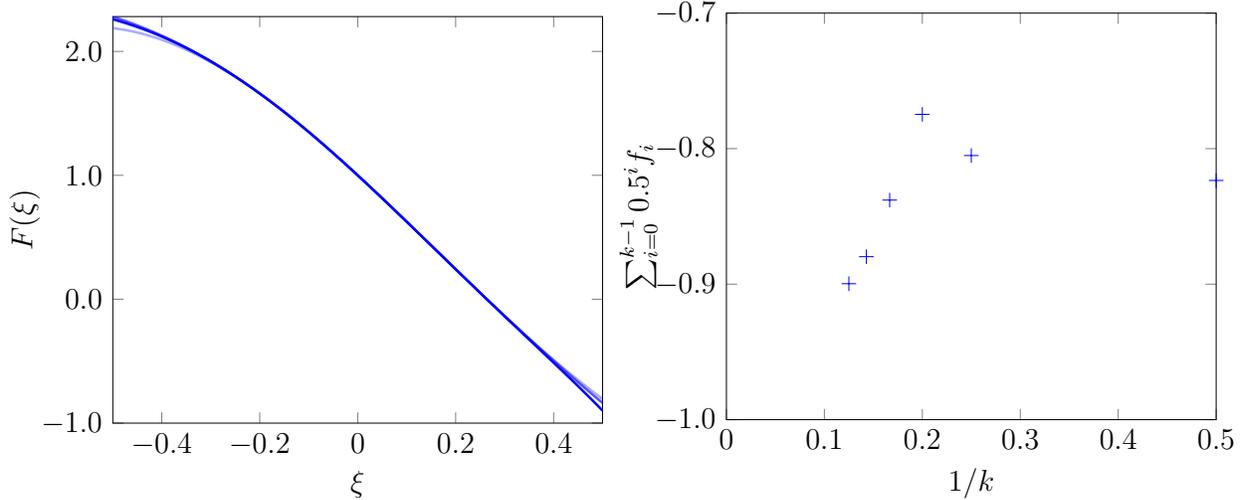

\section{Conclusion}
In this paper, we presented a method to analytically continue energies computed with the Bethe ansatz in the thermodynamic limit to a negative number of Bethe roots, and showed that it permits one to compute the (extensive part of the) energy levels of the $SL(2,\mathbb{C})$ non-compact spin chain in the thermodynamic limit. As a proof of principle, we recovered the value of the ground state previously obtained \cite{fadeevkorchemsky,derkachovkorchemsky2,braunderkachovmanashov,devegalipatov,janik2,janik} by extrapolating small sizes.

\medskip

The starting point was to observe that an energy of the $SL(2,\mathbb{C})$ spin chain for a state with magnetization $u$ has to be a sum of two energies of the $SL(2,\mathbb{R})$ spin chain at magnetizations $u$ and $\bar{u}=-1-u^*$, and that each of these can be obtained with the ABA provided $u\leq Ls$ and $\bar{u}\leq L \bar{s}$. Since these two conditions cannot be satisfied simultaneously, one needs to analytically continue the energies in terms of $m=s-\tfrac{u}{L}$ to $m<0$, in particular to $m$ close to $-1$. 

In order to perform this analytic continuation, we found it useful to introduce an imaginary extensive twist $\varphi$ and to study the behaviour of the thermodynamic limit of the energies at large $\varphi \to\infty$. Indeed, these thermodynamic energies are found to be expandable in a series in $e^{-2\varphi}$ with coefficients depending smoothly on $m$, which allows their analytic continuation to $m<0$. Although these series are convergent, their radius of convergence unfortunately does not include the sought value $\varphi=0$. 

To solve this problem, we identified a very special state for which all the coefficients of the series as well as their $m$-derivatives can be explicitly computed and resummed at $m=-1$. Remarkably, the absence of singularities of these expressions for $0\leq \varphi<\infty$  allowed us to analytically continue them down to $\varphi=0$, which provides the value of the energy of one specific state in the $SL(2,\mathbb{C})$ spin chain. It is not the ground state, but a state in the bulk of the spectrum.

In order to obtain the other energy levels, we used this special state as another `pseudo-vacuum' by expanding the energy levels on any trajectory that departs from this special state, and  explores the energy landscape of the chain. The coefficients of the corresponding  series can be efficiently computed one by one, and yield convergent series that allowed us to reach another state in the spectrum, not necessarily close to this pseudo-vacuum. These series also permit to study a certain (but large) class of excitations above a state, so that we are able to identify one whose energy is minimal with respect to any of these particle-hole excitations. The energy of this ground state that we compute is indeed the value previously obtained by extrapolation from  small-size studies \cite{derkachovkorchemsky2}.

\medskip

All throughout  the paper,  series and analytic continuations were compared  with stringent numerical tests. In particular, this led to an  expansion around $m=-1$ of the free energy obtained with all the Bethe roots symmetrically packed around the origin, whereas its natural point of expansion is around the usual pseudo-vacuum at $m=0$. These two expansions are in excellent agreement with the numerics.

The existence of this other pseudo-vacuum is not a specificity of this spin chain, and is also present in the $s=1/2$ Heisenberg spin chain, in which case it is simply the second ferromagnetic ground state. In the present case, however, this other pseudo-vacuum reveals new insights on the analytic structure of the Bethe equations and their solutions, and suggests exciting further studies.

\medskip

The present method is not restricted to the thermodynamic limit $L\to\infty$ only, and is expected to work as well to study $L^{-1}$ corrections. These contain crucial information on the field theory that describes the $SL(2,\mathbb{C})$ chain in the thermodynamic limit: further work along these lines will be the object of a subsequent paper.

\paragraph{Acknowledgements}
We are grateful to G.~Korchemsky for helpful discussions. We also thank F. H. L. Essler and Ch. Kopper for helpful comments. This work was supported by the ERC Advanced Grant NuQFT and by the EPSRC under grant EP/N01930X.

\appendix

\section{Series expansion in $e^{-2\varphi}$ \label{mathematica1}}
We give here a {\sc Mathematica} code to compute the series \eqref{fbchi}:
\begin{multicols}{2}
\tiny
\begin{verbatim}
ClearAll["Global`*"];
M = 7;
R = Join[{0}, Series[1/Pi ArcTan[x], {x, 0, M}][[3]]];
S = ConstantArray[0, M];
S[[1]] = -I/2/Pi Log[-2 I];
For[i = 2, i < M + 1, i++,
  S[[i]] = I/(2 Pi)*(I/2)^(i - 1)/(i - 1);
  ];
X[a_] := If[a == 0, m, Sin[Pi a m]/a/Pi];
Clear[X];
CC = ConstantArray[0, {M, M, M}];
CCtilde = ConstantArray[0, {M, M, M}];
SetAttributes[ComputeNext, HoldAll];
SetAttributes[FillMultiple, HoldAll];
SetAttributes[CompleteTable, HoldAll];
ComputeNext[a_, b_, c_, ctilde_] := 
  Module[{res, a1, b2, q, n, res2, Res},
   res = 0;
   For[a1 = 0, a1 < M, a1++,
    For[b2 = 0, b2 < M, b2++,
     If[(b - b2 >= 0),
      For[n = 0, n < M, n++,
       For[q = 0, q < n + 1, q++,
        res += -2 I Pi (-1)^q Binomial[n, q]*R[[n + 1]] X[a1]*
           c[[a1 + 1, b - b2 + 1, q + 1]]*
           c[[a + 1, b2 + 1, n - q + 1]];
        ]
       ]
      ]
     ]
    ];
   For[n = 2, n < M, n++,
    res += (-1)^n ctilde[[a + 1, b + 1, n + 1]]/n;
    ];
   For[n = 0, n < M, n++,
    res += 2 I Pi c[[a + 1, b + 1, n + 1]]*S[[n + 1]];
    ];
   Print[a, " ",b];
   c[[a + 1 + 1, b + 1 + 1, 2]] = c[[2, 2, 2]]*res;
   ctilde[[a + 1, b + 1, 2]] = res;
   ];


FillMultiple[c_, a0_, b0_, max_] := 
  Module[{a, b, a1, b1, x, y, n, k, A, cCop},
   A = c[[a0 + 1, b0 + 1, 2]];
   c[[a0 + 1, b0 + 1, 2]] = 0;
   cCop = c;
   For[n = 2, n < M, n++,
    For[k = 0, k < n, k++,
     For[a = a0*(n - k), a < M, a++,
      For[b = b0*(n - k), b < M, b++,
        cCop[[a + 1, b + 1, n + 1]] += 
          Binomial[n, k]*A^(n - k)*
           c[[a - a0*(n - k) + 1, b - b0*(n - k) + 1, k + 1]];
        ];
      ]
     ]
    ];
   cCop[[a0 + 1, b0 + 1, 2]] = A;
   c = cCop;
   ];

CompleteTable[c_, ctilde_] := Module[{d, b},
   For[d = 0, d < M - 1, d++,
    For[b = 0, b < d + 1, b++,
     If[d == 0 && b == 0,
      c[[1, 1, 1]] = 1;
      c[[1, 1, 2]] = 0;
      ctilde[[1, 1, 1]] = 1;
      ctilde[[1, 1, 2]] = 0;
      FillMultiple[c, 0, 0, 0 + 2];
      FillMultiple[ctilde, 0, 0, 0 + 2];
      c[[2, 2, 2]] = -2 I;
      FillMultiple[c, 1, 1, d + 5];
      ,
      ComputeNext[b, d, c, ctilde];
      FillMultiple[c, b + 1, d + 1, d + 3];
      FillMultiple[ctilde, b, d, d + 3];
      ]
     ]
    ]
   ];

EnergyTable[c_, cstar_] := Module[{F, a, b, n},
   F = ConstantArray[0, M-1];
   F[[1]] = X[-1]/c[[2, 2, 2]]/I;
   For[a = 0, a < M - 1, a++,
    For[n = 0, n < M - 1, n++,
     For[b = 0, b < M - 2, b++, 
       F[[b + 1 + 1]] += 
         1 /c[[2, 2, 2]]/I (-1)^n*
           cstar[[a + 1, b + 1 + 1, n + 1]] X[a - 1] - 
          2 Pi (n + 1) S[[n + 1 + 1]] c[[a + 1, b + 1, n + 1]] X[a];
       ];
     ]
    ];
   Return[F];
   ];
CompleteTable[CC, CCtilde]
F= EnergyTable[CC, CCtilde]
\end{verbatim}
\normalsize
\end{multicols}

\newpage

\section{Series expansion for the  energy at $m=-1$} \label{seriesmath}
Equation \eqref{genderiv} permits one to compute iteratively all the $\gamma_p(t,x)$ defined in \eqref{gamap}, and then deduce all the derivatives of the  energy $F(m,\varphi)$ at $m=-1$ with \eqref{ffderiv}. However, directly solving iteratively for $\gamma_p(t,x)$ with a computer is costly since it requires symbolic manipulation. One actually sees that only the evaluation of the $t$-derivatives of $\gamma_p(t,x)$ at $t=0$ and $t=-1$ are needed, and $x$ is only a `spectator' variable. Hence if we define $c_{ab}(x)$ and $d_{ab}(x)$ by
\begin{equation}
\begin{aligned}
\gamma_p(t,x)&=\sum_{a\geq 0} c_{ap}(x)t^a\\
&=\sum_{a\geq 0} d_{ap}(x)(t+1)^a\,,
\end{aligned}
\end{equation}
we can turn \eqref{genderiv} into nested recurrence relations for $c_{ab}(x)$ and $d_{ab}(x)$, that require only manipulating numbers. We will consider $x$ fixed and drop the explicit dependence to lighten the notations. Following the steps of section \ref{exptwist}, we obtain
\begin{equation}
\begin{aligned}
\tilde{c}_{ab}&=2i\Bigg[\sum_{n\geq 0}\frac{\arctan^{(n)}(0)}{n!}c_{ab}^{[n]}-\frac{1}{2i}\sum_{n\geq 0} \frac{(-1)^n c_{ab}^{[n]}}{n(2i)^n}+\frac{1}{2i}\sum_{n\geq 2}(-1)^n\frac{\tilde{c}_{ab}^{[n]}}{n}\\
& -\sum_{\substack{p\geq 0\\ b_1+b_2+2p+1=b\\ a_1\geq 0\\ 0\leq q\leq n}}\frac{\pi^{2p}}{(2p+1)!}(-1)^p\kappa_{2p}^{a_1}\frac{\arctan^{(n)}(-\Delta(-x))}{q!(n-q)!}(-1)^q d_{a_1b_1}^{[q]}c_{ab_2}^{[n-q]}\Bigg]\,,
\end{aligned}
\end{equation}
and
\begin{equation}
\begin{aligned}
&d_{ab}=\\
&\frac{1}{\arctan'(i+\Delta(-x))-\arctan'(\Delta(-x))}\Bigg[\sum_{n\geq 2}\frac{\arctan^{(n)}(\Delta(-x))-\arctan^{(n)}(i+\Delta(-x))}{n!}c_{ab}^{[n]}\\
&-\frac{1}{2i}\frac{\delta_{b,0}}{a} -\sum_{\substack{p\geq 0\\ b_1+b_2+2p+1=b\\ a_1\geq 0\\ 0\leq q\leq n}}\frac{\pi^{2p}}{(2p+1)!}(-1)^p\kappa_{2p}^{a_1}\frac{\arctan^{(n)}(0)}{q!(n-q)!}(-1)^q d_{a_1b_1}^{[q]}d_{ab_2}^{[n-q]}\Bigg] \,,
\end{aligned}
\end{equation}
where we defined
\begin{equation}
\tilde{c}_{ab}=\begin{cases}
\frac{c_{a+1,b}}{c_{10}}=\frac{c_{a+1,b}}{x \Delta'(0)}\quad\text{if }(a,b)\neq(0,0)\\
0\quad\text{if }a=b=0\,,
\end{cases}
\end{equation}
and
\begin{equation}
\kappa_i^j=\slashed{\partial}^i (t+1)^j|_{t=-1}=\sum_{k=0}^j {j\choose k}k^i (-1)^k \,.
\end{equation}
Then we obtain
\begin{equation}
F(\xi)=\sum_{p\geq 0}f_p \xi^p\,,
\end{equation}
with $f_0=1$, and for $p\geq 1$
\begin{equation}
\begin{aligned}
f_p=\frac{i}{c_{10}}\sum_{b,n\geq 0}\tilde{c}^{[n]}_{1b}(-1)^n+ 2\!\!\! \sum_{\substack{a,m,n,b\geq 0\\b+2m+1=p}}\frac{\pi^{2m}}{(2m+1)!}(-1)^m \frac{\arctan^{(n)}(i+\Delta(-x))}{n!}\kappa_{2m}^a d_{ab}^{[n]} \,.
\end{aligned}
\end{equation}

This expansion can be computed with the following {\sc Mathematica} code:
\begin{multicols}{2}
\tiny\begin{verbatim}
ClearAll["Global`*"];
M = 7;
S = ConstantArray[0, M];
SL = ConstantArray[0, M];
R = ConstantArray[0, M];
Ggamma[x_] := -I/2 + I/2 Sqrt[1 - 8 x/(1 - x)];
S = Join[{0}, 
   Series[-1/Pi ArcTan[x], {x, 0, M}][[3]]*Factorial[Range[M]]];
Clear[xx];
Sm1 = Join[{}, 
   Series[-1/Pi ArcTan[(x + I + Ggamma[-xx])], {x, 0, M}, 
      Assumptions -> {xx <= 1, xx > 0}][[3]]*Factorial[Range[0, M]]];
      R = Join[{0}, 
   Series[1/Pi ArcTan[x], {x, 0, M}][[3]]*Factorial[Range[M]]];
Rm1 = Join[{}, 
   Series[1/Pi ArcTan[(x + Ggamma[-xx])], {x, 0, M}, 
      Assumptions -> {xx <= 1, xx > 0}][[3]]*Factorial[Range[0, M]]];
sigma = 1/(2 Pi);
SL[[1]] = -I/2/Pi Log[-2 I];
For[i = 2, i < M + 1, i++,
 SL[[i]] = I/(2 Pi)*(-I)^(i - 1)/(i - 1);
 SL[[i]] = I/(2 Pi)*(I/2)^(i - 1)/(i - 1);
 ];
CC = ConstantArray[0, {M, M, M}];
DD = ConstantArray[0, {M, M, M}];
DDt = ConstantArray[0, {M, M, M}];
kappa = ConstantArray[0, {M, M}];
For[i = 1, i < M + 1, i++,
 For[j = 1, j < M + 1, j++,
  kappa[[i, j]] = 
    If[i > 1, Sum[Binomial[j - 1, k] k^(i - 1) (-1)^k, {k, 0, j - 1}],
      Sum[Binomial[j - 1, k] (-1)^k, {k, 0, j - 1}]];
  ]
 ]

xx; (*Value of x*)

SetAttributes[ComputeNext, HoldAll]
SetAttributes[ComputeNextDD, HoldAll]
SetAttributes[FillMultiple, HoldAll]
SetAttributes[CompleteTable, HoldAll]
SetAttributes[CompleteTableDD, HoldAll]

ComputeNext[a_, b_, c_] := Module[{res, a1, b2, q, n, res2, Res, mm},
   res = 0;
   For[a1 = 0, a1 < b, a1++,
    For[b2 = 0, b2 < b, b2++,
     For[mm = 0, mm < b/2 + 2, mm++,
      If[(b - b2 - 1 - 2 mm >= 0) && (2 mm + 1 < M + 1),
       For[n = 0, n < M, n++,
        For[q = 0, q < n + 1, q++,
         res += (-1)^(q + mm)/Factorial[q]/Factorial[n - q]*
            R[[n + 1]] Pi^(2 mm)/Factorial[2 mm + 1] kappa[[2 mm + 1, 
             a1 + 1]]*c[[a1 + 1, b - b2 - 1 - 2 mm + 1, q + 1]]*
            c[[a + 1, b2 + 1, n - q + 1]];
         ]
        ]
       ]
      ]
     ]
    ];
   For[n = 2, n < M, n++,
    res += -c[[a + 1, b + 1, n + 1]]*(Sm1[[n + 1]] + Rm1[[n + 1]]) /
        Factorial[n];
    ];
   If[b == 0, res += -I/2/Pi/a ];
   res /= (Sm1[[2]] + Rm1[[2]] );
   Print[a," ",b];
   c[[a + 1, b + 1, 2]] = res;
   ];

ComputeNextDD[a_, b_, c_, d_, dt_] := 
  Module[{res, a1, b2, q, n, res2, Res, mm},
   res = 0;
   For[a1 = 0, a1 < b, a1++,
    For[b2 = 0, b2 < b, b2++,
     For[mm = 0, mm < b/2 + 2, mm++,
      If[(b - b2 - 1 - 2 mm >= 0) && (2 mm + 1 < M + 1),
       For[n = 0, n < M, n++,
        For[q = Max[0, n - a - b2], 
         q < Min[a1 + b - b2 - 2 mm - 1 + 1, n + 1], q++,
         res += -(-1)^(n - q + mm)/Factorial[q]/Factorial[n - q]*
            Rm1[[n + 1]] Pi^(2 mm)/Factorial[2 mm + 1] kappa[[
             2 mm + 1, a1 + 1]]*
            c[[a1 + 1, b - b2 - 1 - 2 mm + 1, q + 1]]*
            d[[a + 1, b2 + 1, n - q + 1]];
         ]
        ]
       ]
      ]
     ]
    ];
   For[n = 1, n < M, n++,
    res += -d[[a + 1, b + 1, n + 1]]*(R[[n + 1]] /Factorial[n] + 
         SL[[n + 1]]) ;
    ];
   For[n = 2, n < M, n++,
    res += (-1)^n dt[[a + 1, b + 1, n + 1]]*sigma*I/n;
    ];
   res *= xx Ggamma'[0]/I/sigma;
   Print[a," ",b];
   d[[a + 1 + 1, b + 1, 2]] = res;
   dt[[a + 1, b + 1, 2]] = res/Ggamma'[0]/xx;
   ];



FillMultiple[c_, a0_, b0_, max_] := 
  Module[{a, b, a1, b1, x, y, n, k, A, cCop},
   A = c[[a0 + 1, b0 + 1, 2]];
   c[[a0 + 1, b0 + 1, 2]] = 0;
   cCop = c;
   For[n = 2, n < M, n++,
    For[k = 0, k < n, k++,
     For[a = a0*(n - k), a < M, a++,
      For[b = b0*(n - k), b < M, b++,
        cCop[[a + 1, b + 1, n + 1]] += 
          Binomial[n, k]*A^(n - k)*
           c[[a - a0*(n - k) + 1, b - b0*(n - k) + 1, k + 1]];
        ];
      ]
     ]
    ];
   cCop[[a0 + 1, b0 + 1, 2]] = A;
   c = cCop;
   ];


CompleteTable[c_] := Module[{d, b},
   For[d = 0, d < M, d++,
    For[b = 0, b < d + 1, b++,
     If[d == 0 && b == 0,
      c[[1, 1, 1]] = 1;
      c[[1, 1, 2]] = 0;
      FillMultiple[c, 0, 0, 0 + 2],
      If[d == 1 && b == 0,
       c[[2, 1, 2]] = Ggamma'[-xx]*xx;
       FillMultiple[c, 1, 0, d],
       ComputeNext[d - b, b, c];
       FillMultiple[c, d - b, b, d];
       ]
      ]
     ]
    ]
   ];


CompleteTableDD[d_, dt_] := Module[{dy, b},
   For[dy = 0, dy < M - 1, dy++,
    For[b = 0, b < dy + 1, b++,
     If[dy == 0 && b == 0,
      d[[1, 1, 1]] = 1;
      d[[1, 1, 2]] = 0;
      FillMultiple[d, 0, 0, 0 + 2];
      dt[[1, 1, 1]] = 1;
      dt[[1, 1, 2]] = 0;
      FillMultiple[dt, 0, 0, 0 + 2];
      d[[2, 1, 2]] = Ggamma'[0] xx;
      FillMultiple[d, 1, 0, dy],
      ComputeNextDD[dy - b, b, CC, d, dt];
      FillMultiple[d, dy - b + 1, b, dy + 2];
      FillMultiple[dt, dy - b, b, dy + 2];
      ]
     ]
    ]
   ];

FreeEnergyTable[c_, dt_] := Module[{F, a, b, n, mm},
   F = ConstantArray[0, M-2];
   For[a = 0, a < M-2, a++,
    For[n = 0, n < M - 1, n++,
     For[b = 0, b < a, b++,
      For[mm = 0, mm < M/2 + 1, mm++,
       If[a - 2 mm - 1 >= 0,
        F[[a + 1]] += -2 Pi Sm1[[n + 1 + 1]]/Factorial[n]*
           
           c[[b + 1, a - 2 mm - 1 + 1, n + 1]] kappa[[2 mm + 1, 
            b + 1]] Pi^(2 mm)/Factorial[2 mm + 1] (-1)^mm;
        ]
       ]
      ]
     ];
    ];
   For[a = 0, a < M-2, a++,
    For[n = 0, n < M, n++,
      F[[a + 1]] += 
        2 I Pi sigma/Ggamma'[0]/xx dt[[1 + 1, a + 1, n + 1]] (-1)^n;
      ];
    ];
   F[[1]] += 2 Pi SL[[2]];
   Return[F];
   ];
CompleteTable[CC];
CompleteTableDD[DD, DDt];
F = FreeEnergyTable[CC, DDt]
\end{verbatim}
\normalsize
\end{multicols}

\section{Series expansion for the trajectory to the ground state\label{recurtraj}}
In section \ref{trajgr} we presented a trajectory to the ground state, with values of $\Xi^p$ given in \eqref{xitptraj}. Hence only the $t$-derivatives of $\gamma_p(t,x)$ at $t=0,-1,i,-i$ are needed to compute the $\xi$-derivatives of the free energy. Hence if we define $c_{ab}(x)$, $d_{ab}(x)$ and $e^\pm_{ab}(x)$ by
\begin{equation}
\begin{aligned}
\gamma_p(t,x)&=\sum_{a\geq 0} c_{ap}t^a(x)\\
&=\sum_{a\geq 0} d_{ap}(t+1)^a(x)\\
&=\sum_{a\geq 0} e^+_{ap}(t-i)^a(x)\\
&=\sum_{a\geq 0} e^-_{ap}(t+i)^a(x)\,,
\end{aligned}
\end{equation}
we can obtain nested recurrence relations on $c_{ab},d_{ab},e^\pm_{ab}$, similarly as in Appendix \ref{seriesmath}. We obtain then the following {\sc Mathematica} code:
\begin{multicols}{2}
\tiny
\begin{verbatim}
ClearAll["Global`*"];
M = 7;
S = ConstantArray[0, M];
R = ConstantArray[0, M];
Ggamma[x_] := -I/2 + I/2 Sqrt[1 - 8 x/(1 - x)];
sfunction[x_] := -1/Pi ArcTan[x];
S = Join[{0}, Series[sfunction[x], {x, 0, M}][[3]]];
Sm1 = Series[sfunction[(x + I + Ggamma[-xx])], {x, 0, M}, 
      Assumptions -> {xx <= 1, xx > 0}][[3]];
Spi = Series[sfunction[(x + I + Ggamma[I xx])], {x, 0, M}, 
      Assumptions -> {xx <= 1, xx > 0}][[3]];
Smi = Series[sfunction[(x + I + Ggamma[-I xx])], {x, 0, M}, 
      Assumptions -> {xx <= 1, xx > 0}][[3]];

rfunction[x_] := 1/Pi ArcTan[x ];
R = Join[{0}, Series[rfunction[x], {x, 0, M}][[3]]];
Rm1 = Series[rfunction[(x + Ggamma[-xx])], {x, 0, M}, 
      Assumptions -> {xx <= 1, xx > 0}][[3]];
Rpi = Series[rfunction[(x + Ggamma[I xx])], {x, 0, M}, 
      Assumptions -> {xx <= 1, xx > 0}][[3]];
Rmi = Series[rfunction[(x + Ggamma[-I xx])], {x, 0, M}, 
      Assumptions -> {xx <= 1, xx > 0}][[3]];
Rm1pi =Series[rfunction[(x + Ggamma[-xx] - Ggamma[I xx])], {x, 0, M}, 
      Assumptions -> {xx <= 1, xx > 0}][[3]];
Rm1mi =Series[rfunction[(x + Ggamma[-xx] - Ggamma[-I xx])], {x, 0, M}, 
      Assumptions -> {xx <= 1, xx > 0}][[3]];
Rpimi = Series[rfunction[(x + Ggamma[I xx] - Ggamma[-I xx])], {x, 0, M}, 
      Assumptions -> {xx <= 1, xx > 0}][[3]];
sigma = 1/(2 Pi);
Sl = Range[M];
Sl[[1]] = -I/2/Pi Log[-2 I];
For[i = 2, i < M + 1, i++,
  Sl[[i]] = I/(2 Pi)*(-I)^(i - 1)/(i - 1);
  Sl[[i]] = I/(2 Pi)*(I/2)^(i - 1)/(i - 1);
  ];
CC = ConstantArray[0, {M, M, M}];
EE = ConstantArray[0, {M, M, M}];
FF = ConstantArray[0, {M, M, M}];
DD = ConstantArray[0, {M, M, M}];
DDt = ConstantArray[0, {M, M, M}];
Kappa = ConstantArray[0, {M, M}];
kappa[i_, j_] := 
  If[i > 0, Sum[Binomial[j, k] k^(i) (-1)^k, {k, 0, j}], 
   Sum[Binomial[j, k] (-1)^k, {k, 0, j}]];
For[i = 0, i < M, i++,
 For[j = 0, j < M, j++,
  Kappa[[i + 1, j + 1]] = kappa[i, j];
  ]
 ]

SetAttributes[ComputeNext, HoldAll];
SetAttributes[ComputeNextDD, HoldAll];
SetAttributes[FillMultiple, HoldAll];
SetAttributes[CompleteTable, HoldAll];
SetAttributes[CompleteTableDD, HoldAll];
ComputeNext[a_, b_, c_, d_, e_] := 
  Module[{res, a1, b2, q, n, res1, res2, mm},
   res = 0;
   res1 = 0;
   res2 = 0;
   For[a1 = 0, a1 < b, a1++,
    For[b2 = 0, b2 < b, b2++,
     For[mm = a1 + 1, mm < b + 1, mm++,
      If[(b - b2 - mm >= 0),
       For[n = 0, n < M, n++,
        For[q = Max[0, n - a - b2], 
         q < Min[a1 + b - b2 - mm + 1, n + 1], q++,
         res += (-1)^(q)*Binomial[n,q]*(I Pi)^(mm - 1)/2/Factorial[mm] 
         Kappa[[mm - 1 + 1,a1 + 1]]*(R[[n + 1]] c[[a1 + 1, b - b2 - mm + 1, 
               q + 1]] (1 - (-1)^mm) + (-I)^a1 Rm1pi[[n + 1]] d[[
               a1 + 1, b - b2 - mm + 1, q + 1]] ((-1)^mm) - (I)^
               a1 Rm1mi[[n + 1]] e[[a1 + 1, b - b2 - mm + 1, 
               q + 1]] (1))*c[[a + 1, b2 + 1, n - q + 1]];
         res1 += (-1)^(q)*Binomial[n,q]*(I Pi)^(mm - 1)/2/Factorial[mm] 
         Kappa[[mm - 1 + 1,a1 + 1]]*((-1)^(n + 1) Rm1pi[[n + 1]] c[[a1 + 1, 
               b - b2 - mm + 1, q + 1]] (1 - (-1)^mm) + (-I)^a1 R[[
               n + 1]] d[[a1 + 1, b - b2 - mm + 1, 
               q + 1]] ((-1)^mm) - (I)^a1 Rpimi[[n + 1]] e[[a1 + 1, 
               b - b2 - mm + 1, q + 1]] (1))*
           d[[a + 1, b2 + 1, n - q + 1]];
         res2 += (-1)^(q)*Binomial[n,q]*(I Pi)^(mm - 1)/2/Factorial[mm] 
         Kappa[[mm - 1 + 1, a1 + 1]]*((-1)^(n + 1) Rm1mi[[n + 1]] c[[a1 + 1, 
               b - b2 - mm + 1, q + 1]] (1 - (-1)^mm) + (-I)^
               a1 (-1)^(n + 1) Rpimi[[n + 1]] d[[a1 + 1, 
               b - b2 - mm + 1, q + 1]] ((-1)^mm) - (I)^a1 R[[
               n + 1]] e[[a1 + 1, b - b2 - mm + 1, q + 1]] (1))*
           e[[a + 1, b2 + 1, n - q + 1]];
         ]
        ]
       ]
      ]
     ]
    ];
   For[n = 2, n < M, n++,
    res += -c[[a + 1, b + 1, n + 1]]*(Sm1[[n + 1]] + Rm1[[n + 1]]);
    res1 += -d[[a + 1, b + 1, n + 1]]*(Spi[[n + 1]] + Rpi[[n + 1]]);
    res2 += -e[[a + 1, b + 1, n + 1]]*(Smi[[n + 1]] + Rmi[[n + 1]]);
    ];
   If[b == 0,
    res += - I/2/Pi/a ;
    res1 += - I/2/Pi/a (-1/I)^a;
    res2 += - I/2/Pi/a (1/I)^a;
    ];
   res /= (Sm1[[2]] + Rm1[[2]] );
   res1 /= (Spi[[2]] + Rpi[[2]] );
   res2 /= (Smi[[2]] + Rmi[[2]] );
   Print[a," ",b];
   
   c[[a + 1, b + 1, 2]] = res;
   d[[a + 1, b + 1, 2]] = res1;
   e[[a + 1, b + 1, 2]] = res2;
   ];

ComputeNextDD[a_, b_, c_, d_, dt_, e_, f_] := 
  Module[{res, a1, b2, q, n, res2, Res, mm},
   res = 0;
   For[a1 = 0, a1 < b, a1++,
    For[b2 = 0, b2 < b, b2++,
     For[mm = 1, mm < b + 1, mm++,
      If[(b - b2 - mm >= 0),
       For[n = 0, n < M, n++,
        For[q = Max[0, n - a - b2], 
         q < Min[a1 + b - b2 - mm + 1, n + 1], q++,
         res += -(-1)^(n - q)*Binomial[n,q]*(I Pi)^(mm - 1)/2/Factorial[mm] 
         Kappa[[mm - 1 + 1, a1 + 1]]*(Rm1[[n + 1]] c[[a1 + 1, b - b2 - mm + 1, 
                q + 1]] (1 - (-1)^mm) + (-I)^a1 Rpi[[n + 1]] e[[
                a1 + 1, b - b2 - mm + 1, q + 1]] ((-1)^mm) - (I)^
                a1 Rmi[[n + 1]] f[[a1 + 1, b - b2 - mm + 1, 
                q + 1]] (1))*d[[a + 1, b2 + 1, n - q + 1]];
         ]
        ]
       ]
      ]
     ]
    ];
   For[n = 1, n < M, n++,
    res += -d[[a + 1, b + 1, n + 1]]*(R[[n + 1]]  + 
         Sl[[n + 1]]) ;
    ];
   For[n = 2, n < M, n++,
    res += (-1)^n dt[[a + 1, b + 1, n + 1]]*sigma*I/n;
    ];
   res *= xx Ggamma'[0]/I/sigma;
   Print[a," ",b];
   d[[a + 1 + 1, b + 1, 2]] = res;
   dt[[a + 1, b + 1, 2]] = res/Ggamma'[0]/xx;
   ];



FillMultiple[c_, a0_, b0_, max_] := 
  Module[{a, b, a1, b1, x, y, n, k, A, cCop},
   A = c[[a0 + 1, b0 + 1, 2]];
   c[[a0 + 1, b0 + 1, 2]] = 0;
   cCop = c;
   For[n = 2, n < M, n++,
    For[k = 0, k < n, k++,
     For[a = a0*(n - k), a < M, a++,
      For[b = b0*(n - k), b < M, b++,
        cCop[[a + 1, b + 1, n + 1]] += 
          Binomial[n, k]*A^(n - k)*
           c[[a - a0*(n - k) + 1, b - b0*(n - k) + 1, k + 1]];
        ];
      ]
     ]
    ];
   cCop[[a0 + 1, b0 + 1, 2]] = A;
   c = cCop;
   ];


CompleteTable[c_, e_, f_] := Module[{d, b},
   For[d = 0, d < M, d++,
    For[b = 0, b < d + 1, b++,
     If[d == 0 && b == 0,
      c[[1, 1, 1]] = 1;
      c[[1, 1, 2]] = 0;
      e[[1, 1, 1]] = 1;
      e[[1, 1, 2]] = 0;
      f[[1, 1, 1]] = 1;
      f[[1, 1, 2]] = 0;
      FillMultiple[c, 0, 0, 0 + 2];
      FillMultiple[e, 0, 0, 0 + 2];
      FillMultiple[f, 0, 0, 0 + 2];
      ,
      If[d == 1 && b == 0,
       c[[2, 1, 2]] = Ggamma'[-xx]*xx*1;
       e[[2, 1, 2]] = Ggamma'[I xx]*xx*1;
       f[[2, 1, 2]] = Ggamma'[-I xx]*xx*1;
       FillMultiple[c, 1, 0, M];(*d*)
       
       FillMultiple[e, 1, 0, M];
       FillMultiple[f, 1, 0, M];
       ,
       ComputeNext[d - b, b, c, e, f];
       FillMultiple[c, d - b, b, d];
       FillMultiple[e, d - b, b, d];
       FillMultiple[f, d - b, b, d];
       ]
      ]
     ]
    ]
   ];


CompleteTableDD[d_, dt_, cc_, e_, f_] := Module[{dy, b},
   For[dy = 0, dy < M - 1, dy++,
    For[b = 0, b < dy + 1, b++,
     If[dy == 0 && b == 0,
      d[[1, 1, 1]] = 1;
      d[[1, 1, 2]] = 0;
      FillMultiple[d, 0, 0, 0 + 2];
      dt[[1, 1, 1]] = 1;
      dt[[1, 1, 2]] = 0;
      FillMultiple[dt, 0, 0, 0 + 2];
      d[[2, 1, 2]] = Ggamma'[0] xx;
      FillMultiple[d, 1, 0, dy],
      ComputeNextDD[dy - b, b, cc, d, dt, e, f];
      FillMultiple[d, dy - b + 1, b, dy + 2];
      
      FillMultiple[dt, dy - b, b, dy + 2];
      ]
     ]
    ]
   ];

FreeEnergyTable[c_, dt_, e_, f_] := Module[{F, a, b, n, mm},
   F = ConstantArray[0, M-2];
   Print[F];
   For[a = 0, a < M-2, a++,
    For[n = 0, n < M - 1, n++,
     For[b = 0, b < a, b++,
      For[mm = 1, mm < M, mm++,
       If[a - mm >= 0 && (mm >= b + 1),
        F[[
           a + 1]] += -2 Pi* (Sm1[[n + 1 + 1]] c[[b + 1, a - mm + 1, 
               n + 1]] (1 - (-1)^mm) + (-I)^b Spi[[n + 1 + 1]] e[[
               b + 1, a - mm + 1, n + 1]] ((-1)^mm) - (I)^b Smi[[
               n + 1 + 1]] f[[b + 1, a - mm + 1, n + 1]] (1)) Kappa[[
            mm - 1 + 1, b + 1]] (I Pi)^(mm - 1)/2/Factorial[mm];
        ]
       ]
      ]
     ];
    Print[a];
    ];
   For[a = 0, a < M-2, a++,
    For[n = 0, n < M, n++,
      F[[a + 1]] += 
        2 I Pi sigma/Ggamma'[0]/xx dt[[1 + 1, a + 1, n + 1]] (-1)^n;
      ];
    ];
   F[[1]] += 2 Pi Sl[[2]];
   Return[F];
   ];
   CompleteTable[CC, EE, FF];
CompleteTableDD[DD, DDt, CC, EE, FF];
F = FreeEnergyTable[CC, DDt, EE, FF]
\end{verbatim}
\normalsize
\end{multicols}

\bibliography{../bibliographie}

\bibliographystyle{ieeetr}

\end{document}